\documentclass[authoryear,preprint,review,12pt]{elsarticle}

\usepackage{setspace}

\usepackage[margin=1in]{geometry}
\usepackage{bbm}
\usepackage{float}
\usepackage{graphicx}
\usepackage{natbib}
\usepackage[dvipsnames]{xcolor}
\usepackage{pgfplots}
\pgfplotsset{compat=1.18}
\usetikzlibrary{arrows.meta}
\usetikzlibrary{patterns}
\definecolor{darkgreen}{rgb}{0.0, 0.5, 0.0}
\usepackage[colorlinks=true, linkcolor=blue, citecolor=green, urlcolor=red]{hyperref}

\newcounter{warnings}

\usepackage{amssymb, amsthm, amsfonts}
\usepackage[cmex10]{amsmath}
\usepackage{mathtools}
\usepackage{subcaption}
\usepackage{verbatim}
\usepackage[shortlabels]{enumitem}

\renewcommand\H{h}
\renewcommand\L{l}
\newcommand\lp{\L\PE}
\newcommand\lo{\L\OP}
\newcommand\hp{\H\PE}
\newcommand\ho{\H\OP}
\newcommand\st{\text{s.t.}}

\newtheorem{assumption}{Assumption}

\newcounter{subassumption}[assumption]
\renewcommand{\thesubassumption}{\theassumption.\alph{subassumption}}

\newcommand{\subassumption}[1]{
    \refstepcounter{subassumption}
    \thesubassumption\ \label{#1}
}

\usepackage{tcolorbox} 

\newtheorem{theorem}{Theorem}
\newtheorem*{theorem*}{Theorem}
\newtheorem{proposition}[theorem]{Proposition}
\newtheorem*{proposition*}{Proposition}

\newtheorem*{lemma*}{Lemma}
\newtheorem{corollary}[theorem]{Corollary}
\newtheorem*{corollary*}{Corollary}

\newtheorem*{claim*}{Proposition}
\theoremstyle{definition}

\newcommand\Util{U}
\newcommand\ELoss{J}

\newcommand\Welf{L}
\newcommand\V{v}
\newcommand\F{f}
\newcommand\OP{2}
\newcommand\PE{1}
\newcommand\Prof{P}

\DeclareMathOperator{\argmin}{argmin}
\usepackage{booktabs}

\journal{Energy Policy}
\usepackage{fancyhdr}
\begin{document}
\usetikzlibrary{arrows.meta, positioning, calc}
\begin{frontmatter}



\title{When Do Consumers Lose from Variable Electricity Pricing?}


\author[inst1]{Nathan Engelman Lado\corref{cor1}}
\cortext[cor1]{Corresponding author}
\ead{engelm25@mit.edu}

\affiliation[inst1]{organization={Institute for Data Systems and Society, MIT},
            addressline = {77 Massachusetts Ave}, 
            city={Cambridge},
            postcode={02139}, 
            state={MA},
            country={USA}}

\author[inst2]{Richard Chen}
\ead{rachen@mit.edu}
\author[inst3]{Saurabh Amin}
\ead{amins@mit.edu}
\affiliation[inst2]{organization={Civil and Environmental Engineering, MIT},
            addressline = {77 Massachusetts Ave}, 
            city={Cambridge},
            postcode={02139}, 
            state={MA},
            country={USA}}
\affiliation[inst3]{organization={Laboratory for Information and Decision Systems, MIT},
            addressline = {77 Massachusetts Ave}, 
            city={Cambridge},
            postcode={02139}, 
            state={MA},
            country={USA}}
            

\begin{abstract}
Time-varying electricity pricing better reflects the varying cost of electricity compared to flat-rate pricing. Variations between peak and off-peak costs are increasing due to weather variation, renewable intermittency, and increasing electrification of demand. Empirical and theoretical studies suggest that variable pricing can lower electricity supply costs and reduce grid stress. However, the distributional impacts, particularly on low-income consumers, remain understudied. This paper develops a theoretical framework to analyze how consumer heterogeneity affects welfare outcomes when electricity markets transition from flat-rate to time-varying pricing, considering realistic assumptions about heterogeneous consumer demand, supply costs, and utility losses from unmet consumption.

We derive sufficient conditions for identifying when consumers lose utility from pricing reforms and compare welfare effects across consumer types. Our findings reveal that consumer vulnerability depends on the interaction of consumption timing, demand flexibility capabilities, and price sensitivity levels. Consumers with high peak-period consumption and inflexible demand, characteristics often associated with low-income households, are most vulnerable to welfare losses. Critically, we demonstrate that demand flexibility provides welfare protection only when coincident with large price changes, a condition more easily met by higher-income consumers in most empirical studies. Our equilibrium analysis reveals that aggregate flexibility patterns generate spillover effects through pricing mechanisms, with peak periods experiencing greater price changes when they have less aggregate flexibility, potentially concentrating larger price increases among vulnerable populations that have a limited ability to respond. These findings suggest that variable pricing policies designed to support clean energy transitions risk exacerbating energy inequality unless accompanied by targeted policies ensuring equitable access to demand response capabilities and pricing benefits.
\end{abstract}

\begin{graphicalabstract}

\begin{figure}[H]
    \centering
    \begin{tikzpicture}[scale=1.4]
        
        \draw[thick, blue, fill=blue!10] (-6, 3.5) rectangle (-1.1, 5.8);
        \node at (-3.75, 5.4) {\textbf{Vulnerability Factors}};
        \node[align=left] at (-3.75, 4.3) {
            • High peak consumption\\[0.25cm]
            • Limited flexibility\\[0.25cm]
            • High price sensitivity
        };
        
        \draw[thick, orange, fill=orange!10] (1.1, 3.5) rectangle (6, 5.8);
        \node at (3.75, 5.4) {\textbf{Key Insights}};
        \node[align=left] at (3.5, 4.5) {
            • Flexibility has the greatest impact \\in periods of large price shifts\\
            • Peak period inflexibility \\increases peak price shift
        };
        
        \draw[thick, green, fill=green!10] (-6, 0.5) rectangle (-1.1, 2.8);
        \node at (-3.75, 2.4) {\textbf{Theory Contribution}};
        \node[align=left] at (-3.75, 1.6) {
            Sufficient conditions identifying\\when consumers lose from\\pricing transitions.
        };
        
        \draw[thick, red, fill=red!10] (1.1, 0.5) rectangle (6, 2.8);
        \node at (3.75, 2.4) {\textbf{Distributional Pattern}};
        \node[align=left] at (3.75, 1.4) {
            Benefits concentrate among\\flexible consumers with\\less peak-dominated usage and \\high peak period flexibility.
        };
        
        \draw[->, very thick] (-.9, 4.65) -- (.9, 4.65);
        
        \draw[->, very thick] (-3.75, 3.3) -- (-3.75, 3.0);
        \draw[->, very thick] (3.75, 3.3) -- (3.75, 3.0);
        
        \node[align=center] at (0, -0.3) {
            \textbf{Contribution:} Theoretical framework identifying conditions\\
            for consumer welfare losses from the transition to variable pricing
        };
        
    \end{tikzpicture}
    \caption{Theoretical framework for analyzing consumer vulnerability to variable electricity pricing transitions.}
    \label{fig:graphical_abstract}
\end{figure}

\end{graphicalabstract}

\begin{highlights}
\item We derive sufficient conditions identifying when consumers lose utility from variable electricity pricing transitions based on consumption patterns, demand flexibility, and price sensitivity.
\item Demand flexibility provides welfare protection only when coincident with large price changes, explaining why higher-income consumers typically benefit more from time-varying rates.
\item Low-income consumers face amplified welfare effects from pricing reforms due to higher price sensitivity, experience larger utility changes even when their consumption and flexibility patterns are identical.
\item Our results suggest variable pricing policies require complementary measures to prevent disproportionate harm to vulnerable consumers.
\end{highlights}


\end{frontmatter}


\section{Introduction}
The high demand for electricity during peak hours drives up supply costs, making reliable production and delivery more expensive. Flat rates obscure these temporal costs, pushing demand up during peaks and down during off-peak hours \citep{borenstein2022two}. 
Time-varying electricity pricing better aligns consumption with supply costs by providing price signals that can lower system costs and increase reliability. Yet variable pricing also creates distributional concerns. 
Consumers with flexible demand or new technologies can adapt, but low-income households often face higher bills and welfare losses from foregone consumption. We develop a theoretical framework to identify which consumers gain or lose from this transition, and under what conditions low-income households are disproportionately harmed.

In practice, most regions experience pronounced peaks during evening hours and lower demand overnight and mid-day \citep{eia2020electricity}.
Climate change intensifies these fluctuations by increasing cooling demand during hot afternoons, while renewable integration creates additional variability as solar generation peaks mid-day but drops during evening demand peaks. Electrification of transportation and heating further concentrates demand during specific hours when people return home and charge vehicles. Together, these trends widen the gap between peak and off-peak supply costs, as grid operators must maintain expensive backup generation and transmission capacity to meet increasingly sharp demand spikes \citep{Auffhammer_Baylis_Hausman_2017}.  
Grid operators can respond by adding storage or backup generation, expanding transmission, or implementing demand-side measures that shift consumption to lower cost periods \citep{schittekatte2023reforming}. 

Even when flexibility is technically possible, this might come at a significant cost. Low-income households often keep electricity use at subsistence levels for health and safety needs such as heating, cooling, and lighting \citep{Anderson_White_Finney_2012, Cong_Nock_Qiu_Xing_2022, Kwon_2023}. Evidence from multiple regions shows that cost pressures already force compromises: in Britain, households reduce heating at the expense of health \citep{Anderson_White_Finney_2012}; in the U.S., poor insulation makes basic cooling unaffordable \citep{Best_Sinha_2021}, and in Arizona, low-income households delay air conditioning until homes are $4.7$–$7.5^{\circ}$F hotter than those of higher-income households \citep{Cong_Nock_Qiu_Xing_2022}. Such constraints lead to elevated rates of heat-related hospitalization and mortality \citep{Kwon_2023}.

This creates a double burden for low-income households under variable pricing. Budget pressures may force them to be highly price-responsive \citep{Brannlund_Vesterberg_2021}. However, their flexibility is unlikely to be uniform across periods: off-peak discretionary uses may be highly elastic, while peak-period services tied to health and safety may remain relatively inelastic due to health constraints. This asymmetric flexibility pattern suggests that low-income households could appear responsive in aggregate demand studies, yet remain less able to adjust consumption during the peak periods when variable pricing raises costs.

Our theoretical analysis reveals that demand flexibility increases the change in welfare when aligned with large price changes: consumers who are flexible during periods with small price adjustments gain little benefit, while those who can adjust consumption during periods with substantial price changes receive significant welfare protection. This helps explain why variable pricing benefits might be concentrated among higher-income consumers who have greater access to high-quality housing, flexible work schedules, and other demand response capabilities. Additionally, our equilibrium analysis demonstrates important spillover effects: individual flexibility contributes to aggregate demand responsiveness, which affects the magnitude of price adjustments across periods. A peak period with higher aggregate flexibility experiences smaller price changes than one in which consumers are flexible. This creates interactions between consumers in which one consumer’s flexibility influences the welfare of the other.

This paper contributes to the literature by developing a theoretical framework that identifies conditions under which consumers lose utility from variable pricing reforms and when low-income consumers are disproportionately harmed. We derive sufficient conditions based on three key consumer characteristics: consumption patterns across periods, flexibility, and price sensitivity. We show that consumers with high peak use, limited flexibility, and high price sensitivity face the greatest risk from a switch to variable pricing. By incorporating heterogeneous consumer utility functions and realistic assumptions about price responsiveness, we provide analytical tools for evaluating the distributional impacts of pricing reforms before implementation. This theoretical foundation connects empirical work on energy poverty and consumption patterns with the growing literature on electricity market design, offering insights essential for designing variable pricing policies that achieve both efficiency and equity objectives in the transition to more flexible, renewable-powered electricity systems.

\subsection{Contributions}

\textit{Theoretical Framework for Consumer Heterogeneity:} We develop a framework that explicitly incorporates three dimensions of consumer heterogeneity in electricity pricing: consumption patterns across time, flexibility, and price sensitivity (marginal disutility of expenditure). While prior studies often assume constant elasticity across consumers and periods, our framework allows these traits to vary and shows how they affect welfare outcomes.

\textit{Sufficient Conditions for Identifying Vulnerable Consumers:} We derive tractable sufficient conditions that identify which consumers lose from variable pricing transitions, both in absolute terms and relative to other groups. For quadratic loss functions, these conditions are necessary and sufficient; for general demand functions, they provide conservative screening tools. These analytical tools enable ex-ante assessment of distributional impacts.

\textit{Timing of Flexibility Benefits:} We demonstrate that demand flexibility provides welfare protection only when coincident with large price changes. Because welfare benefits scale with the square of price changes, consumers who are flexible during small shifts gain little. This timing-dependent relationship explains why higher-income consumers often benefit more from variable pricing.

\textit{Equilibrium Spillover Effects:} We show how individual flexibility aggregates to shape market prices. High aggregate peak flexibility dampens peak price increases, while low peak flexibility allows larger price adjustments that disproportionately harm inflexible consumers. These spillovers create cross-consumer interactions where one household's flexibility directly affects others through equilibrium price formation.

\textit{Distributional Impact Mechanisms:} Price sensitivity ($A_i$) amplifies welfare changes without affecting their direction. Even with identical consumption patterns and flexibility, consumers with higher marginal disutility of expenditure (e.g., low-income households with $A_\L > A_\H$) experience larger magnitude welfare swings from variable pricing.

Together, these contributions connect empirical research on energy poverty and consumption with theoretical work on electricity market design, providing analytical foundations for designing variable pricing policies that balance efficiency with equity in demand-responsive systems.

\subsection{Related Literature}

\noindent\textbf{Previous Work on Time-Varying Pricing:} The economic theory underlying variable electricity pricing dates back at least to Mohring’s (1970) peak-load pricing model, which established that efficiency requires marginal-cost pricing in each period \citep{Mohring_1970}. Early work by Feldstein (1972) identified key equity concerns, noting that lower-income consumers have less flexibility to adapt to peak pricing due to constraints in appliance ownership and housing quality, while also benefiting more from low prices for essential goods due to higher marginal utility of income \citep{Feldstein_1972}. This analysis highlighted that pricing policies ignoring heterogeneity in flexibility and electricity's relative importance across periods risk disproportionately impacting vulnerable households.

More recent theoretical work has quantified the efficiency gains from time-varying pricing. \citep{Borenstein_longrun} demonstrates how real-time pricing reduces reliance on expensive peaking plants and improves capacity cost recovery, while \cite{borenstein2022two} shows that flat-rate pricing creates larger deadweight losses than other common pricing distortions, including emissions externalities and volumetric recovery of fixed costs. These studies establish the strong efficiency case for variable pricing but give limited attention to distributional consequences.

\noindent\textbf{Empirical Evidence from Pricing Experiments:} Early empirical studies of time-of-use (TOU) tariffs in the United States, reviewed in \citep{aigner1985residential}, found significant behavioral responses, with consumers primarily reducing peak demand rather than shifting consumption to off-peak periods. \citep{caves1984consistency} provided crucial insights into consumer heterogeneity, demonstrating that responsiveness to time-varying prices depends on appliance ownership, weather conditions, and climate control needs. Importantly, they found that households with fewer appliances, typically lower income—showed limited responsiveness to TOU tariffs, reducing their ability to benefit from time-varying prices compared to households with more appliances. More recent studies of critical peak pricing (CPP) programs confirm these patterns, with greater load reductions than TOU but continued evidence that benefits vary across income groups based on appliance flexibility \citep{Faruqui_Sergici_2010}.

\noindent\textbf{Simulation Studies and Distributional Analysis:} When experimental data is unavailable, researchers use historical consumption data to simulate variable pricing impacts, typically assuming constant elasticity across consumers and periods. Results from these studies are mixed regarding distributional effects. \cite{Horowitz_Lave_2014} find that low-income consumers often perform worse under real-time pricing due to smaller, less flexible loads and limited appliance ownership. Conversely, \cite{simshauser2016inequity} conclude that lower-income households with flatter load profiles can benefit from TOU and CPP rates. \cite{burger2020efficiency} shows that while low-income consumers may face higher bills under time-varying tariffs, those with demand flexibility can gain consumer surplus due to lower marginal electricity costs.
These conflicting empirical findings reflect a key limitation in the existing literature: the assumption of homogeneous elasticity across consumers and time periods obscures the heterogeneity in flexibility that our theoretical analysis shows is crucial for understanding distributional outcomes. Our work addresses this gap by developing a framework that explicitly incorporates variation in demand responsiveness across consumer types and time periods, providing theoretical foundations for understanding when and why variable pricing creates winners and losers.

\begin{table}[H]
\centering
\caption{Key Notation and Definitions}
\label{tab:notation}
\begin{tabular}{cl}
\toprule
\textbf{Symbol} & \textbf{Definition} \\
\midrule
\multicolumn{2}{l}{\textit{Periods and Consumers}} \\
$t \in \{\text{PE}, \text{OP}\}$ & Time periods: peak (PE) and off-peak (OP) \\
$i \in \{\L, \H\}$ & Consumer types: low-income ($\L$) and high-income ($\H$) \\
Superscripts $^{\V}$, $^{\F}$ & Variable pricing and flat pricing regimes \\
\multicolumn{2}{l}{\textit{Prices and Consumption}} \\
$\pi^{\F}$ & Flat-rate electricity price (uniform across periods) \\
$\pi_t^{\V}$ & Variable price in period $t$ \\
$d_{it}^{\F}$, $d_{it}^{\V}$ & Consumption by consumer $i$ in period $t$ under flat/variable pricing \\
$d_{it}^*$ & Optimal consumption by consumer $i$ in period $t$ \\
$\bar{d}_{it}$ & Maximum consumption capacity for consumer $i$ in period $t$ \\
\multicolumn{2}{l}{\textit{Consumer Preferences}} \\
$\Util_i$ & Total utility of consumer $i$ across periods \\
$\ELoss_{it}(d)$ & Loss function from consuming $d$ units in period $t$  \\
$A_i$ & Price sensitivity (marginal disutility of expenditure) \\
$A_{\L} > A_{\H}$ & Low-income consumers more price-sensitive \\
$\hat{\ELoss}_{it}(d) = \ELoss_{it}(d)/A_i$ & Normalized loss function\\
\multicolumn{2}{l}{\textit{Supply and Welfare}} \\
$C(d_t)$ & Cost of supplying $d_t$ units in period $t$ \\
$\Welf_t$ & Social loss in period $t$ \\
$\Welf = \sum_t \Welf_t$ & Total social loss \\
\multicolumn{2}{l}{\textit{Flexibility and Elasticity}} \\
$\left|\frac{\partial d_{it}^*}{\partial \pi_t}\right|$ & Consumer $i$'s demand flexibility in period $t$ \\
$\left|\frac{\partial d_t}{\partial \pi_t}\right|$ & Aggregate flexibility in period $t$ \\
$\epsilon_{it}$ & Price elasticity of consumer $i$ in period $t$ \\
\multicolumn{2}{l}{\textit{Key Differences}} \\
$\Delta \Util_i$ & Change in consumer $i$'s utility from flat to variable pricing \\
$\Delta \pi_t = |\pi_t^{\V} - \pi^{\F}|$ & Absolute price change in period $t$ \\
$\Delta \Util_{\H} - \Delta \Util_{\L}$ & Difference in utility changes between consumer types \\
\bottomrule
\end{tabular}
\end{table}

\pagebreak
\section{Model} \label{sec:Model}
\subsection{Consumer Model}

We model a two-period electricity market with peak ($\PE$) and off-peak ($\OP$) periods, denoted $t \in \mathcal{T} \coloneqq \{\PE,\OP\}$. This market serves two consumer types $i \in \mathcal{N} \coloneqq \{\L,\H\}$, representing low-income ($\L$) and high-income ($\H$) households. Let $d_{it}$ denote consumer $i$'s electricity consumption in period $t$, constrained by maximum consumption $\bar{d}_{it}$, reflecting physical or technological limits. Taking the per-kWh price $\pi_t$ as given, consumer $i$'s utility is:
\begin{align}
\Util_i \coloneqq M_i - \sum_{t \in \mathcal{T}}(\ELoss_{it}(d_{it})+A_{i}\pi_t d_{it})
\label{eqn:ConsumUtil},\end{align}

\noindent where $\ELoss_{it}(d_{it})$ represents the loss from consuming $d_{it} < \bar{d}_{it}$, $M_i$ is consumer $i$'s maximum attainable utility, and $A_i > 0$ captures consumer $i$'s marginal disutility of expenditure. We assume $A_{\L} > A_{\H}$ to reflect the greater marginal utility of income for low-income consumers. We normalize utility so that consuming zero electricity yields zero net utility: when $d_{it} = 0$, we have $\Util_i = M_i - \ELoss_{it}(0) = 0$, which implies $M_i = \ELoss_{it}(0)$. We use the hat notation $\hat{\ELoss}_{it} \coloneqq \frac{\ELoss_{it}}{A_{i}}$ to denote electricity loss normalized by the marginal disutility of expenditure. Consumer surplus, measuring net benefit relative to cost, equals $\Util_i/A_i$. 

We focus on a two-period model to capture the fundamental distinction between high-demand periods (when air conditioning or heating needs are critical) and low-demand periods. This stylized framework isolates the key distributional mechanisms while remaining simple for exposition. As shown in Section \ref{sec:extensions}, the form of the results holds for a more general $N$ consumer, $T$ period model. 

For any feasible consumption level $d \in [0,\bar{d}_{it}]$, we assume the consumer model satisfies:
\begin{assumption} \label{Assumption:Consumer}
The electricity loss function $\ELoss_{it}(d)$ satisfies:
\begin{enumerate}
    \item[\subassumption{Assumption:ELoss_Properties}] \textbf{Regularity:} $\ELoss_{it}$ is three times differentiable with $\ELoss_{it} \geq 0$, $\ELoss_{it}' < 0$, $\ELoss_{it}'' > 0$, and $\ELoss_{it}''' < 0$.
    
    \item[\subassumption{Assumption: PeriodLoss_order}] \textbf{Peak period criticality:} Loss from under-consumption is greater during peak periods: $\ELoss_{i\OP}(d) \leq \ELoss_{i\PE}(d)$ for all $d$.
    
    \item[\subassumption{Assumption: PeriodFW_order}] \textbf{Marginal loss ordering:} The marginal reduction in loss from additional consumption satisfies $|\ELoss_{i\PE}'(d)| \geq |\ELoss_{i\OP}'(d)|$ for all $d$.
    
    \item[\subassumption{Assumption: PeriodSW_order}] \textbf{Curvature ordering:} The rate of change in marginal loss is greater during peak periods: $\ELoss''_{i\PE}(d) \geq \ELoss''_{i\OP}(d) > 0$ for all $d$.
\end{enumerate}
\end{assumption}

Assumption \ref{Assumption:Consumer} provides the foundation for our consumer model. The regularity conditions in \ref{Assumption:ELoss_Properties} are standard for non-appliance-based models of electricity consumption \citep{jiang2011multi, Jordehi_optDR_review, deng2015survey}. The decreasing loss function ($\ELoss_{it}' < 0$) reflects that consuming more electricity reduces loss, while convexity ($\ELoss_{it}'' > 0$) ensures diminishing marginal loss reduction: initial units serve critical needs (lighting, basic appliances), while later units provide smaller benefits (comfort, non-essential uses). The third derivative condition ($\ELoss_{it}''' < 0$) implies convex demand curves (see \eqref{clm:J_triple_demand_convex}). This means that demand becomes more price-responsive at higher prices, consistent with models where non-essential uses are eliminated first and consumers reduce more critical consumption as prices continue to rise \citep{burger2020efficiency}. Note that the consumer utility function \eqref{eqn:ConsumUtil} is concave in $d_{it}$ since $\frac{\partial^2 \Util_i}{\partial d_{it}^2} = -\ELoss_{it}''(d_{it}) < 0$.

We are especially concerned with the energy demand of heating or cooling devices during critical high-demand periods. The period-specific assumptions \ref{Assumption: PeriodLoss_order}--\ref{Assumption: PeriodSW_order} capture systematic differences between peak and off-peak electricity demand during weather extremes. Assumption \ref{Assumption: PeriodLoss_order} reflects that peak periods involve additional critical electricity uses (e.g., air conditioning during heat waves), leading to greater loss from forgone consumption at any given level compared to off-peak periods. Assumption \ref{Assumption: PeriodFW_order} implies that each unit of electricity results in less loss during peak periods, since appliances work harder to maintain temperature under extreme conditions, and more appliances are in use. Assumption \ref{Assumption: PeriodSW_order} indicates that marginal benefits of consumption decrease more quickly during peak periods—the first units address critical health and safety needs, but once satisfied, additional consumption yields rapidly diminishing benefits compared to the more gradual decline in the off-peak period.

These assumptions enable analysis using standard electricity loss functions, including the quadratic specifications commonly employed in demand response models \citep{burger2020efficiency, fell2014new}, while ensuring our results reflect realistic consumption patterns across time periods and income groups.

\subsection{Optimal Demand}

Formally, for a given pricing scheme $\pi_{t}, t \in \mathcal{T}$, consumer $i$'s decision problem is equivalent to:
\begin{subequations}\label{eqn:consumer_prob}
\begin{align}\label{eqn:consumer_opt}
    \underset{d_{it}}{\min} &\hspace{0.1cm}\sum_{t \in \mathcal{T}}\left(\ELoss_{it}(d_{it})+A_{i}\pi_{t} d_{it}\right)\\
    \st\quad 
    & 0 \leq d_{it} \leq \bar{d}_{it} \quad \forall t \in \mathcal{T}\label{eqn:cons_constraint2}
\end{align}
\end{subequations}

For interior solutions where $0 < d_{it}^{*} < \bar{d}_{it}$, the first-order condition is necessary and sufficient:

\begin{align}\frac{\partial\ELoss_{it}}{\partial d_{it}} \biggr \rvert_{d_{it}^*} + A_{i}\pi_{t} = 0 \implies -\frac{\partial \hat{\ELoss}_{it}}{\partial d_{it}} \biggr \rvert_{d_{it}^*} = \pi_{t}\quad\forall i \in \mathcal{N}, t \in \mathcal{T}.\label{eqn:FOC_cons}\end{align}

That is, when optimal consumption lies in the interior of the interval $[0, \bar{d}_{it}]$, the marginal avoided loss, or willingness to pay, equals the price. We focus on interior solutions throughout this work, which allows us to use the first-order condition \eqref{eqn:FOC_cons} to characterize a consumer's optimal demand as a function of price (full derivation in \ref{subsec:consumer_kkt}). Note that since $\ELoss_{it}' < 0$ and $\ELoss_{it}'' > 0$ from Assumption \ref{Assumption:Consumer}, the function $\ELoss_{it}'$ is strictly decreasing and therefore invertible on the relevant domain.

\begin{proposition} \label{prop:demand_curve}
Optimal demand as a function of price is:
    \begin{align}\label{eqn:x_opt} 
    d_{it}^*(\pi_{t}) &= \min\{\max\{(\hat{\ELoss}_{it}')^{-1}(-\pi_{t}), 0\}, \overline{d}_{it}\}\quad \forall i, t.
    \end{align}
The $\min\{\max\{\cdot\}\}$ construction ensures that demand respects both the non-negativity constraint ($d_{it} \geq 0$) and the capacity constraint ($d_{it} \leq \bar{d}_{it}$).
\end{proposition}

We refer to the optimal consumption function \eqref{eqn:x_opt} as the demand curve; see Figure \ref{fig:pi_vs_d}. The normalized utility formulation ensures that $\hat{\ELoss}_{it}(d_{it}^*) + \pi_{t} d_{it}^*$ produces the same willingness to pay and demand curve, given $\hat{\ELoss}_{it}(d_{it}^*)$, regardless of the consumer's exact price sensitivity $A_i$ or loss $\ELoss_{it}(d_{it}^*)$. The first-order condition represents the trade-off between the marginal benefit from avoided electricity loss and the marginal cost of consumption. Consequently, although a low-income consumer may have greater marginal avoided loss from electricity consumption (due to higher $A_i$), they do not necessarily demand more electricity than high-income consumers, as this depends on the interaction between $A_i$, $\bar{d}_{it}$, and the loss function $\ELoss_{it}(\cdot)$.

\begin{figure}[H]
    \centering
    \begin{tikzpicture}[x=0.5cm]
        \draw[-Latex] (0,0) -- (11,0) node [anchor=north] {$\pi_{t}$};
        \draw[-Latex] (0,0) -- (0,3) node [anchor=east] {$d_{it}$};

        \draw[TealBlue, ultra thick, domain=0:6.34, smooth] plot (\x, {3*(\x+1)^(-0.55)-1});
        \draw[TealBlue, ultra thick] (6.3,0) -- (10,0) node [anchor=south] {$d^*_{it}(\pi_{t})$};
        \draw[Mulberry, very thick, dashed, domain=10:0, smooth] plot (\x, {3*(\x+1)^(-0.55)-1}) node [anchor=west] {$(\hat{\ELoss}_{it}')^{-1}(-\pi_{t})$};
        \draw (-0.1, {3*1^(-0.55)-1}) -- (0.1, {3*1^(-0.55)-1}) node [anchor=east] {$\bar{d}_{it}$};
        \draw (1, -.06) -- (1, .06) node [anchor=north] {$-\hat{\ELoss}_{it}'(\bar{d}_{it})$};
        \draw (6.34, -.06) -- (6.34, .06) node [anchor=north] {$-\hat{\ELoss}_{it}'(0)$};
    \end{tikzpicture}
    \caption{Consumption $d^*_{it}$ as a function of price $\pi_{t}$; $d_{it}$ is bounded below by $0$ and above by $\bar{d}_{it}$.}
    \label{fig:pi_vs_d}
\end{figure}
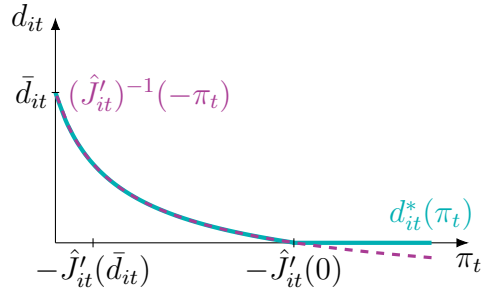

Having characterized optimal demand, we now derive how demand responds to price changes—a key determinant of distributional outcomes under variable pricing. Consumer flexibility—the responsiveness of demand to price changes—is given by the derivative of the demand function. Taking the derivative of \eqref{eqn:x_opt}, we obtain:

\begin{align}\frac{\partial d_{it}^{*}(\pi_{t})}{\partial \pi_{t}} = \begin{cases}
    0 & d_{it}^{*}(\pi_{t}) = \overline{d}_{it},\\
    \frac{-A_i}{\ELoss_{it}''(d_{it}^{*}(\pi_{t}))}  = -\frac{1}{\hat{\ELoss}_{it}''(d_{it}^{*}(\pi_{t}))}  & d_{it}^{*}(\pi_{t}) \in (0, \overline{d}_{it}),\\
    0 & d_{it}^{*}(\pi_{t}) = 0. \\
\end{cases}\label{clm:nondec_consumption}\end{align}

Therefore, for interior solutions, consumer flexibility is $\left|\frac{\partial d_{it}^{*}}{\partial \pi_{t}}\right| = \frac{1}{\hat{\ELoss}_{it}''(d_{it}^{*}(\pi_{t}))}$. For identical loss functions $\ELoss_{\L t}=\ELoss_{\H t}$, the condition $A_{\L} > A_{\H}$ implies $\hat{\ELoss}_{\L t}'' < \hat{\ELoss}_{\H t}''$, so low-income consumers are more flexible than high-income consumers. From Assumption \ref{Assumption:Consumer}, we know that $\hat{\ELoss}''_{it}>0$ and $\hat{\ELoss}''_{it}$ is decreasing in demand, so flexibility increases with demand. Furthermore, since $\hat{\ELoss}''_{it}(d_{it}^{*}(\pi_{t}))$ is the composition of two decreasing functions $\hat{\ELoss}_{it}''(d_{it})$ and $d_{it}^{*}(\pi_{t})$, flexibility also decreases as price increases.
 
Finally, we examine how consumer utility changes with price. We apply the chain rule to \eqref{eqn:ConsumUtil} and use \eqref{eqn:FOC_cons} and \eqref{clm:nondec_consumption} to show that the decrease in utility is proportional to optimal consumption:
\begin{align}
\frac{\partial \Util_{it}^{*}(\pi_{t})}{\partial \pi_{t}} = 
\begin{cases}0 & d_{it}^{*}(\pi_{t}) = \overline{d}_{it}, \\
-A_{i}d_{it}^{*}(\pi_{t}) & d_{it}^{*}(\pi_{t}) \in (0, \overline{d}_{it}),\\
 0 & d_{it}^{*}(\pi_{t}) = 0. \label{clm:dec_U}\\
\end{cases}
\end{align}

For the same level of demand, the low-income consumer experiences larger changes in utility than the high-income consumer because $A_{\L} > A_{\H}$; the low-income consumer gains more utility from a price decrease but loses more from a price increase.

Finally, we introduce the aggregate consumer to analyze market-level behavior. The aggregate consumer represents the sum of all individual consumers with aggregate consumption $d_{t} = \sum_{i\in \mathcal{N}} d_{it}$ in period $t$. We construct an aggregate loss function such that the aggregate consumer's optimization problem mirrors that of individual consumers. In \ref{Appendix:Agg_Cons}, we prove that under Assumption \ref{Assumption:Consumer}, an appropriately constructed aggregate consumer satisfies the same optimality conditions as individual consumers: the aggregate first-order condition $-\hat{\ELoss}_{t}'(d_t^*) = \pi_{t}$, the aggregate demand characterization $d_{t}^* = (\hat{\ELoss}_{t}')^{-1}(-\pi_{t})$, and aggregate flexibility $|\frac{\partial d_{t}^*}{\partial \pi_{t}}| = \frac{1}{\hat{\ELoss}_{t}''(d_t^*)}$. We also establish that $\hat{\ELoss}_{t}''$ is increasing in the curvature of individual loss functions, so aggregate flexibility falls when consumers are less flexible individually.

\subsection{Consumer Heterogeneity}

Our modeling approach addresses two critical gaps identified above. First, despite empirical evidence of substantial heterogeneity in price responsiveness \citep{Faruqui_Sergici_2010, aigner1985residential}, existing studies typically assume uniform elasticity when analyzing variable pricing impacts \citep{burger2020efficiency, Horowitz_Lave_2014}. Second, while empirical studies document varying distributional outcomes \citep{simshauser2016inequity, borenstein2012redistributional}, the literature lacks a unified theoretical framework for predicting when low-income consumers will be harmed by pricing reforms \citep{joskow2012dynamic}.

We address these gaps through two key modeling innovations: heterogeneous demand flexibility and parameterized consumption patterns.

\textbf{Demand flexibility.} From \eqref{clm:nondec_consumption}, consumer flexibility is $\left|\frac{\partial d_{it}^{*}}{\partial \pi_{t}}\right| = \frac{1}{\hat{\ELoss}_{it}''(d_{it}^{*}(\pi_{t}))}$. Existing studies of variable pricing typically assume either price-insensitive consumers or price-sensitive consumers with isoelastic demand (constant proportional response to price changes) \citep{burger2020efficiency, Horowitz_Lave_2014, simshauser2016inequity}. In these studies, elasticity is constant across consumers and periods. However, experimental studies reveal substantial heterogeneity in price responsiveness. \citet{Faruqui_Sergici_2010} document different elasticities between peak and off-peak periods, while \citet{aigner1985residential} and \citet{caves1984consistency} show that elasticity varies with temperature, appliance use, and local climate—factors that differ systematically between periods and across consumer groups. Our framework captures this heterogeneity by allowing flexibility to vary through the parameters $A_i$ and $\ELoss_{it}''$, enabling analysis of how different responsiveness patterns affect distributional outcomes. We use flexibility (quantity response per unit price change) rather than elasticity (proportional response) because flexibility directly determines the magnitude of consumption adjustments and their welfare implications.

\textbf{Consumption patterns.} A consumer's consumption pattern—encompassing demand levels and timing across periods—largely determines how they are affected by variable pricing transitions. Empirical studies reveal that low-income households exhibit different patterns than higher-income groups, but these differences vary significantly across contexts. In Australia, \citet{simshauser2016inequity} find that low-income households have larger overall consumption with more consistent demand across peak and off-peak periods. Conversely, \citet{Horowitz_Lave_2014} document that such households in their sample have more peak-dominated usage, while \citet{borenstein2012redistributional} find that high-income households are the largest overall electricity users. These differences translate directly into distributional outcomes: \citet{simshauser2016inequity} find that low-income households benefit most from tariff reform due to their flatter profiles, while \citet{Horowitz_Lave_2014} and \citet{borenstein2012redistributional} conclude that low-income households face higher bills under variable pricing due to their peak-heavy or low overall usage patterns, respectively. Our framework captures this empirical heterogeneity by parameterizing maximum consumption ($\bar{d}_{it}$) and loss curvature ($\ELoss_{it}''$), allowing analysis of how consumption patterns and pricing structures interact to determine distributional impacts.

Our parameterization generates heterogeneity in consumer flexibility $\left(\left|\frac{\partial d_{it}^*}{\partial \pi_{t}}\right| = \frac{A_i}{\ELoss_{it}''(d_{it}^*(\pi_{t}))}\right)$ and consumption patterns under flat pricing. 
We use this framework to derive general properties of optimal consumption and characterize how flexibility and consumption patterns determine consumer welfare changes. Using no assumptions regarding the differences between consumers beyond $A_{\L} > A_{\H}$ (reflecting higher marginal disutility of expenditure for low-income consumers), we identify conditions under which pricing changes disproportionately impact low-income consumers, providing theoretical foundations for understanding the distributional consequences of electricity rate design.

\subsection{Social Welfare}
We model a social welfare-maximizing planner who sets electricity prices and procures supply to serve consumer demand. Following \citet{deng2015survey}, this setup simplifies the two-stage market structure of \citet{Joskow_Tirole_2006} by combining wholesale procurement and retail pricing decisions into a single optimization problem. The planner faces a cost function $C(d_t)$ representing the cost of procuring aggregate electricity $d_t$ in period $t$, where the same cost technology applies across both peak and off-peak periods. The cost function satisfies:

\begin{assumption} \label{Assumption:Supply}
The electricity procurement cost function is non-negative, non-decreasing, and convex:
    $C(d) \geq 0$,
    $C'(d) \geq 0$, and
    $C''(d) \geq 0$ for all $d \geq 0$.
\end{assumption}

These conditions reflect standard properties of electricity supply: non-negative costs, non-decreasing marginal costs (as cheaper generation units are dispatched first), and convex costs (reflecting capacity constraints and the merit order of generation). We assume identical cost functions across periods to isolate the role of demand-side heterogeneity in driving price differences and distributional outcomes. This assumption isolates the effect of heterogeneity in consumer demand on price variation rather than mixing these effects with those of supply-side factors. 

The planner maximizes total economic surplus by maximizing aggregate consumer surplus minus electricity procurement costs. Equivalently, we formulate this as minimizing total social loss:
\begin{subequations}\label{eqn:operator_prob}
\begin{align}
    \Welf^* = & \underset{\{\pi_{t}\}}{\min} \sum_{t \in \mathcal{T}} \Welf_t(\pi_t) \label{eqn:operator_opt}\\
    \text{where } \Welf_t(\pi_t) &= C(d_{t}^*(\pi_t)) + \hat{\ELoss}_{t}(d_{t}^*(\pi_t)) \label{eqn:welfare_def}\\
    \text{and } d_{t}^*(\pi_t) &= \underset{d \in [0,\bar{d}_t]}{\argmin} \{\hat{\ELoss}_{t}(d) + \pi_td\} \label{eqn:consumer_response}\\
    \st \quad \pi_t &\geq 0 \quad\forall t \in \mathcal{T} \label{eqn:price_constraint}
\end{align}
\end{subequations}

Here, $\Welf_t(\pi_t)$ represents the social loss in period $t$, comprising procurement costs and aggregate consumer loss from under-consumption (where $\hat{\ELoss}_{t}(d_{t}^*(\pi_t))=\sum_{i \in \mathcal{N}} \ELoss_{it}(d_{it}^*(\pi_t))$. The constraint \eqref{eqn:consumer_response} reflects that consumers optimally respond to prices by minimizing their individual loss functions, where $\bar{d}_t = \sum_{i \in \mathcal{N}} \bar{d}_{it}$ represents aggregate maximum consumption in period $t$. Under Assumptions \ref{Assumption:Consumer} and \ref{Assumption:Supply}, both $C(d)$ and $\hat{\ELoss}_{t}(d)$ are convex, and the constraint set is compact, ensuring that $\Welf_t$ is convex in $d_t$ and that the planner's problem has a unique solution.

Under flat pricing, we add the constraint:
\begin{align}\label{eqn:flat_equal_price}
    \pi_t = \pi^{\F} \text{ for all } t \in \mathcal{T}
\end{align}

We use superscript $^{\V}$ to denote variable pricing equilibrium outcomes and $^{\F}$ for flat pricing equilibrium outcomes. The key difference is that variable pricing allows period-specific prices that reflect marginal costs, while flat pricing uses a single price across all periods.

Since variable pricing does not include the constraint \eqref{eqn:flat_equal_price}, the planner achieves a weakly lower social loss under variable pricing: $\Welf^{\V} \leq \Welf^{\F}$. However, improved aggregate welfare does not guarantee that all consumers benefit, nor that impacts are distributed equally. Therefore, our analysis focuses on distributional outcomes, particularly for low-income consumers who may be most vulnerable to pricing changes.

We introduce notation for changes resulting from the transition from flat to variable pricing:
\begin{align}
\Delta d_{t} &\coloneqq d_{t}^{\V} - d_{t}^{\F} \quad \text{(change in consumption in period } t\text{)}\\
\Delta \pi_{t} &\coloneqq \pi_{t}^{\V} - \pi^{\F} \quad \text{(change in price in period } t\text{)}\\
\Delta \Util_{it} &\coloneqq \Util_{it}^{\V} - \Util_{it}^{\F} \quad \text{(change in period-} t \text{ utility for consumer } i\text{)}\\
\Delta \Util_{i} &\coloneqq \Delta \Util_{i\PE} + \Delta \Util_{i\OP} \quad \text{(total utility change for consumer } i\text{)}
\end{align}

In Section \ref{sec:equilibrium}, we characterize equilibrium prices and consumption under both regimes. Section \ref{sec:Cons_Outcomes} then analyzes consumer welfare impacts is  examining two key distributional questions: (1) when does a consumer experience utility losses ($\Delta \Util_i < 0$), and (2) when do low-income consumers fare worse than high-income consumers ($\Delta \Util_{\L} - \Delta \Util_{\H} < 0$)?

\section{Equilibrium Characterization} \label{sec:equilibrium}
This section characterizes equilibrium solutions under variable and flat pricing, focusing on how differences in prices, demand, and welfare shape consumer outcomes. A central insight from this section is that the magnitude of price changes depends on the curvature of the planner’s welfare function. Steeper functions (higher curvature) limit price adjustments, while flatter functions allow larger ones. Flexibility drives this sensitivity: greater flexibility in the peak period steepens the welfare function, while inflexibility flattens it and increases price movements. 
 To demonstrate these effects, we begin by showing the equilibrium conditions under both regimes.

\subsection{Equilibrium Properties} \label{subsec:equilibrium_properties}

Under \emph{variable pricing}, the planner chooses period-specific prices $\pi_t^{\V}$ to minimize social loss. The first-order conditions require that the marginal effect of a price change on social welfare equals zero in each period:
\begin{align}\label{eqn:var_FOC} 
     \frac{\partial \Welf_{t}(\pi_t)}{\partial \pi_t} \biggr \rvert_{\pi_t^{\V}} = 0, \quad \forall t \in \mathcal{T}
\end{align}
Using the chain rule and the definition $\Welf_t(\pi_t) = C(d_t^*(\pi_t)) + \hat{\ELoss}_t(d_t^*(\pi_t))$:
\begin{align}
     \frac{\partial \Welf_{t}}{\partial \pi_t} = \left(\frac{\partial C(d_t)}{\partial d_t} + \frac{\partial \hat{\ELoss}_t(d_t)}{\partial d_t}\right) \frac{\partial d_t^*(\pi_t)}{\partial \pi_t}
\end{align}
For interior aggregate demand ($0 < d_t^{\V} < \sum_i \bar{d}_{it}$), Assumption \ref{Assumption:Consumer} ensures that $\hat{\ELoss}_t'$ is strictly decreasing and therefore invertible, which implies $\frac{\partial d_t^*}{\partial \pi_t} = -\frac{1}{\hat{\ELoss}_t''(d_t^*)} \neq 0$. Therefore, the first-order condition \eqref{eqn:var_FOC} requires:
\begin{align}\label{eqn:var_FOC_expanded} 
     \frac{\partial C(d_{t}^{\V})}{\partial d_t} + \frac{\partial \hat{\ELoss}_t(d_{t}^{\V})}{\partial d_t} = 0 \iff \frac{\partial C(d_{t}^{\V})}{\partial d_t} = -\frac{\partial \hat{\ELoss}_t(d_{t}^{\V})}{\partial d_t} = \pi_t^{\V}, 
\end{align}
where the final equality uses the aggregate consumer's first-order condition from Section \ref{sec:Model}. This equilibrium condition shows that marginal procurement cost equals marginal avoided loss, which in turn equals the equilibrium price. This aligns consumer incentives with social costs: consumers pay exactly the marginal cost of providing additional electricity, ensuring efficient consumption decisions.

Under \emph{flat pricing}, consumers still optimize according to \eqref{eqn:FOC_cons}, so marginal avoided loss equals the flat price in both periods:
\begin{align}
    -\frac{\partial \hat{\ELoss}_{\PE}(d_{\PE}^{\F})}{\partial d_{\PE}} = -\frac{\partial \hat{\ELoss}_{\OP}(d_{\OP}^{\F})}{\partial d_{\OP}} = \pi^{\F} \label{eqn:flat_consumer_BR}
\end{align}

However, marginal costs need not equal marginal avoided loss in each period under flat pricing, since the planner is constrained to use a single price. The planner's first-order condition for the constrained problem is:
\begin{align*}
    \frac{d}{d\pi^{\F}}\left[\Welf_{\PE}(\pi^{\F}) + \Welf_{\OP}(\pi^{\F})\right] = 0
\end{align*}
This yields:
\begin{align}\frac{\partial \Welf_{\PE}(\pi^{\F})}{\partial \pi^{\F}} + \frac{\partial \Welf_{\OP}(\pi^{\F})}{\partial \pi^{\F}} = 0 \iff \frac{\partial \Welf_{\PE}(\pi^{\F})}{\partial \pi^{\F}} = -\frac{\partial \Welf_{\OP}(\pi^{\F})}{\partial \pi^{\F}} \label{eqn:flat_FOC_simple}
\end{align}

This condition shows that the marginal welfare loss from raising the price in one period must exactly offset the marginal welfare gain from the same price change in the other period. This cross-period balancing is the key difference from variable pricing, where each period's price independently balances marginal costs and benefits within that period alone. Figure \ref{fig:op_probs} illustrates this equilibrium condition: at the flat price $\pi^{\F}$, the slopes of the period-specific welfare functions (shown as tangent lines) are equal in magnitude but opposite in sign, representing the marginal welfare balance described in \eqref{eqn:flat_FOC_simple}.

Figure \ref{fig:op_probs} illustrates the equilibrium condition \eqref{eqn:flat_FOC_simple}. At the flat price $\pi^{\F}$, the slopes of the period-specific welfare functions are equal in magnitude but opposite in sign, visualizing the balance of marginal welfare between periods at the flat price.

\begin{figure}[H]
        \centering
        \begin{minipage}{\textwidth}
            \centering
            \begin{tikzpicture}[x=1.8cm, y=0.3cm]
                \draw[Latex-Latex] (-1,0) -- (5.5,0) node [anchor=north] {$\pi$};
                \draw[Latex-Latex] (0,-3) -- (0,25) node [anchor=east] {$\Welf(\pi)$};
                
                \draw[ProcessBlue, domain=1:5, smooth, very thick] plot (\x, {.9*\x^2-6*\x+15}) node [anchor=south] {};
                \draw[dashed] (3.3, 4.1*3.3^2-15.6*3.3+29) -- (3.3, 0) node [anchor=north] {\scriptsize $\pi_{\PE}^{\V}$};
                \draw[dashed] (3.3, 4.1*3.3^2-15.6*3.3+29) -- (0, 4.1*3.3^2-15.6*3.3+29) node [anchor=east] {\scriptsize $\Welf(\pi_{\PE}^{\V})$};
                \draw[dashed]   (3.3, 3.2*3.3^2-9.6*3.3+1.8) -- (0, 3.2*3.3^2-9.6*3.3+1.8) node [anchor=east] {\scriptsize $\Welf_{\PE}(\pi_{\PE}^{\V})$};
                \draw[dashed]   (1.9, 3.2*1.9^2-9.6*1.9+9) -- (0, 3.2*1.9^2-9.6*1.9+9);
                \node[text=ProcessBlue] at (4.5,4.5) {\footnotesize $\Welf_{\PE}(\pi)$};

                \draw[Red, domain=0:3, smooth, very thick] plot (\x, {3.2*\x^2-9.6*\x+9}) node [anchor=east] {};
                 
                 \draw[dashed]   (1.5, 3.2*1.5^2-9.6*1.5+9) -- (0, 3.2*1.5^2-9.6*1.5+9) ;
                 \draw[dashed]   (1.9, 3.2*1.9^2-9.6*1.9+9) -- (0, 3.2*1.9^2-9.6*1.9+9);
                 \node at (-.34,3.2*1.5^2-9.6*1.5+10) {\scriptsize $\Welf_{\OP}(\pi^{\F})$} ;
                 \draw[dashed]   (1.5, 3.2*1.5^2-9.6*1.5+9) -- (0, 3.2*1.5^2-9.6*1.5+9); 
                 \node at (-.34,3.2*1.5^2-9.6*1.5+8.5) {\scriptsize $\Welf_{\OP}(\pi_{\OP}^{\V})$} ;
                \draw[dashed]   (1.5, 4.1*1.5^2-15.6*1.5+29) -- (1.5, 0) node [anchor=north] {\scriptsize $\pi_{\OP}^{\V}$};
                \node at (-.3,4.1*1.9^2-15.6*1.9+30.5) {\scriptsize $\Welf(\pi_{\OP}^{\V})$};
                \draw[dashed]   (1.5, 4.1*1.5^2-15.6*1.5+29) -- (0, 4.1*1.5^2-15.6*1.5+29) ;
                \node[text=red] at (2.6,8) {\footnotesize$\Welf_{\OP}(\pi)$};
                
                \draw[ForestGreen, domain=.3:3.5, smooth, very thick] plot (\x, {4.1*\x^2-15.6*\x+29}) node [anchor=south] {};
                
                \draw[dashed] (1.9, 4.1*1.9^2-15.6*1.9+29) -- (1.9, 0) node [anchor=north] {\footnotesize $\pi^{\F}$} ;
                \node at (-.3,4.1*1.9^2-15.6*1.9+28.5) {\scriptsize $\Welf(\pi^{\F})$};
                \node[text=ForestGreen] at (1.6,18) {\footnotesize $\Welf(\pi)$};
                
                \draw[dashed] (1.9, 4.1*1.9^2-15.6*1.9+29) -- (0, 4.1*1.9^2-15.6*1.9+29) ;
                \draw[dashed] (1.9, 0.9*1.9^2-6*1.9+15) -- (0, 0.9*1.9^2-6*1.9+15) node [anchor=east] {\scriptsize $\Welf_{\PE}^{\F}$};;
                \draw[black, very thick] (1.9, {0.9*1.9^2-6*1.9+15}) -- ++(0.6, {2*1*1.9-5.35}) ;
                \draw[black, very thick] (1.9, {0.9*1.9^2-6*1.9+15}) -- ++(-0.6, {-(2*1*1.9-5.35)}) node [anchor=north] {\scriptsize $\frac{\partial\Welf_{\PE}^{\F}}{\partial \pi^{\F}}$};
                
                \draw[black, very thick] (1.9, {3.2*1.9^2-9.6*1.9+9}) -- ++(0.6, {-(2*1*1.9-5.35)}) node [anchor=north] {\scriptsize $\frac{\partial\Welf_{\OP}^{\F}}{\partial \pi^{\F}}$};
                \draw[black, very thick] (1.9, {3.2*1.9^2-9.6*1.9+9}) -- ++(-0.6, {(2*1*1.9-5.35)})  ;
            
            \end{tikzpicture}
       \caption{The planner's welfare functions by period and pricing regime. The peak period objective $\Welf_{\PE}(\pi)$ (blue) and off-peak period objective $\Welf_{\OP}(\pi)$ (red) sum to the total objective $\Welf(\pi)$ (green). Under variable pricing, each period's price independently minimizes its welfare function ($\pi_{\PE}^{\V}$ and $\pi_{\OP}^{\V}$). Under flat pricing, the single price $\pi^{\F}$ balances welfare effects across periods: the tangent lines show that $\frac{\partial\Welf_{\PE}}{\partial \pi^{\F}} = -\frac{\partial\Welf_{\OP}}{\partial \pi^{\F}}$, satisfying the equilibrium condition \eqref{eqn:flat_FOC_simple}.
       }      
            \label{fig:op_probs}        
        \end{minipage}
    \end{figure}
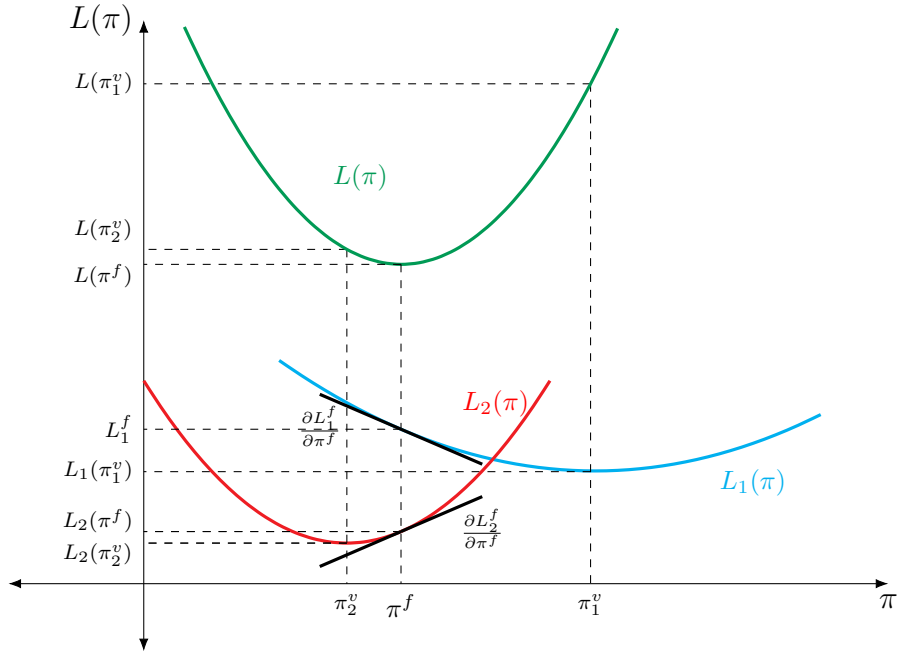

Since the flat pricing problem is simply the variable pricing problem with the additional constraint \eqref{eqn:flat_equal_price}, we can order the total welfare under the two pricing regimes:

\begin{proposition}
    \label{lem:eq_order} Under equilibrium conditions, the planner's loss under variable pricing is no more than the loss under flat pricing:
    $\Welf^{\V} \leq \Welf^{\F}$.
\end{proposition}
Under variable pricing, consumers pay the marginal cost of provision, eliminating the cross-subsidization inherent in flat pricing.

In addition to establishing an order between the operator's objective under both regimes, the equilibrium conditions shown in \eqref{eqn:flat_consumer_BR}-\eqref{eqn:flat_FOC_simple} allow us to understand how a pricing change impacts each period. We begin by discussing the relationship between equilibrium prices and quantities under both regimes.

\textbf{Price Ordering:} Using the convexity of $\Welf_{t}$ and the first-order conditions, we can show that the flat price lies between the variable prices. Since the cost function $C(d)$ is identical across periods but demand patterns differ due to consumer preferences (Assumptions \ref{Assumption: PeriodLoss_order}--\ref{Assumption: PeriodSW_order}), the period with a higher baseline demand (peak) will have a higher variable price to balance marginal costs and benefits.

\textbf{Demand Ordering:} Since demand is decreasing in price (equation \eqref{clm:nondec_consumption}), the price ordering directly implies the consumption ordering. Variable pricing reduces peak consumption and increases off-peak consumption relative to flat pricing.

We provide a detailed characterization in the following proposition, with formal proof in \ref{Appendix:equilibrium}:

\begin{proposition}\label{prop:price_demand_order}
Aggregate prices and demand under both pricing regimes satisfy:
$$\pi_{\PE}^{\V} \geq \pi^{\F} \geq \pi_{\OP}^{\V}, \quad d_{\OP}^{\F} \leq d_{\OP}^{\V} \leq d_{\PE}^{\V} \leq d_{\PE}^{\F}$$

Individual consumer demand orderings depend on their flexibility response to price changes:
\begin{enumerate}
\item[] Case 1: $d_{i\OP}^{\F} \leq d_{i\OP}^{\V} \leq d_{i\PE}^{\V} \leq d_{i\PE}^{\F}$. This occurs when the consumer's demand remains greater in the peak period than off-peak under variable pricing.
\item[] Case 2: $d_{i\OP}^{\F} \leq d_{i\PE}^{\V} \leq d_{i\OP}^{\V} \leq d_{i\PE}^{\F}$. This occurs when the consumer is sufficiently flexible that their usage pattern reverses and they consume at least as much in the off-peak as in the peak period.
\end{enumerate}
\end{proposition}

\textbf{Economic Interpretation:} Flat pricing implicitly subsidizes peak consumption: off-peak is overpriced relative to marginal cost, and peak is underpriced, leading to inefficient overuse at peak and underuse off-peak. Variable pricing eliminates this by setting period-specific prices that reflect true marginal costs.

The distributional implications depend on individual consumption patterns. Consumers who use relatively more electricity during off-peak periods effectively subsidize those who consume more during peak periods under flat pricing. 
The individual demand also reveals an important insight about flexibility. 
All consumers reduce peak and increase off-peak use under variable pricing, but highly flexible consumers may shift so much that off-peak consumption exceeds peak use, benefiting disproportionately from lower off-peak prices.

Having established the equilibrium price and quantity relationships, we now examine the \emph{magnitude} of price changes when switching from flat to variable pricing. The key determinant is the curvature of the planner’s welfare function $\Welf_t(\pi_t)$. As shown in \eqref{eqn:flat_FOC_simple}, the flat price balances marginal welfare effects across periods, but Figure \ref{fig:op_probs} illustrates that the curvature $\frac{\partial^{2}\Welf_{t}}{\partial \pi^{2}}$ dictates how far prices can deviate from this balance. As Figure \ref{fig:op_probs} shows, flatter welfare functions allow larger price changes, amplifying impact on consumers.
We formalize the relationship between curvature and price changes with the following proposition:
\begin{proposition}\label{lem:pi_deltas}
The ratio of price changes when switching from flat to variable pricing is inversely related to the average curvature of the planner's welfare functions:
$$\frac{|\Delta \pi_{\PE}|}{|\Delta \pi_{\OP}|} = {\overline{\frac{\partial^2 \Welf_{\OP}}{\partial \pi^2}}}\bigg/ {\overline{\frac{\partial^2 \Welf_{\PE}}{\partial \pi^2}}}$$
where $\overline{\frac{\partial^2 \Welf_{t}}{\partial \pi^2}}$ denotes the average second derivative of $\Welf_t$ over the interval between the flat price and the variable price for period $t$.
\end{proposition}

See \ref{Appendix:equilibrium} for proof. This result highlights a key trade-off: flatter welfare functions permit larger price adjustments with smaller welfare penalties, while steeper functions limit deviations. Thus, consumers with high usage in flatter periods face the largest utility impacts from pricing reform.

 Consumer flexibility plays an important role in determining the curvature of the operator's problems, which, in turn, determines the change in prices. Combining the consumer and planner first-order conditions \eqref{eqn:flat_consumer_BR} and \eqref{eqn:flat_FOC_simple}, yields an expression for $\pi^{\F}$ as a flexibility-weighted average of marginal procurement costs:
\begin{align}
        \pi^{\F} = \frac{\left|\frac{\partial d^{\F}_{\PE}}{\partial \pi}\right|}{\left|\frac{\partial d^{\F}_{\PE}}{\partial \pi}\right| + \left|\frac{\partial d^{\F}_{\OP}}{\partial \pi}\right|} \cdot C'(d_{\PE}^{\F}) + \frac{\left|\frac{\partial d^{\F}_{\OP}}{\partial \pi}\right|}{\left|\frac{\partial d^{\F}_{\PE}}{\partial \pi}\right| + \left|\frac{\partial d^{\F}_{\OP}}{\partial \pi}\right|} \cdot C'(d_{\OP}^{\F}), 
        \label{eqn:flat_price}
\end{align}
where the weights are the relative flexibilities of aggregate demand in each period (shown in \ref{Appendix:equilibrium}). This expression demonstrates that the flat price lies closer to the marginal cost of the period with \emph{higher consumer flexibility}. As this is  evaluated at the flat price rather than variable prices, we further examine the role of flexibility in determining operator problem curvature in the next subsection.

\subsection{Price Adjustment Analysis} \label{subsec:price_adjustment_analysis}

We next expand the second derivative $\frac{\partial^{2} \Welf_{t}}{\partial \pi^{2}_t}$ to better understand the relationship between the curvature of the operator's problem and aggregate consumer flexibility. Using the chain rule:
\begin{align}
\frac{\partial^{2} \Welf_t}{\partial \pi_t^{2}} &= \frac{\partial}{\partial \pi_{t}} \left[\frac{\partial d_{t}^{*}}{\partial \pi_t} \cdot \frac{\partial \Welf_{t}}{\partial d_t}\biggr\rvert_{d_t^*} \right] \nonumber\\
&= \frac{\partial^2 d_{t}^{*}}{\partial \pi_t^2} \cdot \frac{\partial \Welf_{t}}{\partial d_t}\biggr\rvert_{d_t^*} + \left(\frac{\partial d_{t}^{*}}{\partial \pi_t}\right)^2 \cdot \frac{\partial^2 \Welf_{t}}{\partial d_t^2}\biggr\rvert_{d_t^*} \nonumber\\
&= \frac{\partial^2 d_{t}^{*}}{\partial \pi_t^2} \cdot \left[C'(d^*_{t}) + \hat{\ELoss}_{t}'(d^*_{t})\right] + \left(\frac{\partial d_{t}^{*}}{\partial \pi_t}\right)^2 \cdot \left[C''(d_{t}^*) + \hat{\ELoss}_{t}''(d^*_{t})\right] \label{eqn:second_deriv_welf}
\end{align}

 The first term of \eqref{eqn:second_deriv_welf} involves the second derivative of demand with respect to price, which relates to flexibility through 
 $\partial^{2}d_{t}/\partial \pi_{t}^{2} = \hat{\ELoss}_{t}'''(d_{t})(\partial d_{t}/\partial \pi_{t})^{3}  $. The curvature of the welfare function also depends on aggregate consumer flexibility through the second term, where higher flexibility $\left|\frac{\partial d_{t}^{*}}{\partial \pi_t}\right|$ directly increases the curvature. The net effect on curvature depends on the relative magnitudes and signs of these terms. 
 
 Examining the first term of \eqref{eqn:second_deriv_welf}, the sign depends on $C'(d^*_{t}) + \hat{\ELoss}_{t}'(d^*_{t})$ as $\frac{\partial^2 d_{t}^{*}}{\partial \pi_t^2} > 0$.  From Proposition \ref{prop:price_demand_order}, we know that at the flat price:
\begin{itemize}
    \item In the off-peak period: $\pi^{\F} \geq \pi_{\OP}^{\V}$, so consumption is below the efficient level. This means $C'(d_{\OP}^{\F}) < \pi^{\F}$ and $\pi^{\F} < |\hat{\ELoss}_{\OP}'(d_{\OP}^{\F})| = -\hat{\ELoss}_{\OP}'(d_{\OP}^{\F})$. Therefore: $C'(d_{\OP}^{\F}) + \hat{\ELoss}_{\OP}'(d_{\OP}^{\F}) < \pi^{\F} - \pi^{\F} = 0$. 
    \item In the peak period: $\pi^{\F} \leq \pi_{\PE}^{\V}$, so consumption exceeds the efficient level. This means $C'(d_{\PE}^{\F}) > \pi^{\F}$ and $\pi^{\F} > |\hat{\ELoss}_{\PE}'(d_{\PE}^{\F})| = -\hat{\ELoss}_{\PE}'(d_{\PE}^{\F})$. Therefore: $C'(d_{\PE}^{\F}) + \hat{\ELoss}_{\PE}'(d_{\PE}^{\F}) > \pi^{\F} - \pi^{\F} = 0$.
\end{itemize}

The second term is always positive since $\left(\frac{\partial d_{t}^{*}}{\partial \pi_t}\right)^2 > 0$ and both $C''$ and $\hat{\ELoss}_{t}''$ are positive.

This analysis suggests that peak and off-peak flexibility influence curvature differently, and thus affect price adjustments under variable pricing in different ways.
 For the peak period, we can establish that higher flexibility increases curvature. Since both terms in \eqref{eqn:second_deriv_welf} are positive, and the second term $\left(\frac{\partial d_{\PE}^{*}}{\partial \pi}\right)^2 \cdot [C'' + \hat{\ELoss}_{\PE}'']$ increases directly with flexibility $\left|\frac{\partial d_{\PE}^{*}}{\partial \pi}\right|$, higher peak flexibility unambiguously increases $\frac{\partial^{2} \Welf_{\PE}}{\partial \pi^{2}}$. For the off-peak period, the negative first term and positive second term create offsetting effects, making the net impact of flexibility on curvature ambiguous. Since steeper welfare functions (higher curvature) constrain the planner's ability to deviate prices from the flat rate, these curvature effects directly determine the magnitude of price adjustments. To summarize, combining flexibility's effect on operator problem curvature with Prop.~\ref{lem:pi_deltas}, we can identify the implications of flexibility for price changes:

\begin{itemize}
\item \emph{Peak period flexibility:} Higher consumer flexibility in the peak period increases the curvature of the welfare function $\frac{\partial^{2} \Welf_{\PE}}{\partial \pi_{\PE}^{2}}$, which by Proposition \ref{lem:pi_deltas} leads to smaller peak price adjustments $|\Delta \pi_{\PE}|$ and therefore a larger ratio $|\Delta \pi_{\OP}|/|\Delta \pi_{\PE}|$.

\item \emph{Off-peak period flexibility:} The effect of off-peak flexibility on curvature $\frac{\partial^{2} \Welf_{\OP}}{\partial \pi_{\OP}^{2}}$ remains ambiguous due to offsetting effects in equation \eqref{eqn:second_deriv_welf}. Consequently, the overall impact of off-peak flexibility on the relative magnitude of price adjustments $|\Delta \pi_{\OP}|/|\Delta \pi_{\PE}|$ cannot be determined without additional assumptions about the relative magnitudes of the cost and loss function derivatives.

\end{itemize}
 Together with Proposition \ref{lem:pi_deltas}, these curvature effects map directly into asymmetric peak vs. off-peak price adjustments. This reveals a key insight: lower consumer flexibility in the peak period decreases curvature of the operator's problem and increases price changes. This counterintuitive result occurs because flexible peak consumers increase the curvature $\frac{\partial^{2} \Welf_{\PE}}{\partial \pi_{\PE}^{2}}$, making the planner's welfare function more sensitive to peak price deviations, restricting the extent to which the peak price can move from the flat rate and forcing larger adjustments in the off-peak price to maintain the overall balance. {This creates complex interactions where demand-side management programs targeted at flexible consumers may generate spillover benefits for less responsive consumer groups, as established through the price formation mechanisms in Section \ref{sec:equilibrium}. During temperature extremes with low aggregate peak flexibility, peak prices rise sharply. Consumers with below-average flexibility suffer most from large peak price increases while benefiting little from smaller off-peak decreases.} Policymakers should note when peak demand reflects essential needs (e.g., temperature control for health), since large peak price increases then have large welfare impacts and may warrant protections.

To summarize, this section established the equilibrium conditions under both pricing regimes and demonstrated that:
\begin{enumerate}
\item Flat prices lie between variable prices: $\pi_{\OP}^{\V} \leq \pi^{\F} \leq \pi_{\PE}^{\V}$
\item Cross-subsidization occurs under flat pricing, with redistribution across periods due to overconsumption in the peak period and underconsumption in the off-peak period.
\item Consumer flexibility determines the magnitude of price changes through curvature effects
\end{enumerate}
In Section \ref{sec:Cons_Outcomes}, we utilize these equilibrium results to analyze how individual consumers are affected by the transition from flat to variable pricing, examining both the magnitude and distribution of welfare changes across different consumer types. We will show that the relationship between flexibility and price adjustments has important distributional implications.

\section{Effects of Pricing Change on Consumer 
Utility}\label{sec:Cons_Outcomes}
This section analyzes how the transition from flat to variable pricing affects individual consumer welfare.
We examine the role of individual consumption patterns, flexibility, and price sensitivity in determining individual welfare changes. The analysis reveals that the effects of pricing reform are heterogeneous across consumers. Utility changes depend on the interaction of relative price shifts, consumption, and flexibility. We show that consumers with high peak demand, low off-peak demand, and little flexibility lose the most, especially when peak prices rise more than off-peak prices fall. Additionally, price sensitivity scales the magnitude of changes but not their direction.
We develop our theoretical results using general demand functions, but illustrate the key insights using two specific functional forms: linear demand (derived from quadratic utility) and isoelastic demand. These examples demonstrate how our general results apply in commonly used modeling frameworks
and provide concrete illustrations of the distributional effects. Details of both demand functions are provided in \ref{appendix:demand_func}. 

\subsection{Change in Utility When Switching from Flat to Variable Pricing}  
The total change in utility for consumer $i$ in period $t$ when switching from flat to variable pricing is $\Delta \Util_{it} = \Util_{it}^{\V} - \Util_{it}^{\F}$. Using the envelope theorem and equation \eqref{clm:dec_U}, this can be expressed as: 
\begin{align}
    \Delta \Util_{it} = \int_{\pi^{\F}}^{\pi_{t}^{\V}}\frac{\partial \Util_{it}(\pi_{t})}{\partial \pi_t} d\pi_t &= -\int_{\pi^{\F}}^{\pi_{t}^{\V}} A_id^*_{it}(\pi_t) d\pi_t. \label{eqn:utility_change} 
\end{align}
 This integral represents the change in consumer surplus over the price interval,
scaled by price sensitivity. 
We can then find the total change in utility of a consumer by summing \eqref{eqn:utility_change} across both periods:
$$\Delta \Util_{i} = \Delta \Util_{i\PE} + \Delta \Util_{i\OP}$$
Hence, the total utility change depends on the consumer's consumption profile across periods and the relative magnitude of price changes in each period. 

Figure \ref{fig:delta_util} illustrates the geometric interpretation of utility changes for consumer $i$. The figure plots demand curves for both periods, with price on the vertical axis and quantity (scaled by $A_i$) on the horizontal axis. Based on equation \eqref{eqn:utility_change}, we note: 

\begin{itemize}
\item[-] The \emph{red shaded area} between $\pi^{\F}$ and $\pi_{\PE}^{\V}$ represents the utility loss $|\Delta \Util_{i\PE}|$ from higher peak prices
\item[-] The \emph{blue shaded area} between $\pi_{\OP}^{\V}$ and $\pi^{\F}$ represents the utility gain $\Delta \Util_{i\OP}$ from lower off-peak prices
\end{itemize}

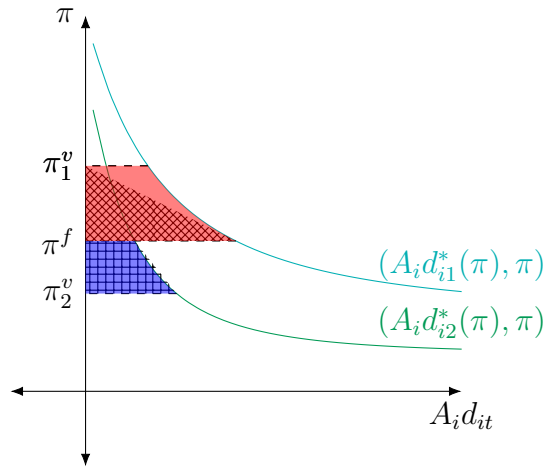
\begin{figure}[H]
    \centering
    \begin{tikzpicture}
        \draw[Latex-Latex] (-1,0) -- (5,0) node [anchor=north] {$A_{i}d_{it}$};
        \draw[Latex-Latex] (0,-1) -- (0,5) node [anchor=east] {$\pi$};

        \draw[TealBlue, domain=0.1:5, smooth] plot (\x, {(2/(0.5*\x+1))^2 + 1}) node [anchor=south] {$(A_id^*_{i\PE}(\pi),\pi)$};
        \draw[ForestGreen, domain=0.1:5, smooth] plot (\x, {((0.5*\x+1)/1.5)^-3.3 + 0.5}) node [anchor=south] {$(A_id^*_{i\OP}(\pi),\pi)$};

        \draw[dashed] (2,2) -- (0,2) node [anchor=east] {$\pi^{\F}$};
        \draw[dashed] (1.2099, 1.3) -- (0,1.3) node [anchor=east] {$\pi_{\OP}^{\V}$};
        \draw[dashed] (0.8284, 3) -- (0,3) node [anchor=east] {$\pi_{\PE}^{\V}$};
        \draw[dashed, domain= 0.6531:1.2099] (0.8284, 3) -- (0,3) node [anchor=east] {$\pi_{\PE}^{\V}$};
        \fill[pattern=grid]  (0.6531 ,-1.8658*0.6531 + 3.2187)--(1.2099,-1.8658*1.2099 + 3.5575 )--(0,1.3) -- (0,2);
        \fill[pattern=crosshatch] (0, 2)--(2, 2) -- (0.8284, -0.5*0.8284 + 3) -- (0,3) -- cycle;
        \fill[blue, domain=0.6531:1.2099, opacity=0.5] plot (\x, {((0.5*\x+1)/1.5)^-3.3 + 0.5}) -- (0,1.3) -- (0,2) -- cycle; 
        \fill[red, domain=0.8284:2, opacity=0.5] plot (\x, {(2/(0.5*\x+1))^2 + 1}) -- (0,2) -- (0,3) -- cycle;
    \end{tikzpicture}
    \caption{Consumer utility changes from switching to variable pricing. The red area shows utility loss from higher peak prices ($|\Delta \Util_{i\PE}|$), while the blue area shows utility gain from lower off-peak prices ($\Delta \Util_{i\OP}$). Net utility loss occurs when the red area exceeds the blue area. The crosshatched region represents the conservative lower bound on peak utility loss used in condition \eqref{eqn:delta_util_suff}, while the gridded region represents the conservative upper bound on off-peak utility gain, illustrating how the linear demand condition provides sufficient conditions through these approximations used in Theorem~\ref{thm:delta_util}.}
    \label{fig:delta_util}
\end{figure}

Consumer $i$ experiences a net utility loss when the red area (peak period loss) exceeds the blue area (off-peak period gain).
This occurs when the consumer has: \emph{(i)} High peak period consumption relative to off-peak consumption; or \emph{(ii)} Low flexibility (steep demand curves that create large areas for given price changes), or \emph{(iii)} Exposure to periods where price increases are large relative to price decreases. In practice, consumers are most vulnerable when they exhibit multiple of these characteristics simultaneously, but any single factor can be sufficient for net utility losses if it is pronounced enough.  

We now formalize these intuitions and identify the conditions under which consumers lose utility from pricing reform. Recall that from \eqref{clm:nondec_consumption}, $\frac{\partial d_{it}^*}{\partial \pi_t} = \frac{-A_i}{\hat{\ELoss}_{it}''(d^*_{it})}$. From Assumption \ref{Assumption:Consumer}, $\hat{\ELoss}_{it}'''(d^*_{it}) < 0$, which implies that $\hat{\ELoss}_{it}''$ is decreasing in consumption. Since higher prices reduce consumption, $\hat{\ELoss}_{it}''$ increases with price, causing flexibility $\left|\frac{\partial d_{it}^*}{\partial \pi_t}\right| = \frac{A_i}{\hat{\ELoss}_{it}''(d^*_{it})}$ to decrease as price increases. This establishes that demand curves are convex.

\textbf{Flexibility Bounds:} Since demand curves are convex (flexibility decreases with price):
\begin{itemize}
\item In the \emph{peak period}: Since $\pi^{\V}_{\PE} \geq \pi^{\F}$, flexibility ranges from highest at $\pi^{\F}$ to lowest at $\pi^{\V}_{\PE}$
\item In the \emph{off-peak period}: Since $\pi^{\F} \geq \pi^{\V}_{\OP}$, flexibility ranges from highest at $\pi^{\V}_{\OP}$ to lowest at $\pi^{\F}$
\end{itemize}

These bounds allow us to construct conservative estimates in Theorem \ref{thm:delta_util} by evaluating flexibility at the most favorable points for the consumer.

\textbf{Demand Bounds:} Since demand decreases with price:
\begin{itemize}
\item \emph{Peak consumption}: decreases from $d_{i\PE}^{\F}$ (at lower price $\pi^{\F}$) to $d_{i\PE}^{\V}$ (at higher price $\pi^{\V}_{\PE}$)
\item \emph{Off-peak consumption}: increases from $d_{i\OP}^{\F}$ (at higher price $\pi^{\F}$) to $d_{i\OP}^{\V}$ (at lower price $\pi^{\V}_{\OP}$)
\end{itemize}
These bounds allow us to establish sufficient conditions for utility loss by finding the extreme cases. If a consumer loses utility even under the most favorable assumptions about their flexibility and consumption levels, they will certainly lose utility under actual conditions.

\begin{theorem} \label{thm:delta_util} 
Consumer $i$'s total utility change under a switch from flat to variable pricing is:
\begin{align} 
\Delta \Util_{i} &= A_{i}\left(-\int_{\pi^{\F}}^{\pi_{\PE}^{\V}} d_{i\PE}^{*}(\pi_t)  d\pi_t + \int_{\pi_{\OP}^{\V}}^{\pi^{\F}} d_{i\OP}^{*}(\pi_t)  d\pi_t\right) \label{eqn:total_util_breakdown}
\end{align}

Consumer $i$ experiences a utility loss ($\Delta \Util_{i} < 0$) under the following conditions:

\textbf{(a) Linear Demand:} If and only if:
\begin{align} 
(|\Delta \pi_{\OP}|d_{i\OP}^{\F} - |\Delta \pi_{\PE}| d_{i\PE}^{\F}) + \frac{(\Delta\pi_{\PE})^2}{2}\left|\frac{\partial d^{\F}_{i\PE}}{\partial \pi_{\PE}}\right| + \frac{(\Delta\pi_{\OP})^2}{2}\left|\frac{\partial d^{\V}_{i\OP}}{\partial \pi_{\OP}}\right| < 0 \label{eqn:delta_util_suff}
\end{align}

\textbf{(b) Isoelastic Demand:} If and only if:
\begin{align}
\frac{d_{i\PE}^\F }{1 + \epsilon_{i\PE}^\F} \left[ 1 - \left( \frac{\pi_{\PE}^\V}{\pi^\F} \right)^{1 + \epsilon_{i\PE}^\F} \right] 
+ \frac{d_{i\OP}^\F }{1 + \epsilon_{i\OP}^\F} \left[ 1 - \left( \frac{\pi_{\OP}^\V}{\pi^\F} \right)^{1 + \epsilon_{i\OP}^\F} \right] < 0 \label{eqn:delta_util_suff_3}
\end{align}
\end{theorem}

\textbf{Applicability and interpretation of conditions:} The linear demand condition \eqref{eqn:delta_util_suff} is necessary and sufficient under two scenarios: when $\ELoss_{it}$ is quadratic (yielding linear demand), where the convexity bounds used in the derivation become equalities, and when consumer flexibility approaches zero ($\hat{\ELoss}_{it}'' \rightarrow \infty$) in both periods, where the quadratic flexibility terms vanish and the condition reduces to purely consumption-weighted price changes: $(|\Delta \pi_{\OP}|d_{i\OP}^{\F} - |\Delta \pi_{\PE}| d_{i\PE}^{\F}) < 0$.

The isoelastic condition \eqref{eqn:delta_util_suff_3} provides a necessary and sufficient condition for $\Delta\Util_i<0$ when demand is isoelastic. The $\frac{1 }{1 + \epsilon_{it}^\F} \left[ 1 - \left( \frac{\pi_{t}^{\V}}{\pi^\F} \right)^{1 + \epsilon_{it}^\F} \right] $ terms capture the interaction between elasticity and price changes. This becomes clearer in the alternative formulation of \eqref{eqn:total_util_breakdown}:
\begin{align}
\Delta \Util_i = -\int_{\pi^\F}^{\pi_{\PE}^\V} d_{i\PE}^\F \cdot \left( \frac{\pi_{\PE}}{\pi^\F} \right)^{\epsilon_{i\PE}^\F} \cdot    d\pi_{\PE}
+ \int_{\pi_{\OP}^\V}^{\pi^\F} d_{i\OP}^\F \cdot \left( \frac{\pi_{\OP}}{\pi^\F} \right)^{\epsilon_{i\OP}^\F}    d\pi_{\OP}, \label{eqn:iso_elastic_utility}
\end{align} 
which shows that utility change increases in $|\epsilon_{it}|$. Like flexibility in the linear case, elasticity has a uniformly positive impact on the change in utility.

For general demand functions that are neither linear nor isoelastic, the linear demand condition provides a conservative screening test:

\begin{corollary} \label{cor:sufficient_change}
Under Assumption \ref{Assumption:Consumer}, if condition \eqref{eqn:delta_util_suff} holds for consumer $i$, then they experience a utility loss ($\Delta \Util_i < 0$) from the pricing change regardless of the specific functional form of their loss functions $\ELoss_{i\PE}$ and $\ELoss_{i\OP}$.
\end{corollary}

Condition \eqref{eqn:delta_util_suff} is sufficient because it uses conservative bounds that underestimate the welfare benefits of flexibility. Since the convexity bounds become equalities for linear demand and inequalities otherwise, any consumer satisfying condition \eqref{eqn:delta_util_suff} will experience utility losses under their actual demand curve. The areas used for condition \eqref{eqn:delta_util_suff} are illustrated in Figure \ref{fig:delta_util}, where the crosshatched region shows how the utility loss is lower bounded using the slope of the demand curve (reciprocal of flexibility) and the demand at the flat price. The gridded region shows how utility gain is upper bounded using demand at the flat price and flexibility at the variable price. The generality of these bounds makes condition \eqref{eqn:delta_util_suff} a valuable screening tool: policymakers can identify consumers who will definitely lose from pricing reform, even without knowing the exact functional form of their demand curves.

\textbf{Economic interpretation:} Both conditions decompose utility changes into four distinct effects:

\emph{1. Consumption Level Effects:} The terms $(|\Delta \pi_{\OP}|d_{i\OP}^{\F} - |\Delta \pi_{\PE}| d_{i\PE}^{\F})$ capture the direct impact of price changes on consumers' flat-rate consumption levels. This represents what would happen if consumers could not adjust their consumption in response to price changes—pure incidence effects based on existing consumption patterns.

\emph{2. Flexibility Effects:} The quadratic terms $\frac{(\Delta\pi_t)^2}{2}|\frac{\partial d_{it}^{\F}}{\partial \pi_t}|$ (linear case) or elasticity-weighted terms (isoelastic case) capture second-order welfare effects from consumers' ability to adjust consumption in response to price changes. These terms represent the welfare value of demand responsiveness—the benefit consumers derive from being able to partially offset adverse price changes through consumption adjustments.

\emph{3. Price sensitivity effects:} From \eqref{eqn:total_util_breakdown}, $A_i$ scales the magnitude of utility changes: more sensitive consumers experience larger gains or losses from the same price shift. This scaling effect becomes crucial when comparing low- and high-income consumers.

\emph{4. Relative price change effects:} As $|\Delta \pi_{\PE}|$ increases relative to $|\Delta \pi_{\OP}|$, the balance shifts toward utility losses. The term $(|\Delta \pi_{\OP}|d_{i\OP}^{\F} - |\Delta \pi_{\PE}| d_{i\PE}^{\F})$ becomes more negative as peak price increases dominate off-peak price decreases. This means consumers are worse off when peak period price adjustments are large relative to off-peak adjustments, independent of their flexibility.

\textbf{Policy Implications:}
The decomposition in \eqref{eqn:delta_util_suff} shows that flexibility mitigates losses (or enhances gains). However, the quadratic terms $\frac{(\Delta \pi_t)^2}{2}|\frac{\partial d_{it}^{\F}}{\partial \pi_t}|$ demonstrate that demand responsiveness has differential value depending on timing. Those flexible only during periods with small price changes gain little additional protection since the flexibility terms remain negligible. However, consumers who are flexible during periods with large price changes benefit substantially as the quadratic terms provide significant welfare protection through consumption adjustments. Consumers who are inflexible in both periods face utility changes determined purely by consumption-weighted price differences $(|\Delta \pi_{\OP}|d_{i\OP}^{\F} - |\Delta \pi_{\PE}| d_{i\PE}^{\F})$.

These results suggest two sets of policies for different consumer types. Policies that directly promote demand responsiveness, such as smart meter deployment, time-varying rate education, or flexible appliance adoption, should be prioritized for consumers who lack flexibility due to technical limitations. However, a different set of policies may be needed for consumers who have very different levels of flexibility between periods or who remain inflexible after technical limitations have been addressed. These consumers are more likely to be trapped between negative health impacts from foregone electricity consumption and high electricity bills. It is not the ability to decrease consumption that is limiting their response, but the negative effects of doing so. Here, targeted protection mechanisms such as bill assistance programs, gradual rate transitions, or programs to insulate homes may be more effective.

\subsection{Difference in Change of Utility Between Consumers}
We now derive conditions under which the low-income consumer loses more utility than the high-income consumer from pricing reform.

\begin{proposition}\label{prop:low_high_util_compare}
The difference in utility changes between high-income and low-income consumers is:
\begin{align}
    \Delta \Util_{\H} - \Delta \Util_{\L} &= \int_{\pi^{\F}}^{\pi_{\PE}^{\V}} (A_{\L} d_{\L\PE}^{*}(\pi) - A_{\H} d_{\H\PE}^{*}(\pi))  d\pi + \int_{\pi_{\OP}^{\V}}^{\pi^{\F}} (A_{\H}d_{\H\OP}^{*}(\pi) - A_{\L}d_{\L\OP}^{*}(\pi))  d\pi \label{eqn:util_diff_general}
\end{align}
The low-income consumer loses more utility ($\Delta \Util_{\L} < \Delta \Util_{\H}$) if and only if the following conditions hold:

\textbf{Linear Demand:}
\begin{align} 
&A_{\L}\left((|\Delta \pi_{\OP}|d_{\L\OP}^{\F} - |\Delta \pi_{\PE}| d_{\L\PE}^{\F} ) + \frac{(\Delta\pi_{\PE})^2}{2}\left|\frac{\partial d^{\F}_{\L\PE}}{\partial \pi}\right| + \frac{(\Delta\pi_{\OP})^2}{2}\left|\frac{\partial d^{\V}_{\L\OP}}{\partial \pi}\right|\right) \nonumber\\ 
&< A_{\H}\left((|\Delta \pi_{\OP}|d_{\H\OP}^{\F} - |\Delta \pi_{\PE}| d_{\H\PE}^{\F} ) + \frac{(\Delta\pi_{\PE})^2}{2}\left|\frac{\partial d^{\V}_{\H\PE}}{\partial \pi}\right| + \frac{(\Delta\pi_{\OP})^2}{2}\left|\frac{\partial d^{\F}_{\H\OP}}{\partial \pi}\right|\right) \label{eqn:linear_comparison}
\end{align}

\textbf{Isoelastic Demand:}
\begin{align}
&A_{\L}\left(\frac{d_{\L\PE}^\F }{1 + \epsilon_{\L\PE}} \left[ 1 - \left( \frac{\pi_{\PE}^\V}{\pi^\F} \right)^{1 + \epsilon_{\L\PE}} \right] + \frac{d_{\L\OP}^\F }{1 + \epsilon_{\L\OP}} \left[ 1 - \left( \frac{\pi_{\OP}^\V}{\pi^\F} \right)^{1 + \epsilon_{\L\OP}} \right]\right) \nonumber \\ 
&< A_{\H}\left(\frac{d_{\H\PE}^\F }{1 + \epsilon_{\H\PE}} \left[ 1 - \left( \frac{\pi_{\PE}^\V}{\pi^\F} \right)^{1 + \epsilon_{\H\PE}} \right] + \frac{d_{\H\OP}^\F }{1 + \epsilon_{\H\OP}} \left[ 1 - \left( \frac{\pi_{\OP}^\V}{\pi^\F} \right)^{1 + \epsilon_{\H\OP}} \right]\right) \label{eqn:isoelastic_comparison}
\end{align}
\end{proposition}

Proposition \ref{prop:low_high_util_compare} provides conditions under which low-income consumers experience worse utility changes than high-income consumers under variable pricing reform. The linear demand condition is necessary and sufficient when consumer demand is linear (quadratic utility) and in the limit where flexibility approaches zero ($\hat{\ELoss}_{it}'' \rightarrow \infty$). This latter case corresponds to studies that only consider the change in incidence for consumers at a fixed historical demand. The linear demand condition provides a conservative test for identifying when pricing reform will disproportionately harm low-income consumers for general demand functions that satisfy Assumption \ref{Assumption:Consumer}.
\begin{corollary}\label{cor:sufficient_change_compare}
Under Assumption \ref{Assumption:Consumer}, if condition \eqref{eqn:linear_comparison} holds, then low-income consumers lose more utility than high-income consumers ($\Delta \Util_{\L} < \Delta \Util_{\H}$) from the pricing change, regardless of the specific functional forms of their loss functions $\ELoss_{\L\PE}, \ELoss_{\L\OP}, \ELoss_{\H\PE},$ and $\ELoss_{\H\OP}$.
\end{corollary}
  As with Corollary \ref{cor:sufficient_change}, this condition uses convexity bounds to ensure reliable screening across all demand functions satisfying Assumption \ref{Assumption:Consumer}. To establish conservative bounds, we lower bound low-income consumer welfare losses (making their situation appear better) while upper bounding high-income consumer welfare losses (making their situation appear worse), ensuring the sufficient condition is robust. 

The interpretation of Proposition \ref{prop:low_high_util_compare} mirrors that of Theorem \ref{thm:delta_util}: 
\begin{itemize}
\item Consumers fare better relative to others when their consumption is less concentrated in the peak (i.e., relatively more off-peak demand).
\item 
When high-income consumers are more flexible than low-income consumers, they benefit from price decreases and are hurt less by price increases, as the ability to reduce peak consumption and increase off-peak consumption results in greater cost savings relative to low-income consumers.

\item As with individual utility analysis, the timing of flexibility matters: consumers benefit when aggregate peak flexibility is high (constraining peak price increases) and when they personally can adjust during periods with large price changes, as shown by the squared price terms in \eqref{eqn:linear_comparison}.

\end{itemize}
Consumer price sensitivity plays a crucial role in both conditions of Proposition \ref{prop:low_high_util_compare}. The ratio $A_{\L}/A_{\H} > 1$ acts as a multiplier that amplifies the relative impact of pricing changes on low-income consumers' utility. Because low-income consumers have a higher marginal disutility of expenditure ($A_{\L} > A_{\H}$), they experience larger utility changes from any given pricing reform, even when consumption patterns and flexibility are held constant. When consumers have substantially different price sensitivities, a pricing change can have disproportionately large impacts on low-income households while barely affecting high-income consumers. This creates an inherent distributional asymmetry in electricity pricing policy: reforms that appear modest in aggregate can impose significant welfare costs on vulnerable populations. Consequently, rate-makers implementing variable pricing should recognize that these changes may have outsized impacts on low-income consumers and consider complementary measures—such as targeted bill assistance or gradual rate transitions—to mitigate potential distributional harms.

For isoelastic demand, condition \eqref{eqn:isoelastic_comparison} shows that elasticity parameters mediate the impact of price changes on relative welfare outcomes, with higher elasticity allowing consumers to better adjust to price changes and providing advantages in the distributional comparison similar to flexibility benefits in the linear demand case.

These results demonstrate that distributional outcomes from variable pricing depend critically on the interaction between individual consumption patterns, flexibility characteristics, and aggregate market responses. The conditions derived provide policymakers with analytical tools for identifying vulnerable consumer groups before implementing pricing reforms.

\section{Model Extensions}
\label{sec:extensions}
While Sections \ref{sec:Model}-\ref{sec:Cons_Outcomes} consider a model with only two periods, two consumers, and no profit restrictions, the core results generalize under more complex settings. In this section, we discuss how our model extends to (1) incorporate profit constraints and (2) account for arbitrarily many periods and consumers. In the following, we present a concise version of our results. For the full proofs, please see \ref{appendix:extensions}.

\subsection{Profit Constraints in Variable Pricing Implementation}

Real-world implementations of variable pricing face regulatory constraints requiring utilities to maintain financial viability while transitioning from flat-rate structures. Unlike our baseline model, where the social planner freely optimizes welfare through marginal cost pricing, utility regulators typically mandate that pricing reforms satisfy revenue adequacy requirements or predetermined profit levels. These constraints ensure recovery of sunk infrastructure investments while addressing political economy concerns about utility earnings under rate restructuring. While such requirements move prices away from first-best welfare outcomes, they better reflect practical policy implementation constraints.

To analyze how profit constraints reshape equilibrium pricing and consumer outcomes, we extend the baseline operator problem in \eqref{eqn:operator_opt}-\eqref{eqn:price_constraint} to include an explicit profit requirement. We constrain total profit to equal a predetermined level $\bar{\Prof}$, such as the profit under flat pricing: $\sum_{t\in\mathcal{T}}d_{t}(\pi^{\F}) \cdot \pi^{\F} - C(d_{t}(\pi^{\F}))$. The modified optimization problem becomes:
\begin{subequations}\begin{align}
    \Welf^* = & \underset{\{\pi_{t}\}}{\min} \sum_{t \in \mathcal{T}} \Welf_t(\pi_t) \label{eqn:operator_ext}\\
    \text{where } \Welf_t(\pi_t) &= C(d_{t}^*(\pi_t)) + \hat{\ELoss}_{t}(d_{t}^*(\pi_t))  
    \text{  and  } d_{t}^*(\pi_t) = \underset{d \in [0,\bar{d}_t]}{\argmin} \{\hat{\ELoss}_{t}(d) + \pi_td\}  &&\forall t \in \mathcal{T}\\
    \st \quad &\pi_t \geq 0 &&\forall t \in \mathcal{T} 
    \\&\sum_{t \in \mathcal{T}} \Prof_{t}(\pi_t)  = \bar{\Prof}, \label{eqn:profit_eq}
\end{align}
\end{subequations}
where $\Prof_{t}(\pi_t) = d_{t}^*(\pi_t) \pi_t - C(d_{t}^*(\pi_t))$ represents period-$t$ profit. The Lagrangian for this problem is:
\begin{align}
\mathcal{L} = \sum_t \Welf_t(\pi_t) +  \nu \left(\sum_t \Prof_t(\pi_t) - \bar{\Prof}\right). \label{eqn:lagrangian_profit}
\end{align}

Let $\pi^*_t$ denote the optimal price in period $t$ for the profit-constrained operator problem. Taking the first-order condition from the Lagrangian \eqref{eqn:lagrangian_profit} with respect to $\pi_t$ yields:
\begin{align}
& (C'(d^*_t) - \pi_{t}^*) + \nu \left(\left(1 + \frac{1}{\epsilon_t}\right)\pi_{t}^* - C'(d^*_t)\right) = 0, \label{eqn:revenue_lagrange}
\end{align}
where $\nu$ is the Lagrange multiplier associated with the profit constraint and $\epsilon_t < 0$ is the price elasticity of demand in period $t$.

Rearranging equation \eqref{eqn:revenue_lagrange} gives:
\begin{align}
\frac{\nu}{1-\nu} = \frac{\pi_{t}^* - C'(d^*_t)}{\pi_{t}^*/\epsilon_t}. \label{eqn:revenue_nu}
\end{align}

Since $\nu$ is constant across periods, equating expressions for any two periods $t$ and $t'$ yields the ratio of price distortions:
\begin{align}
\frac{\pi^*_t - C'(d_t^*)}{\pi^*_{t'} - C'(d_{t'}^*)} = \frac{\pi^*_t/\epsilon_t}{\pi^*_{t'}/\epsilon_{t'}}. \label{eqn:distortion_ratio}
\end{align}
This shows that price distortions from marginal cost are inversely related to demand elasticity, with less elastic periods experiencing larger markups under binding profit constraints. 

The direction of price adjustment depends on whether the required profit level binds above or below the unconstrained optimum:

\begin{enumerate}
\item \emph{Revenue adequacy constraint:} 
If $\bar{\Prof} > \sum_{t\in\mathcal{T}}[\pi_t^{\V}d_t^{\V} - C(d^{\V}_t)]$ (requiring higher profits than welfare-optimal pricing generates), then $\pi_t^* > \pi_t^{\V}$ in all periods to increase total profit.

\item \emph{Profit limitation constraint:} 
If $\bar{\Prof} < \sum_{t\in\mathcal{T}}[\pi_t^{\V}d_t^{\V} - C(d^{\V}_t)]$ (capping profits below the welfare-optimal level), then $\pi_t^* < \pi_t^{\V}$ in all periods to reduce total profit.
\end{enumerate}

Since price distortions are inversely related to demand elasticity, periods with less elastic demand experience larger deviations from marginal cost pricing under binding profit constraints.

The distributional implications follow directly from our main results in Section \ref{sec:Cons_Outcomes}. The sufficient conditions in Theorem \ref{thm:delta_util} remain valid with profit-constrained prices $\pi_t^*$ substituted for unconstrained prices $\pi_t^{\V}$. Revenue adequacy constraints that increase price deviations will expand the set of consumers satisfying the loss condition, while profit caps that reduce price deviations will contract this set. The consumer characteristics that determine vulnerability—consumption timing, flexibility, and price sensitivity—remain the same, but the magnitude of distributional effects changes with the binding profit constraint.

\subsection{Extension to Many Consumers and Periods}

Our model readily accommodates more than two consumers since the aggregate consumer properties (Assumption \ref{Assumption:Consumer} and equations \eqref{eqn:FOC_cons}-\eqref{clm:nondec_consumption}) hold for any number of consumers. The distributional analysis extends to comparisons across consumer quantiles or relative to population means rather than simple pairwise comparisons.

For multiple time periods, the key results generalize with appropriate notation changes. The flat price remains a flexibility-weighted average of marginal costs:
\begin{align}
\pi^{\F} = \frac{\sum_{t \in \mathcal{T}} C'(d_{t}(\pi^{\F})) \left|\frac{\partial d_{t}}{\partial \pi}\right|_{\pi^{\F}}}{\sum_{t \in \mathcal{T}} \left|\frac{\partial d_{t}}{\partial \pi}\right|_{\pi^{\F}}}
\end{align}
The sufficient conditions for utility loss extend by summing across all periods, maintaining the same form as our two-period analysis. Complete proofs are provided in \ref{appendix:extensions}.

\section{Discussion}
\label{sec:discussion}

Our theoretical framework identifies conditions under which consumer heterogeneity creates systematic patterns of winners and losers under variable electricity pricing reforms. Theorem \ref{thm:delta_util} and Proposition \ref{prop:low_high_util_compare} demonstrate that consumers most vulnerable to utility losses from variable pricing exhibit predictable characteristics: high consumption during peak periods, limited consumption during off-peak periods, and inflexible demand during periods with large price changes. Empirical studies such as \cite{Horowitz_Lave_2014} and \cite{borenstein2012redistributional} find that this vulnerability profile often corresponds to households with limited economic resources, older housing stock, or constrained ability to shift electricity usage across time periods.

Conversely, consumers who benefit from pricing reforms typically have more balanced consumption profiles and greater flexibility to respond to price incentives. While these characteristics are more often associated with higher-income households that have better appliances and housing stock, our theoretical framework demonstrates that the underlying consumption and flexibility patterns matter more than income categories per se.

\subsection{Empirical Literature and Theoretical Insights}
Much of the empirical literature on variable pricing impacts predominantly relies on simulation studies that impose constant elasticity assumptions across consumers and time periods \citep{burger2020efficiency, Horowitz_Lave_2014, simshauser2016inequity, Borenstein_billvolatility}. While these studies provide valuable insights into aggregate demand response, the constant elasticity approach obscures two critical sources of heterogeneity that our analysis shows are fundamental to distributional outcomes: variation in flexibility across consumer groups and variation in flexibility across time periods for individual consumers. 

The empirical literature's conflicting findings on distributional impacts can be understood through our theoretical lens. Studies finding that low-income consumers are harmed by variable pricing typically examine contexts where these consumers have peak-dominated consumption profiles—such as heating-intensive usage in cold climates—combined with limited flexibility during high-price periods \citep{Horowitz_Lave_2014}. Conversely, studies finding benefits for low-income consumers analyze settings where these consumers have more balanced consumption profiles across time periods \citep{simshauser2016inequity}. Our theoretical results explain these divergent empirical findings: it is not income per se that determines outcomes, but rather the interaction between consumption patterns, flexibility capabilities, and the specific pattern of price changes across periods.

The constant elasticity assumption used in simulation studies creates an artificial link between baseline consumption and demand responsiveness, meaning that consumers' flexibility patterns mechanically mirror their consumption patterns. However, our analysis demonstrates that if flexibility and consumption patterns are not completely dependent, each plays an important role in determining distributional outcomes. For example, a consumer with high peak consumption might have flexibility concentrated in off-peak periods (due to work schedules, appliance constraints, or housing characteristics), leading to very different welfare outcomes than would be predicted by uniform elasticity assumptions. Proposition \ref{prop:low_high_util_compare} shows that distributional outcomes depend on how each consumer group's flexibility aligns with the pattern of price changes across periods. Consumers who are flexible during periods with small price adjustments gain little benefit, while those who can adjust consumption during periods with large price changes receive substantial welfare benefits. This suggests that empirical evaluations of variable pricing programs should incorporate more sophisticated representations of heterogeneous demand response capabilities, particularly when assessing equity implications.


Beyond reconciling empirical findings, the equilibrium analysis reveals a spillover effect operating through the pricing mechanism: individual flexibility contributes to aggregate demand flexibility, which in turn affects the curvature of the planner's welfare function and the magnitude of price adjustments across periods (Proposition \ref{lem:pi_deltas}). This creates a counterintuitive result: when consumers are highly flexible in peak periods, the planner constrains peak price increases because small price deviations cause large welfare losses when consumption changes drastically. Conversely, when peak flexibility is low, prices can rise substantially without triggering major demand shifts, leading to larger price increases that disproportionately harm inflexible consumers.


Price sensitivity ($A_i$), and its role in encouraging flexibility create a tension inherent in encouraging demand response. Higher price sensitivity enables greater demand flexibility, which can provide welfare protection. Simultaneously, price sensitivity also amplifies the utility impacts of any pricing change. Low-income consumers, who tend to have higher $A_i$, are more heavily affected by pricing reforms regardless of their flexibility or usage profile, leaving them systematically more vulnerable to adverse changes and greater beneficiaries of positive ones. 


This amplification effect has implications for energy justice. When low-income consumers are less flexible than high-income consumers during periods experiencing large price increases—a common scenario given differential access to well-insulated housing and health needs—they face a double burden: they face a higher loss from foregone electricity use when adjusting consumption, and their higher price sensitivity magnifies the welfare impact of those changes.


\subsection{Methodological Considerations and Limitations}
Our analysis uses utility rather than consumer surplus as the primary welfare measure to capture broader impacts of electricity access on household well-being. Consumer surplus, calculated by normalizing utility by the marginal disutility of expenditure ($A_i$), measures welfare changes in terms of willingness to pay and facilitates interpersonal comparisons by expressing all impacts in monetary units. However, this approach may be inadequate for evaluating policies affecting essential services like electricity, where foregone consumption can have serious health and safety consequences that extend beyond consumers' revealed willingness to pay. The choice of welfare measure has particular importance for energy justice analysis. When low-income households reduce electricity consumption in response to price increases, the welfare loss encompasses not just the monetary burden but also the health impacts of inadequate heating or cooling \citep{Anderson_White_Finney_2012, Cong_Nock_Qiu_Xing_2022, Shi_Wang_Qiu_Deng_Xie_Zhang_Ma_2024}. We highlight temperature control, but other impacts, include the educational consequences of insufficient lighting and the social isolation caused by the inability to power communication devices, are worthy of attention as well. Our utility-based approach captures these broader welfare impacts while maintaining analytical tractability, with the parameter $A_i$ representing not just price sensitivity but also the relative importance of electricity versus other goods in each consumer's welfare function.
This methodological choice aligns with growing recognition in energy policy that traditional economic metrics may inadequately capture the stakes involved in electricity access for vulnerable populations.


Our analysis prioritizes distributional precision over supply-side complexity, making several simplifying assumptions that highlight opportunities for future research while preserving the core insights about consumer heterogeneity and equity. The most significant simplification is our focus on short-run marginal costs, which abstracts from fixed cost recovery, capacity payments, and long-term investment incentives. This choice reflects both analytical tractability and the extensive existing literature on these topics \citep{borenstein2022two, burger2020efficiency, Borenstein_longrun, Joskow_Tirole_2006, Joskow_Tirole_2007}, which has established efficient mechanisms for recovering non-marginal costs through various pricing structures.

However, combining our heterogeneous consumer framework with realistic fixed cost recovery represents an important avenue for future research, particularly as utilities increasingly face pressure to recover grid modernization investments while maintaining equitable rate structures. The interaction between distributional concerns and cost recovery methods, such as capacity charges, connection fees, or progressive rate structures, could significantly affect the welfare impacts we identify. Similarly, incorporating investment in demand response technologies and grid infrastructure could reveal how capital cost allocation affects the distribution of variable pricing benefits.


Our exclusion of supply cost variation, while analytically convenient, may understate the distributional challenges of variable pricing in practice. As electricity systems integrate higher shares of renewable energy, the variation in net load between peak and off-peak periods is likely to increase \citep{Chao_2010}, amplifying the price differences that drive our distributional results. Climate policies promoting renewable integration may thus intensify the equity challenges we identify, making our findings increasingly policy-relevant as decarbonization accelerates. 

Our assumption of perfect price responsiveness, while providing an upper bound on demand flexibility benefits, may actually understate risks to low-income consumers. Limited access to smart appliances, home automation systems, and flexible work arrangements means that low-income households often cannot fully respond to price signals even when they understand them \citep{calver2021demand}. Future research incorporating these realistic constraints on demand response would likely find even larger distributional disparities, reinforcing our key finding that variable pricing requires complementary policies to protect vulnerable populations. 


\subsection{Policy Implications}

The conditions we derive reveal that successful variable pricing programs must go beyond promoting aggregate demand response to ensure that flexibility capabilities are equitably distributed across consumer groups. Policies that enhance demand flexibility only among high-income consumers may inadvertently worsen distributional outcomes by creating larger welfare gaps between consumer groups. Understanding the spillover effects through price formation is crucial for designing equitable variable pricing programs and complementary policies to protect vulnerable populations.

Our results demonstrate that standard variable pricing implementations risk exacerbating energy inequality unless accompanied by targeted policy interventions. Several policy mechanisms can preserve the efficiency benefits of marginal cost pricing while protecting vulnerable consumers from adverse distributional impacts.

\textbf{Rate Design Modifications:} Progressive block pricing structures can address both volumetric inefficiencies and equity concerns by incorporating consumer income or historical usage patterns into rate design \citep{burger2020efficiency, Borenstein_equityblock}. Opt-in variable pricing programs offer another pathway that maintains system benefits while providing consumer protection \citep{Borenstein_billvolatility, gambardella2018time}.

\textbf{Technology Access and Infrastructure:} Beyond rate design modifications, complementary technology and assistance programs are essential for equitable variable pricing implementation. Subsidized access to smart thermostats, programmable appliances, and battery storage can democratize demand flexibility capabilities, while targeted bill assistance during high-price periods can provide safety nets for vulnerable households.

\textbf{Targeted Protection Measures:} Demand response initiatives targeting vulnerable populations should prioritize measures not directly tied to demand response, such as insulation investments and assistance with healthcare during temperature extremes. This should be done not just because of fairness concerns, but because it will help these consumers reduce electricity consumption when it is expensive for society to produce.

The key insight from our analysis is that successful variable pricing requires coordinated policy packages that address both individual flexibility constraints and the systemic factors that concentrate benefits among already-advantaged consumers. The distributional challenges identified in our analysis stem from a fundamental tension in electricity policy: the increasing need for demand flexibility to support renewable energy integration and grid reliability conflicts with the potential for time-varying prices to harm vulnerable populations. As electrification accelerates and intermittent renewables compose larger shares of the generation mix, peak demand management becomes increasingly critical, making some form of dynamic pricing essential for system efficiency.

\section*{Acknowledgements}
We thank the MIT Portugal Program and the IDSS Initiative for Combatting Systemic Racism for financial support of this research. We also gratefully acknowledge support from the MIT UROP program, which enabled undergraduate research contributions to this project.   
\appendix

\newpage

\section{Linear and Isoelastic Demand Functions} \label{appendix:demand_func}

\textbf{Linear Demand (Quadratic Utility)}
Linear demand arises from quadratic utility functions of the form:
r$$ \Util_{it} = k_{it} \bar{d}_{it}^{2} - k_{it}(\bar{d}_{it} - d_{it})^2 - A_i \pi_t d_{it} $$
where $k_{it} > 0$ measures consumer $i$'s sensitivity to electricity shortfalls in period $t$, and $\bar{d}_{it}$ represents the consumer's ideal consumption level. This specification, commonly used in demand response models \citep{deng2015survey}, yields the loss function:
$$ \ELoss_{it}(d_{it}) = k_{it}(\bar{d}_{it} - d_{it})^2 $$
with derivatives $\frac{\partial \ELoss_{it}}{\partial d_{it}} = -2k_{it}(\bar{d}_{it} - d_{it})$ and $\frac{\partial^2 \ELoss_{it}}{\partial d_{it}^2} = 2k_{it}$.

The resulting optimal demand is:
$$ d^*_{it}(\pi_t) = \min\left\{\max\left\{\bar{d}_{it} - \frac{A_i \pi_t}{2k_{it}}, 0\right\}, \bar{d}_{it}\right\} $$
For interior solutions, this yields linear demand with flexibility $\left|\frac{\partial d^*_{it}}{\partial \pi_t}\right| = \frac{A_i}{2k_{it}}$.

\textbf{Isoelastic Demand}
Standard isoelastic demand functions exhibit constant price elasticity:
$$ d_{it}(\pi_t) = d_{it}^{\F} \left(\frac{\pi_t}{\pi^{\F}}\right)^{\epsilon_{it}} $$
where $d_{it}^{\F}$ is consumption at the reference price $\pi^{\F}$, and $\epsilon_{it} < 0$ is the constant price elasticity. However, pure isoelastic demand creates integrability issues at extreme prices.

\emph{Piecewise Specification:} To ensure analytical tractability while maintaining isoelastic properties over relevant price ranges, we use the following piecewise specification:

\begin{align}
d_{it}(\pi_t) = \begin{cases}
d_{it}^{\F} \left(\frac{\pi_t}{\pi^{\F}}\right)^{\epsilon_{it}} & \text{if } \pi_{\text{low}} \leq \pi_t \leq \pi_{\text{high}} \\[0.5em]
d(\pi_{\text{high}}) - m_{\text{high}} \cdot (\pi_t - \pi_{\text{high}}) & \text{if } \pi_t > \pi_{\text{high}} \\[0.5em]
d(\pi_{\text{low}}) + m_{\text{low}} \cdot (\pi_t - \pi_{\text{low}}) & \text{if } \pi_t < \pi_{\text{low}}
\end{cases}
\end{align}

where the boundary values and slopes ensure continuity:

\emph{Upper boundary ($\pi_t > \pi_{\text{high}}$):}
\begin{align}
d(\pi_{\text{high}}) &= d_{it}^{\F} \left(\frac{\pi_{\text{high}}}{\pi^{\F}}\right)^{\epsilon_{it}} \\
m_{\text{high}} &= \epsilon_{it} \cdot d_{it}^{\F} \left(\frac{\pi_{\text{high}}}{\pi^{\F}}\right)^{\epsilon_{it}-1} \cdot \frac{1}{\pi^{\F}}
\end{align}

\emph{Lower boundary ($\pi_t < \pi_{\text{low}}$):}
\begin{align}
d(\pi_{\text{low}}) &= d_{it}^{\F} \left(\frac{\pi_{\text{low}}}{\pi^{\F}}\right)^{\epsilon_{it}} \\
m_{\text{low}} &= \epsilon_{it} \cdot d_{it}^{\F} \left(\frac{\pi_{\text{low}}}{\pi^{\F}}\right)^{\epsilon_{it}-1} \cdot \frac{1}{\pi^{\F}}
\end{align}

This construction ensures that both the demand function and its derivative are continuous at the boundary points $\pi_{\text{low}}$ and $\pi_{\text{high}}$. By setting these thresholds sufficiently far from the relevant price range (e.g., $\pi_{\text{low}} = 0.01 \cdot \pi^{\F}$ and $\pi_{\text{high}} = 100 \cdot \pi^{\F}$), the linear regions have negligible impact on the analysis while ensuring integrability.

\emph{Properties:} In the isoelastic region, the price elasticity is:
$$ \epsilon_{it} = \frac{d \ln d_{it}}{d \ln \pi_t} = \epsilon_{it} $$
and the flexibility is:
$$ \left|\frac{\partial d_{it}}{\partial \pi_t}\right| = \frac{|\epsilon_{it}| d_{it}^{\F}}{\pi^{\F}} \left(\frac{\pi_t}{\pi^{\F}}\right)^{\epsilon_{it}-1} $$

\textbf{Relationship Between Functional Forms}
Both specifications satisfy Assumption \ref{Assumption:Consumer} regarding convexity and monotonicity. The key difference lies in flexibility patterns. Linear demand exhibits constant flexibility $\left|\frac{\partial d^*_{it}}{\partial \pi_t}\right| = \frac{A_i}{2k_{it}}$
while isoelastic demand exhibits price-dependent flexibility that varies with price level.


\section{Section \ref{sec:Model} Proofs}
We derive the KKT conditions for the consumer and planner problems. Because the consumer problem is a convex optimization problem that satisfies Slater's condition, the KKT conditions provide necessary and sufficient conditions for optimality \cite{boyd_convex_2004}. 

\subsection{Consumer Problem}\label{subsec:consumer_kkt}

\textbf{Proof of \eqref{eqn:FOC_cons} and Proposition \ref{prop:demand_curve}}
\begin{proof}
The objective function $\frac{\ELoss_{it}(d_{it})}{A_{i}} + \pi_{t}d_{it}$ is convex in $d_{it}$ since $\ELoss_{it}$ is convex by Assumption \ref{Assumption:Consumer} and positive scaling preserves convexity, while $\pi_{t}d_{it}$ is linear. The constraint set $[0,\bar{d}_{it}]$ is convex and compact. For any $d_{it} \in (0,\bar{d}_{it})$, both inequality constraints are satisfied with strict inequality, so Slater's condition holds. Therefore, the KKT conditions are necessary and sufficient for optimality.

The Lagrangian is:
$$\mathcal{L}(d_{it},\lambda_{it}^1,\lambda_{it}^2) = \frac{\ELoss_{it}(d_{it})}{A_{i}} + \pi_{t}d_{it} - \lambda_{it}^1 d_{it} + \lambda_{it}^2 (d_{it}-\bar{d}_{it})$$

The KKT conditions are:
\begin{align}
\text{Stationarity: } \quad &\frac{1}{A_i}\frac{\partial\ELoss_{it}}{\partial d_{it}} \biggr \rvert_{d_{it}^*} + \pi_{t} - \lambda_{it}^{1} + \lambda_{it}^{2} = 0\\
\text{Primal feasibility: } \quad &0 \leq d_{it}^* \leq \bar{d}_{it}\\
\text{Dual feasibility: } \quad &\lambda_{it}^1, \lambda_{it}^2 \geq 0\\
\text{Complementary slackness: } \quad &\lambda_{it}^1 d_{it}^* = 0, \quad \lambda_{it}^2 (\bar{d}_{it} - d_{it}^*) = 0
\end{align}

\emph{Case 1: Interior solution} ($0 < d_{it}^* < \bar{d}_{it}$)
Both constraints are slack, so $\lambda_{it}^1 = \lambda_{it}^2 = 0$. The stationarity condition gives:
$$\frac{\partial\ELoss_{it}}{\partial d_{it}} \biggr \rvert_{d_{it}^*} = -A_{i}\pi_{t}$$

Since $\ELoss_{it}' < 0$ and strictly decreasing by Assumption \ref{Assumption:Consumer}, the inverse function exists and:
$$d_{it}^* = (\ELoss_{it}')^{-1}(-A_{i}\pi_{t}) = (\hat{\ELoss}_{it}')^{-1}(-\pi_{t})$$

\emph{Case 2: Boundary solutions}
\begin{itemize}
\item[-] If $d_{it}^* = 0$: then $\lambda_{it}^1 \geq 0$ and $\lambda_{it}^2 = 0$
\item[-] If $d_{it}^* = \bar{d}_{it}$: then $\lambda_{it}^1 = 0$ and $\lambda_{it}^2 \geq 0$
\end{itemize}

Combining all cases yields the characterization in \eqref{eqn:x_opt}.
\end{proof}
\hspace{2cm}

 \textbf{Proof of \eqref{clm:nondec_consumption}}
 
 \begin{proof}
We compute the derivative $\frac{\partial d_{it}^{*}(\pi_{t})}{\partial \pi_{t}}$ for each case in the demand function \eqref{eqn:x_opt}.

\emph{Interior case:} For $d_{it}^{*}(\pi_{t}) \in (0, \overline{d}_{it})$, we have $d_{it}^{*}(\pi_{t}) = \left(\frac{\ELoss_{it}'}{A_i}\right)^{-1}(-\pi_{t})$. Using the inverse function theorem:

\begin{align}
\frac{\partial d_{it}^{*}(\pi_{t})}{\partial \pi_{t}} &= \frac{\partial}{\partial \pi_{t}}\left[\left(\frac{\ELoss_{it}'}{A_i}\right)^{-1}(-\pi_{t})\right]\\
&= \frac{1}{\left(\frac{\ELoss_{it}'}{A_i}\right)'\left(d_{it}^{*}(\pi_{t})\right)} \cdot (-1)\\
&= \frac{1}{\frac{\ELoss_{it}''(d_{it}^{*}(\pi_{t}))}{A_i}} \cdot (-1)\\
&= \frac{-A_i}{\ELoss_{it}''(d_{it}^{*}(\pi_{t}))}
\end{align}

Since $\ELoss_{it}'' > 0$ by Assumption \ref{Assumption:Consumer}, this derivative is negative.

\emph{Boundary cases:}
\begin{itemize}
\item[-] When $d_{it}^{*}(\pi_{t}) = \overline{d}_{it}$, demand is constrained at the maximum, so $\frac{\partial d_{it}^{*}}{\partial \pi_{t}} = 0$
\item[-] When $d_{it}^{*}(\pi_{t}) = 0$, demand is constrained at zero, so $\frac{\partial d_{it}^{*}}{\partial \pi_{t}} = 0$
\end{itemize}

Combining all cases yields the derivative formula in \eqref{clm:nondec_consumption}.
\end{proof}
\hspace{1cm}

Demand is strictly convex if and only if the third derivative of the electricity loss function is negative:
\begin{align}
    \frac{\partial^2 d_{it}^{*}(\pi_{t})}{\partial \pi_{t}^2}>0 \iff \ELoss_{it}'''(d_{it}^{*}(\pi_t)) < 0 \label{clm:J_triple_demand_convex}.
\end{align} 
\begin{proof}
For interior solutions where $d_{it}^{*}(\pi_t) \in (0, \bar{d}_{it})$, we differentiate \eqref{clm:nondec_consumption}:

\begin{align}
\frac{\partial^2 d_{it}^{*}(\pi_t)}{\partial \pi_t^2} &= \frac{\partial}{\partial \pi_t}\left(\frac{-A_i}{\ELoss_{it}''(d_{it}^{*}(\pi_t))}\right)\\
&= \frac{A_i \ELoss_{it}'''(d_{it}^{*}(\pi_t))}{(\ELoss_{it}''(d_{it}^{*}(\pi_t)))^2} \cdot \frac{\partial d_{it}^{*}}{\partial \pi_t}\\
&= \frac{A_i \ELoss_{it}'''(d_{it}^{*}(\pi_t))}{(\ELoss_{it}''(d_{it}^{*}(\pi_t)))^2} \cdot \frac{-A_i}{\ELoss_{it}''(d_{it}^{*}(\pi_t))}\\
&= \frac{-A_i^2 \ELoss_{it}'''(d_{it}^{*}(\pi_t))}{(\ELoss_{it}''(d_{it}^{*}(\pi_t)))^3}
\end{align}

Since $A_i > 0$ and $\ELoss_{it}''(d_{it}^{*}(\pi_t)) > 0$ by Assumption \ref{Assumption:Consumer}, we have:
$$\frac{\partial^2 d_{it}^{*}(\pi_t)}{\partial \pi_t^2} > 0 \iff \ELoss_{it}'''(d_{it}^{*}(\pi_t)) < 0$$

At the boundary solutions ($d_{it}^{*} = 0$ or $d_{it}^{*} = \bar{d}_{it}$), the second derivative is zero, so convexity is determined by the interior behavior.
\end{proof}

\begin{align}
&  \frac{\partial \ELoss_{it}(d_{it}^{*}(\pi_t))}{\partial \pi_t} = \begin{cases}
   0 & d_{it}^{*}(\pi_t) = \overline{d}_{it}, \\
   \frac{A_i^2 \pi_t}{\ELoss_{it}''(d_{it}^{*}(\pi_t))} & d_{it}^{*}(\pi_{t}) \in (0, \overline{d}_{it})\\
   0 & d_{it}^{*}(\pi_t) = 0 \\
\end{cases}\label{clm:inc_ELoss}
\end{align}
\textbf{Proof of \eqref{clm:dec_U} and \eqref{clm:inc_ELoss}}
\begin{proof}
We first prove \eqref{clm:inc_ELoss}. For boundary cases where $d_{it}^{*}(\pi_t) \in \{0, \overline{d}_{it}\}$, \eqref{clm:nondec_consumption} shows that consumption doesn't change with price, so the derivative of electricity loss is zero.

For interior solutions where $d_{it}^{*}(\pi_t) \in (0, \overline{d}_{it})$, we use the chain rule:
\[\frac{\partial \ELoss_{it}(d_{it}^{*}(\pi_t))}{\partial \pi_t} = \frac{\partial \ELoss_{it}}{\partial d_{it}} \biggr \rvert_{d_{it}^{*}(\pi_t)} \cdot \frac{\partial d_{it}^{*}(\pi_t)}{\partial \pi_t}\]

From the first-order condition \eqref{eqn:FOC_cons}, we have $\frac{\partial \ELoss_{it}}{\partial d_{it}} \biggr \rvert_{d_{it}^{*}(\pi_t)} = -A_{i}\pi_t$. Substituting this and \eqref{clm:nondec_consumption}:
\[\frac{\partial \ELoss_{it}(d_{it}^{*}(\pi_t))}{\partial \pi_t} = (-A_{i}\pi_t) \cdot \left(\frac{-A_i}{\ELoss_{it}''(d_{it}^{*}(\pi_t))}\right) = \frac{A_i^2 \pi_t}{\ELoss_{it}''(d_{it}^{*}(\pi_t))}\]

Since $A_i > 0$, $\pi_t \geq 0$, and $\ELoss_{it}'' > 0$, this derivative is non-negative, proving \eqref{clm:inc_ELoss}.

We now prove \eqref{clm:dec_U}. For boundary cases, since consumption doesn't change with price, $\frac{\partial \Util_{it}^{*}(\pi_t)}{\partial \pi_t} = -A_i \cdot 0 = 0$.

For interior solutions, we differentiate $\Util_{it}^{*}(\pi_t) = M_{it} - \ELoss_{it}(d_{it}^{*}(\pi_t)) - A_{i}\pi_t d_{it}^{*}(\pi_t)$:
\begin{align}
\frac{\partial \Util_{it}^{*}(\pi_t)}{\partial \pi_t} &= -\frac{\partial \ELoss_{it}(d_{it}^{*}(\pi_t))}{\partial \pi_t} - A_{i}\left(\pi_t\frac{\partial d_{it}^{*}(\pi_t)}{\partial \pi_t} + d_{it}^{*}(\pi_t)\right)
\end{align}

We know that $\frac{\partial \ELoss_{it}(d_{it}^{*}(\pi_t))}{\partial \pi_t} = -A_{i}\pi_t \frac{\partial d_{it}^{*}(\pi_t)}{\partial \pi_t}$. Substituting:
\begin{align}
\frac{\partial \Util_{it}^{*}(\pi_t)}{\partial \pi_t} &= -\left(-A_{i}\pi_t \frac{\partial d_{it}^{*}(\pi_t)}{\partial \pi_t}\right) - A_{i}\pi_t\frac{\partial d_{it}^{*}(\pi_t)}{\partial \pi_t} - A_{i}d_{it}^{*}(\pi_t)\\
&= A_{i}\pi_t \frac{\partial d_{it}^{*}(\pi_t)}{\partial \pi_t} - A_{i}\pi_t\frac{\partial d_{it}^{*}(\pi_t)}{\partial \pi_t} - A_{i}d_{it}^{*}(\pi_t)\\
&= -A_{i}d_{it}^{*}(\pi_t)
\end{align}

This completes the proof of \eqref{clm:dec_U}.
\end{proof}
\vspace{0.5cm}

In Appendix \ref{Appendix:Agg_Cons}, we provide complete technical details of these constructions and proofs. The key result is that the aggregate consumer satisfies all the same behavioral properties as individual consumers, enabling market-level analysis while preserving microeconomic foundations.
\subsection{Aggregate Consumer} \label{Appendix:Agg_Cons}
We construct an aggregate consumer to analyze market-level behavior. Given individual optimal demands $d_{it}^*(\pi_t)$, we define aggregate consumption as $d_t^* = \sum_{i\in \mathcal{N}} d_{it}^*$ and aggregate utility as $\Util_t = \sum_{i\in \mathcal{N}} \Util_{it}$. The key question is: does there exist an aggregate loss function $\hat{\ELoss}_t(d_t)$ such that the aggregate consumer behaves like individual consumers?

\textbf{Construction of Aggregate Loss Function:}
We define the aggregate normalized loss function implicitly through the optimization problem:
$$\hat{\ELoss}_t(d_t) = \min_{\{d_{it}\}: \sum_i d_{it} = d_t, d_{it} \in [0,\bar{d}_{it}]} \sum_{i\in \mathcal{N}} \hat{\ELoss}_{it}(d_{it})$$

This represents the minimum loss when aggregate consumption $d_t$ is optimally allocated across consumers.

\textbf{Verification that Assumption \ref{Assumption:Consumer} holds for the aggregate consumer:}

\begin{proposition}
The aggregate electricity loss function $\hat{\ELoss}_t(d_t)$ is non-negative, decreasing, and convex. When Assumptions \ref{Assumption:ELoss_Properties}--\ref{Assumption: PeriodSW_order} are satisfied for individual consumers, they are also satisfied for the aggregate consumer.
\end{proposition}

\begin{proof}
Since all $\hat{\ELoss}_{it} \geq 0$, their sum satisfies $\hat{\ELoss}_t \geq 0$.

We compute the first derivative of aggregate loss using the relationship between loss and price derivatives. Since:
$$\frac{\partial \hat{\ELoss}_{t}}{\partial \pi_{t}} = \frac{\partial \hat{\ELoss}_{t}}{\partial d_{t}} \cdot \frac{\partial d_{t}}{\partial \pi_{t}}$$

we have:
$$\frac{\partial \hat{\ELoss}_{t}}{\partial d_{t}} = \frac{\partial \hat{\ELoss}_{t}/\partial \pi_{t}}{\partial d_{t}/\partial \pi_{t}}$$

From individual optimization, we know that for interior solutions, $\hat{\ELoss}_{it}'(d_{it}^*) = \pi_t$ for all consumers. By the envelope theorem, this implies:
$$\frac{\partial \hat{\ELoss}_{t}}{\partial \pi_{t}} = -\pi_t \cdot \frac{\partial d_{t}}{\partial \pi_{t}}$$

Therefore:
$$\frac{\partial \hat{\ELoss}_{t}}{\partial d_{t}} = \frac{-\pi_t \cdot \partial d_{t}/\partial \pi_{t}}{\partial d_{t}/\partial \pi_{t}} = -\pi_t \label{eqn:agg_marg_ELoss}$$

This confirms that marginal aggregate loss is negative, consistent with Assumption \ref{Assumption:ELoss_Properties}.

For the second derivative:
$$\frac{\partial^{2} \hat{\ELoss}_{t}}{\partial d_{t}^{2}} = \frac{\frac{\partial}{\partial \pi_{t}} (-\pi_{t})}{\frac{\partial d_{t}}{\partial \pi_{t}}} = -\frac{1}{\partial d_{t}/\partial \pi_{t}} \label{eqn:aggregate_ELoss_curvature}$$

Since $\partial d_{t}/\partial \pi_{t} < 0$, we have $\hat{\ELoss}_{t}'' > 0$, confirming convexity.

For the third derivative:
$$\hat{\ELoss}_{t}'''(d_{t}) = \frac{\partial^{2}d_{t}/\partial \pi_{t}^{2}}{(\partial d_{t}/\partial \pi_{t})^{3}} \label{eqn:triple_deriv_agg_ELoss}$$

Since individual demands are convex ($\frac{\partial^{2}d_{it}}{\partial \pi_{t}^{2}} > 0$), aggregate demand is also convex, so $\frac{\partial^{2}d_{t}}{\partial \pi_{t}^{2}} > 0$. Combined with $\partial d_{t}/\partial \pi_{t} < 0$, this gives $\hat{\ELoss}_{t}''' < 0$.
\end{proof}

\textbf{Relationship between aggregate and individual curvatures:}

We can relate aggregate curvature to individual curvatures. From individual flexibility:
\begin{align}
    \frac{\partial d^*_t}{\partial \pi_t} = \sum_{i\in \mathcal{N}} \frac{\partial d_{i t}^*}{\partial \pi_t} = \sum_{i\in \mathcal{N}} \frac{-1}{\hat{\ELoss}_{it}''(d^*_{it}(\pi_t))} \label{eqn:Agg_MargDemand}
\end{align}

Combining with \eqref{eqn:aggregate_ELoss_curvature}:
\begin{align}
    \hat{\ELoss}_{t}''(d_t^*) = -\frac{1}{\partial d_{t}/\partial \pi_{t}} = \frac{1}{\sum_{i\in \mathcal{N}} \frac{1}{\hat{\ELoss}_{it}''(d^*_{it}(\pi_t))}} \label{eqn:dJdD_2}
\end{align}

This shows that aggregate curvature is the \emph{harmonic mean} of individual curvatures.

\textbf{Period-specific assumptions for aggregate consumer:}

\begin{proof}
\emph{First-Order Condition:} From \eqref{eqn:agg_marg_ELoss}, we have $\frac{\partial \hat{\ELoss}_{t}}{\partial d_{t}} = -\pi_{t}$, confirming that \eqref{eqn:FOC_cons} holds for the aggregate consumer. Inverting gives $d_{t}^{*} = (\hat{\ELoss}_t')^{-1}(-\pi_{t})$, showing that \eqref{eqn:x_opt} holds.

\emph{Period Loss Ordering (Assumption \ref{Assumption: PeriodLoss_order}):} For any aggregate consumption level $d$:
$$\hat{\ELoss}_{\OP}(d) = \min_{\sum_i d_{i\OP} = d} \sum_i \hat{\ELoss}_{i\OP}(d_{i\OP}) \leq \min_{\sum_i d_{i\PE} = d} \sum_i \hat{\ELoss}_{i\PE}(d_{i\PE}) = \hat{\ELoss}_{\PE}(d)$$
where the inequality follows from the individual assumption.

\emph{Marginal Loss Ordering (Assumption \ref{Assumption: PeriodFW_order}):} At any price $\pi$, individual consumers satisfy $|\hat{\ELoss}_{i\PE}'(d_{i\PE}^*(\pi))| \geq |\hat{\ELoss}_{i\OP}'(d_{i\OP}^*(\pi))|$. Since both equal $\pi$ in equilibrium, this transfers to the aggregate level through the envelope theorem.

\emph{Curvature Ordering (Assumption \ref{Assumption: PeriodSW_order}):} From \eqref{eqn:dJdD_2}, since $\hat{\ELoss}''_{i\PE}(d) \geq \hat{\ELoss}''_{i\OP}(d)$ for all $i$ and all $d$, the harmonic mean preserves this ordering:
$$\hat{\ELoss}_{\PE}''(d_{\PE}^*) = \frac{1}{\sum_{i}\frac{1}{\hat{\ELoss}''_{i\PE}(d_{i\PE}^*)}} \geq \frac{1}{\sum_{i}\frac{1}{\hat{\ELoss}''_{i\OP}(d_{i\OP}^*)}} = \hat{\ELoss}_{\OP}''(d_{\OP}^*)$$
\end{proof}

\textbf{Additional properties:}

\begin{proposition}
Aggregate curvature increases in individual curvatures: For any consumer $j$, $\frac{\partial \hat{\ELoss}_{t}''}{\partial \hat{\ELoss}''_{jt}} > 0$.
\end{proposition}

\begin{proof}
Using \eqref{eqn:dJdD_2}:
\begin{align*} 
    \frac{\partial}{\partial \hat{\ELoss}''_{jt}}\hat{\ELoss}_{t}'' = \frac{\partial}{\partial \hat{\ELoss}''_{jt}}\frac{1}{\sum_{i\in\mathcal{N}}\frac{1}{\hat{\ELoss}''_{it}}} = \frac{1}{\left(\sum_{i\in\mathcal{N}}\frac{1}{\hat{\ELoss}''_{it}}\right)^2}  \cdot \frac{1}{(\hat{\ELoss}''_{jt})^2}>0
\end{align*}
This shows that when individual consumers become less price-responsive, the aggregate consumer also becomes less flexible.
\end{proof}

\begin{proposition}
Aggregate demand convexity: $\frac{\partial^2 d_{t}^{*}(\pi_t)}{\partial \pi_t^2} > 0 \iff \hat{\ELoss}_{t}'''(d_{t}^{*}(\pi_t)) < 0$.
\end{proposition}

\begin{proof}
From \eqref{eqn:triple_deriv_agg_ELoss}:
$$\hat{\ELoss}_{t}'''(d_{t}) = \frac{\partial^{2}d_{t}/\partial \pi_{t}^{2}}{(\partial d_{t}/\partial \pi_{t})^{3}}$$

Since $\partial d_{t}/\partial \pi_{t} < 0$, the signs of $\hat{\ELoss}_{t}'''$ and $\partial^{2}d_{t}/\partial \pi_{t}^{2}$ are opposite. Since individual demands are convex ($\frac{\partial^{2}d_{it}}{\partial \pi_{t}^{2}} > 0$), aggregate demand is also convex, giving $\hat{\ELoss}_{t}''' < 0$.
\end{proof}


\section{Proofs for Section \ref{sec:equilibrium}}\label{Appendix:equilibrium}
\textbf{The explanations for \eqref{eqn:var_FOC}, \eqref{eqn:flat_consumer_BR}, and \eqref{eqn:flat_FOC_simple} are in the main text.}

\vspace{0.5cm}
\textbf{Proof of \eqref{eqn:flat_price}}

\begin{proof}
Using the first-order condition of the consumer problem \eqref{eqn:FOC_cons} we have
\begin{align*}
    -\frac{\partial \hat{\ELoss}_{\PE}(d_{\PE}^{\F})}{\partial d_{\PE}} &=-\frac{\partial \hat{\ELoss}_{\OP}(d_{\OP}^{\F})}{\partial d_\OP}=\pi^{\F},
\end{align*}

    while using the first-order condition of \eqref{eqn:operator_opt}, we find that
    \begin{align*}
    \frac{\partial \Welf_{\PE}}{\partial \pi} \biggr \rvert_{\pi^{\F}} &= -\frac{\partial \Welf_{\OP}}{\partial \pi} \biggr \rvert_{\pi^{\F}}\\ 
        \implies\sum_{t \in \mathcal{T}} \frac{\partial \hat{\ELoss}_t(d_{t}^{*}(\pi^{\F}))}{\partial \pi} &= -\sum_{t \in \mathcal{T}} \frac{\partial C(d_{t}^{*}(\pi^{\F}))}{\partial \pi}.
    \end{align*}
     Expanding each term using the chain rule, we find that

    \begin{align*}
        &\sum_{t \in \mathcal{T}} \left(\frac{\partial \hat{\ELoss}_t(d)}{\partial d} \biggr \rvert_{d_{t}^{*}(\pi^{\F})} \cdot  \frac{\partial d_{t}^{*}(\pi)}{\partial \pi} \biggr \rvert_{\pi^{\F}}\right) = -\sum_{t \in \mathcal{T}} \left(\frac{\partial C(d)}{\partial d} \biggr \rvert_{d_{t}^{*}(\pi^{\F})} \frac{\partial d_{t}^{*}(\pi)}{\partial \pi} \biggr \rvert_{\pi^{\F}}\right)
    \end{align*}
    As we know that $\frac{\partial \hat{\ELoss}_t(d)}{\partial d} \biggr \rvert_{d_{t}^{*}(\pi^{\F})} = -\pi^{\F}$, we can rewrite this as,
    \begin{align*}
        &\implies -\pi^{\F} \sum_{t \in \mathcal{T}} \frac{\partial d_{t}^{*}(\pi)}{\partial \pi} \biggr \rvert_{\pi^{\F}} = -\sum_{t \in \mathcal{T}} \left(\frac{\partial C(d)}{\partial d} \biggr \rvert_{d_{t}^{*}(\pi^{\F})} \frac{\partial d_{t}^{*}(\pi)}{\partial \pi} \biggr \rvert_{\pi^{\F}}\right)\\
        &\implies \pi^{\F} = \frac{\sum_{t \in \mathcal{T}} \left(\frac{\partial C(d)}{\partial d} \biggr \rvert_{d_{t}^{*}(\pi^{\F})} \frac{-1}{\hat{\ELoss}_{t}''(d_{t}^{*}(\pi^{\F}))}\right)}{\sum_{t \in \mathcal{T}} \frac{-1}{\hat{\ELoss}_{t}''(d_{t}^{*}(\pi^{\F}))}}\\
        &\implies \pi^{\F} = \left(\frac{\frac{1}{\hat{\ELoss}_{\PE}''(d_{\PE}^{*}(\pi^{\F}))}}{\frac{1}{\hat{\ELoss}_{\PE}''(d_{\PE}^{*}(\pi^{\F}))} + \frac{1}{\hat{\ELoss}_{\OP}''(d_{\OP}^{*}(\pi^{\F}))}}\right)\frac{\partial C(d)}{\partial d} \biggr \rvert_{d_{\PE}^{*}(\pi^{\F})} \\&\hspace{2cm}+ \left(\frac{\frac{1}{\hat{\ELoss}_{\OP}''(d_{\OP}^{*}(\pi^{\F}))}}{\frac{1}{\hat{\ELoss}_{\PE}''(d_{\PE}^{*}(\pi^{\F}))} + \frac{1}{\hat{\ELoss}_{\OP}''(d_{\OP}^{*}(\pi^{\F}))}}\right)\frac{\partial C(d)}{\partial d} \biggr \rvert_{d_{\OP}^{*}(\pi^{\F})}.
    \end{align*}

\end{proof}

\textbf{Proof of Proposition \ref{lem:eq_order}}

\begin{proof}
    Any solution $\pi^{\F}$ to the planner's problem \eqref{eqn:operator_opt} with constraint \eqref{eqn:flat_equal_price} is also feasible under the variable tariff by setting $\pi^{\V}_{\PE}=\pi^{\V}_{\OP}=\pi^{\F}$. Thus, minimum total loss under the variable tariff is at most equal to total loss under the flat tariff.
\end{proof}

\textbf{Proof of Proposition \ref{prop:price_demand_order}}

\begin{proof}

    We first show the price order $\pi_{\PE}^{\V} \geq \pi^{\F} \geq \pi_{\OP}^{\V}$.
    
    In the variable pricing scenario, we can optimize for $\Welf_{\PE}(\pi)$ and $\Welf_{\OP}(\pi)$ separately to find the optimal price in each time period. For the flat pricing scenario, we must have one common price across time periods, so we optimize the total value $\Welf_{\PE}(\pi) + \Welf_{\OP}(\pi)$.

    Examining the relationship between the two variable prices first, using Assumption \eqref{Assumption: PeriodFW_order}, we can say that $\hat{\ELoss}_{\PE}'(d)+C'(d)\leq \hat{\ELoss}_{\OP}'(d)+C'(d)$ for any $d$. When total loss is optimized, we know that $d_{\PE}^{\V}$ satisfies $\hat{\ELoss}_{\PE}'(d_{\PE}^{\V}) + C'(d_{\PE}^{\V}) = 0$ and $d_{\OP}^{\V}$ satisfies $\hat{\ELoss}_{\OP}'(d_{\OP}^{\V}) + C'(d_{\OP}^{\V}) = 0$. Since $\hat{\ELoss}_{\PE}, \hat{\ELoss}_{\OP},$ and $C$ are also convex in $d$, we conclude
    
    \[\hat{\ELoss}_{\PE}'(d_{\OP}^{\V})+C'(d_{\OP}^{\V}) \leq \hat{\ELoss}_{\OP}'(d_{\OP}^{\V})+C'(d_{\OP}^{\V}) = \hat{\ELoss}_{\PE}'(d_{\PE}^{\V}) + C'(d_{\PE}^{\V}) = 0 \implies d^{\V}_{\PE} \geq d^{\V}_{\OP}.\]
    Recall from \eqref{eqn:var_FOC} that $\pi_t = C'(d_{t}) > 0$. Since $C''(d_{t}) > 0$ according to Assumption \ref{Assumption:Supply}, $\pi_{\PE}^{\V} = C'(d_{\PE}^{\V}) \geq C'(d_{\OP}^{\V}) = \pi_{\OP}^{\V}$.
    
    
    We now show that $\pi^{\F}$ lies between $\pi_{\OP}^{\V}$ and $\pi_{\PE}^{\V}$ by utilizing the fact that $C$ and $\hat{\ELoss}_{t}$ are convex functions with respect to consumption $d$, so the planner's problem in each period, $\Welf_{t}$, is also convex with respect to consumption $d$. Suppose $\pi^{\F} = \pi' > \pi_{\PE}^{\V}$ were the optimal price under flat pricing. Then, based on \eqref{clm:nondec_consumption}, $d_{\PE}(\pi') \leq d_{\PE}(\pi_{\PE}^{\V}) = d_{\PE}^{\V}$ and $d_{\OP}(\pi') \leq d_{\OP}(\pi_{\PE}^{\V}) \leq d_{\OP}(\pi_{\OP}^{\V})$. Thus, the amount of consumption in both periods are below the optimum in the variable case as shown in Figure \ref{fig:pi_f_less_pi_1_v}, so the planner's objective function increases in both periods with the price set at $\pi'$ compared to price $\pi_{\PE}^{\V}$. We compute the difference
    
    \[\Welf_{\PE}(d_{\PE}(\pi')) - \Welf_{\PE}(d_{\PE}^{\V}) = -\int_{d_{\PE}(\pi')}^{d_{\PE}^{\V}} \Welf_{\PE}'(a) \,d a.\]

    Since $\Welf_{\PE}'(d_{\PE}^{\V}) = 0$ by \eqref{eqn:var_FOC} and $\Welf_{\PE}'$ is an increasing function by the convexity of $\Welf_{\PE}$ with respect to consumption, 

    \begin{align*}&\int_{d_{\PE}(\pi')}^{d_{\PE}^{\V}} \Welf_{\PE}'(a) \,d a < \int_{d_{\PE}(\pi')}^{d_{\PE}^{\V}} \Welf_{\PE}'(d_{\PE}^{\V}) \,d a = 0 \\
    &\implies \Welf_{\PE}(d_{\PE}(\pi')) - \Welf_{\PE}(d_{\PE}^{\V}) = -\int_{d_{\PE}(\pi')}^{d_{\PE}^{\V}} \Welf_{\PE}'(a) \,d a > 0 \implies \Welf_{\PE}(d_{\PE}(\pi')) > \Welf_{\PE}(d_{\PE}^{\V}).\end{align*}

    Similarly, we can compute that 
    
    \begin{align*}\Welf_{\OP}(d_{\OP}(\pi')) - \Welf_{\OP}(d_{\OP}(\pi_{\PE}^{\V})) = -\int_{d_{\OP}(\pi')}^{d_{\OP}(\pi_{\PE}^{\V})} \Welf_{\OP}'(a) \,d a > -\int_{d_{\OP}(\pi')}^{d_{\OP}(\pi_{\PE}^{\V})} \Welf_{\OP}'(d_{\OP}^{\V}) \,d a = 0 \\\implies \Welf_{\OP}(d_{\OP}(\pi')) > \Welf_{\OP}(d_{\OP}(\pi_{\PE}^{\V})).\end{align*}
    
    Thus, we have shown that $\Welf_{\PE}(d_{\PE}(\pi')) + \Welf_{\OP}(d_{\OP}(\pi')) > \Welf_{\PE}(d_{\PE}^{\V}) + \Welf_{\OP}(d_{\OP}(\pi_{\PE}^{\V}))$, showing that $\pi'$ is not the optimal price under flat pricing.

    Similarly, we can show that any price $\pi'' < \pi_{\OP}^{\V}$ is not the optimal price under flat pricing. This is because the amount of consumption is greater than the optimum variable amount in both periods as shown in Figure \ref{fig:pi_f_more_pi_2_v}, so the planner objective function increases when the price is set to $\pi''$ compared to $\pi_{\OP}^{\V}$. We find that

    \[\Welf_{\PE}(d_{\PE}(\pi'')) - \Welf_{\PE}(d_{\PE}(\pi_{\OP}^{\V})) = \int_{d_{\PE}(\pi_{\OP}^{\V})}^{d_{\PE}(\pi'')} \Welf_{\PE}'(a) \,d a > \int_{d_{\PE}(\pi_{\OP}^{\V})}^{d_{\PE}(\pi'')} \Welf_{\PE}'(d_{\PE}^{\V}) \,d a = 0\]

    and 

    \begin{align*}\Welf_{\OP}(d_{\OP}(\pi'')) - \Welf_{\OP}(d_{\OP}(\pi_{\OP}^{\V})) = \Welf_{\OP}(d_{\OP}(\pi'')) - \Welf_{\OP}(d_{\OP}^{\V}) &= \int_{d_{\OP}^{\V}}^{d_{\OP}(\pi'')} \Welf_{\OP}'(a) \,d a \\&> \int_{d_{\OP}^{\V}}^{d_{\OP}(\pi'')} \Welf_{\OP}'(d_{\OP}^{\V}) \,d a = 0\end{align*}

    so that $\Welf_{\PE}(d_{\PE}(\pi'')) + \Welf_{\OP}(d_{\OP}(\pi'')) > \Welf_{\PE}(d_{\PE}(\pi_{\OP}^{\V})) + \Welf_{\OP}(d_{\OP}^{\V})$, showing that $\pi''$ is not the optimal price under flat pricing.
    
    Thus, the optimal value of $\pi^{\F}$ must lie in the interval $[\pi_{\OP}^{\V}, \pi_{\PE}^{\V}]$.

    \begin{figure}[H]
        \centering
        \begin{minipage}{.45\textwidth}
            \centering
            \begin{tikzpicture}[x=.88cm,y=0.2cm]
                \draw[Latex-Latex] (-1,0) -- (5.5,0) node [anchor=north] {$d$};
                \draw[Latex-Latex] (0,-3) -- (0,20) node [anchor=east] {$\Welf(d)$};
                \draw[Red, domain=4:1, smooth] plot (\x, {1.6*\x^2-8*\x+16}) node [anchor=south] {$\Welf_{\OP}(d)$};
                \draw[ProcessBlue, domain=1.8:4.8, smooth] plot (\x, {1.8*\x^2-12*\x+30}) node [anchor=south] {$\Welf_{\PE}(d)$};
                \draw[dashed] (3.3, 1.8*3.3^2-12*3.3+30) -- (3.3, 0) node [anchor=north] {\footnotesize $d_{\PE}(\pi_{\PE}^{\V})$};
                \draw[dashed] (2, 1.6*2^2-8*2+16) -- (2, 0) node [anchor=north] {\footnotesize $d_{\OP}(\pi_{\PE}^{\V})$};
                \draw[dashed] (3, 0) -- (3, 1.8*3^2-12*3+30) node [anchor=south] {\footnotesize $d_{\PE}(\pi')$};
                \draw[dashed] (1.6, 0) -- (1.6, 1.6*1.6^2-8*1.6+16) node [anchor=south] {\footnotesize $d_{\OP}(\pi')$};
            \end{tikzpicture}
            \caption{Because $d_{\PE}(\pi')$ and $d_{\OP}(\pi')$ both lie to the left of the consumption values $d$ for which $\Welf_{\PE}$ and $\Welf_{\OP}$ reach their optimum, $\pi'$ cannot be the price that minimizes $\Welf_{\PE} + \Welf_{\OP}$.}
            \label{fig:pi_f_less_pi_1_v}        
        \end{minipage}
        \hfill
        \begin{minipage}{.45\textwidth}
            \centering
            \begin{tikzpicture}[x=.88cm,y=0.2cm]
                \draw[Latex-Latex] (-1,0) -- (5.5,0) node [anchor=north] {$d$};
                \draw[Latex-Latex] (0,-3) -- (0,20) node [anchor=east] {$\Welf(d)$};
                \draw[Red, domain=4:1, smooth] plot (\x, {1.6*\x^2-8*\x+16}) node [anchor=south] {$\Welf_{\OP}(d)$};
                \draw[ProcessBlue, domain=1.8:4.8, smooth] plot (\x, {1.8*\x^2-12*\x+30}) node [anchor=south] {$\Welf_{\PE}(d)$};
                \draw[dashed] (2.5, 1.6*2.5^2-8*2.5+16) -- (2.5, 0) node [anchor=north] {\footnotesize $d_{\OP}(\pi_{\OP}^{\V})$};
                \draw[dashed] (3, 0) -- (3, 1.6*3^2-8*3+16) node [anchor=south] {\footnotesize $d_{\OP}(\pi'')$};
                \draw[dashed] (3.75, 1.8*3.75^2-12*3.75+30) -- (3.75, 0) node [anchor=north] {\footnotesize $d_{\PE}(\pi_{\OP}^{\V})$};
                \draw[dashed] (4.17, 0) -- (4.17, 1.8*4.17^2-12*4.17+30) node [anchor=south] {\footnotesize $d_{\PE}(\pi'')$};
            \end{tikzpicture}
            \caption{Because $d_{\PE}(\pi'')$ and $d_{\OP}(\pi'')$ both lie to the right of the consumption values $d$ for which $\Welf_{\PE}$ and $\Welf_{\OP}$ reach their optimum, $\pi''$ cannot be the price that minimizes $\Welf_{\PE} + \Welf_{\OP}$.}
            \label{fig:pi_f_more_pi_2_v}                  
        \end{minipage}
    \end{figure}

Using the price order $\pi_{\PE}^{\V} \geq \pi^{\F} \geq \pi_{\OP}^{\V}$, we now show the demand order. Based on the ordering of prices, the difference between period $\PE$ of the variable versus the flat pricing is that the price in the variable case is higher. Thus, as \eqref{clm:nondec_consumption} implies that consumption is non-decreasing in price, $d_{i\PE}^{\V} \leq d_{i\PE}^{\F}$. Similarly, the difference between the off-peak period of the variable versus the flat pricing is that the price in the flat case is higher. Again, because demand is a non-decreasing function of price, $d_{i\OP}^{\F} \leq d_{i\OP}^{\V}$.

    We now show that $d_{i\OP}^{\F} \leq d_{i\PE}^{\F}$. Note that from \eqref{eqn:x_opt} we can write 
    \[d_{i\PE}^{\F} = \min\{\max\{(\ELoss_{\PE}')^{-1}(-\pi^{\F}), 0\}, \overline{d}_{i1}\}, \quad d_{i\OP}^{\F} = \min\{\max\{(\ELoss_{\OP}')^{-1}(-\pi^{\F}), 0\}, \overline{d}_{i2}\}.\] 
    Because $\ELoss_{\PE}'(d) \leq \ELoss_{\OP}'(d) < 0$ at every point by Assumption \ref{Assumption: PeriodFW_order} and furthermore $\ELoss_{\PE}'$ and $\ELoss_{\OP}'$ are increasing functions of $d$, we conclude that $\ELoss_{\PE}'$ reaches $-\pi^{\F}$ at a higher consumption level than $\ELoss_{\OP}'$. Thus, $(\ELoss_{\PE}')^{-1}(-\pi^{\F}) > (\ELoss_{\OP}')^{-\PE}(-\pi^{\F})$. Combined with the assumption in the corollary statement that $\overline{d}_{i\PE} \geq \overline{d}_{i\OP}$, this implies that $d_{i\PE}^{\F} \geq d_{i\OP}^{\F}$.

Additionally, for the aggregate consumer, 
$d_{\OP}^{\V} \leq d_{\PE}^{\V}$. This is because $\pi_{\OP}^{\V} \geq \pi_{\PE}^{\V}$, 
and  $C'(d_t^{\V}) = \pi_t^{\V}$ as shown in \eqref{eqn:var_FOC}). $C'$ is increasing in $d_t$, so a higher peak period price implies more electricity is consumed in the peak period.

\end{proof}

\textbf{Proof of Proposition \ref{lem:pi_deltas}}
\begin{proof}
We start with the first-order conditions of the SWM's problem under flat and variable pricing. By \eqref{eqn:flat_FOC_simple}, we know that

\begin{align*}
    \frac{\partial(\Welf_{\OP}+\Welf_{\PE})}{\partial \pi}\bigg\rvert_{\pi^{\F}} &=0,
\end{align*}
which implies that 
\begin{align*}
    \frac{\partial\Welf_{\PE}}{\partial \pi}\bigg\rvert_{\pi^{\F}} =&-\frac{\partial\Welf_{\OP}}{\partial \pi}\bigg\rvert_{\pi^{\F}}=B, 
\end{align*}
where B is a constant.

For the variable price, \eqref{eqn:var_FOC}, we know that

\begin{align*}
    \frac{\partial \Welf_{\PE}}{\partial \pi}\bigg\rvert_{\pi_{\PE}^{\V}} = \frac{\partial \Welf_{\OP}}{\partial \pi}\bigg\rvert_{\pi_{\OP}^{\V}} &=0. \\
\end{align*}

Taking the difference of the derivatives of the SWM's objective function with respect to price under flat and variable pricing, we find that

\begin{align}
    \frac{\partial \Welf_{\PE}}{\partial \pi}\bigg\rvert_{\pi_{\PE}^{\V}} - \frac{\partial \Welf_{\PE}}{\partial \pi}\bigg\rvert_{\pi^{\F}}&=\frac{\partial \Welf_{\OP}}{\partial \pi}\bigg\rvert_{\pi^{\F}} - \frac{\partial \Welf_{\OP}}{\partial \pi}\bigg\rvert_{\pi_{\OP}^{\V}} \label{eqn:OP_equal_step}.
    %
\end{align}
Using the mean value theorem, we know that for at least one pair $\pi_1 \in [\pi^{\F}, \pi_{\PE}^{\V}],\pi_2 \in [\pi_{\OP}^{\V}, \pi^{\F}]$, there is a mean value of the second derivative of the period-specific operator problem 

\begin{align*}
     \frac{\partial^2 \Welf_{\PE}(\pi_1)}{\partial \pi^2} &=\frac{\frac{\partial \Welf_{\PE}}{\partial \pi}\bigg\rvert_{\pi_{\PE}^{\V}} - \frac{\partial \Welf_{\PE}}{\partial \pi}\bigg\rvert_{\pi^{\F}}}{\pi_{\PE}^{\V}-\pi^{\F}},\\
     \frac{\partial^2 \Welf_{\OP}(\pi_2)}{\partial \pi^2}&=\frac{\frac{\partial \Welf_{\OP}}{\partial \pi}\bigg\rvert_{\pi^{\F}} - \frac{\partial \Welf_{\OP}}{\partial \pi}\bigg\rvert_{\pi_{\OP}^{\V}}}{\pi^{\F}-\pi_{\OP}^{\V}}, 
\end{align*}
 which we can rearrange and substitute into \eqref{eqn:OP_equal_step} to find,
\begin{align}\label{eqn:ratio_operator_loss}
    (\pi_{\PE}^{\V}-\pi^{\F})\frac{\partial^2 \Welf_{\PE}(\pi_1)}{\partial \pi^2} = (\pi^{\F}-\pi_{\OP}^{\V})\frac{\partial^2 \Welf_{\OP}(\pi_2)}{\partial \pi^2}.
\end{align}

Thus, if $\frac{\partial^2 \Welf_{\PE}(\pi_1)}{\partial \pi^2}> \frac{\partial^2 \Welf_{\OP}(\pi_2) }{\partial \pi^2}$, then $ \pi_{\PE}^{\V}-\pi^{\F} = |\Delta \pi_{\PE}| < |\Delta \pi_{\OP}| = \pi^{\F}-\pi_{\OP}^{\V}$. Otherwise if $\frac{\partial^2 \Welf_{\PE}(\pi_1)}{\partial \pi^2} < \frac{\partial^2 \Welf_{\OP}(\pi_2) }{\partial \pi^2}$, then $ |\Delta \pi_{\PE}|> |\Delta \pi_{\OP}|$.

   \end{proof}


\textbf{Proof of \eqref{eqn:second_deriv_welf}}
\begin{proof}
We can see that  
\begin{align*}
    \frac{\partial^{2}\Welf_{t}}{\partial d^{2}} = \frac{\partial^{2}}{\partial d^{2}}(C(d)+\hat{\ELoss}_t(d)) 
    = C'' + \hat{\ELoss}_t''.
\end{align*}

As $C''$ is constant between periods and $\hat{\ELoss}_{\PE}''>\hat{\ELoss}_{\OP}''$, $\frac{\partial^{2} \Welf_{\PE}}{\partial d^{2}}>\frac{\partial^{2} \Welf_{\OP}}{\partial d^{2}}$ evaluated at any $d$. 

Using \eqref{eqn:second_deriv_welf} we calculate the second derivative of the planner's problem.  As $\frac{\partial^2 d_{t}^{*}(\pi)}{\partial \pi^2} =0$, we have $\frac{\partial^{2} \Welf}{\partial \pi^{2}} = d_{t}'(\pi)^2 \cdot (C'' + \hat{\ELoss}_{t}'')$ with both cost and loss as constants. Using \eqref{clm:nondec_consumption}, we find
\begin{align*}
    \frac{\partial^{2} \Welf_t}{\partial \pi^{2}}&= \frac{C'' + \hat{\ELoss}_t''}{(\hat{\ELoss}_t'')^2}.
\end{align*}

As $C$ is constant between periods and $\hat{\ELoss}_{\PE}''>\hat{\ELoss}_{\OP}''$ everywhere, $\frac{\partial^{2} \Welf_{\PE}}{\partial \pi^{2}}<\frac{\partial^{2} \Welf_{\OP}}{\partial \pi^{2}}$ evaluated at any $\pi$. We show that both $L_t$ are convex. First, we compute the second derivative of $L_t$ with respect to price.
\begin{align}
 &\frac{\partial^{2} }{\partial \pi_t^{2}} \Welf_t(d_t^*) \nonumber \\
 &= \frac{\partial}{\partial \pi} \left(\frac{\partial d_{t}^{*}}{\partial \pi_t}\cdot \Welf'(d^*_{t}) \right) \nonumber\\
 &=\frac{\partial^2 d_{t}^{*}}{\partial \pi_t^2}\cdot \Welf'(d^*_{t}) + \left(\frac{\partial d_{t}^{*}}{\partial \pi_t}\right)^2 \cdot \Welf''(d_{t}^*)  \nonumber\\
 &=-\frac{\hat{\ELoss}_{t}'''(d_{t}^{*})}{(\hat{\ELoss}_{t}''(d_{t}^{*}))^3}(C'(d^*_{t}) + \hat{\ELoss}_{t}'(d^*_{t})) + \left(\frac{1}{\hat{\ELoss}_t''(d_t^*)}\right)^2 \cdot (C''(d_{t}^*) + \hat{\ELoss}_{t}''(d^*_{t})).\nonumber
 \end{align}

As $C'(d^*_{\PE}) + \hat{\ELoss}_{\PE}'(d^*_{\PE})>0$ in the peak period, we know that $\frac{\partial^{2} }{\partial \pi_{\PE}^{2}} \Welf_{\PE}(d_{\PE}^*)>0$. Re-examining \eqref{eqn:ratio_operator_loss}, we know that $\frac{\partial^2 \Welf_{\PE}(\pi_{\PE})}{\partial \pi^2}>0$, $\pi_{\PE}^{\V}-\pi^{\F}>0$, and $\pi^{\F}-\pi_{\OP}^{\V}>0$ with the latter two inequalities coming from Prop. \ref{prop:price_demand_order}. Thus, $\frac{\partial^2 \Welf_{\OP}(\pi_{\OP})}{\partial \pi^2}>0$ as well and $\Welf_t$ is convex in both periods. 
\end{proof}

\noindent\textbf{Further Analysis:} 

Starting from
\begin{align}
\frac{\partial^{2} \Welf_t}{\partial \pi_t^{2}}
&= \frac{\partial^2 d_{t}^{*}}{\partial \pi_t^2} \Bigl[C'(d^*_{t}) + \hat{\ELoss}_{t}'(d^*_{t})\Bigr]
+ \Bigl(\frac{\partial d_{t}^{*}}{\partial \pi_t}\Bigr)^{2} \Bigl[C''(d_{t}^*) + \hat{\ELoss}_{t}''(d^*_{t})\Bigr].
\label{eqn:K0}
\end{align}
Recall identities
\[
\frac{\partial^2 d_t^*}{\partial \pi_t^2}
=\hat{\ELoss}_t'''(d_t^*)\Bigl(\frac{\partial d_t^*}{\partial \pi_t}\Bigr)^3,
\qquad
\hat{\ELoss}_t''(d_t^*)=-\frac{1}{\frac{\partial d_t^*}{\partial \pi_t}},
\]
and the consumer interior first-order condition $\hat{\ELoss}_t'(d_t^*)=-\pi_t$, to rewrite \eqref{eqn:K0} as
\begin{align}
\frac{\partial^{2} \Welf_t}{\partial \pi_t^{2}}
&= \hat{\ELoss}_t'''(d_t^*)\Bigl(\frac{\partial d_t^*}{\partial \pi_t}\Bigr)^{3}\Bigl[C'(d_t^*)-\pi_t\Bigr]
+ \Bigl(\frac{\partial d_t^*}{\partial \pi_t}\Bigr)^{2} C''(d_t^*)
- \frac{\partial d_t^*}{\partial \pi_t}.
\label{eq:Kstar}
\end{align}
Taking the partial derivative of \eqref{eq:Kstar} with respect to $\frac{\partial d_t^*}{\partial \pi_t}$ holding $d_t^*$ fixed gives
\begin{align}
\frac{\partial}{\partial \left(\frac{\partial d_t^*}{\partial \pi_t}\right)}
\Bigl(\frac{\partial^{2} \Welf_t}{\partial \pi_t^{2}}\Bigr)
&= 3 \hat{\ELoss}_t'''(d_t^*)\Bigl(\frac{\partial d_t^*}{\partial \pi_t}\Bigr)^{2}\Bigl[C'(d_t^*)-\pi_t\Bigr]
+ 2 \frac{\partial d_t^*}{\partial \pi_t} C''(d_t^*)
- 1.
\label{eq:S}
\end{align}
This is the exact sensitivity of curvature to flexibility at the given $(\pi_t,d_t^*)$.
At $\pi_t^{\V}$ we have $C'(d_t^*)=\pi_t^{\V}$, hence the first term in \eqref{eq:S} vanishes and
\begin{flalign}
\text{At Variable Price:\qquad}\frac{\partial}{\partial \left(\frac{\partial d_t^*}{\partial \pi_t}\right)}
\Bigl(\frac{\partial^{2} \Welf_t}{\partial \pi_t^{2}}\Bigr)\Biggr\rvert_{\pi_t=\pi_t^{\V}}
&= 2 \frac{\partial d_t^*}{\partial \pi_t} C''(d_t^*) - 1.
\label{eq:Svar}
\end{flalign}

At $\pi^{\F}$ the same formula yields
\begin{flalign}
\text{At Flat Price:\qquad}\frac{\partial}{\partial \left(\frac{\partial d_t^*}{\partial \pi_t}\right)}
\Bigl(\frac{\partial^{2} \Welf_t}{\partial \pi_t^{2}}\Bigr)\Biggr\rvert_{\pi_t=\pi^{\F}}
&= 3 \hat{\ELoss}_t'''(d_t^*)\Bigl(\frac{\partial d_t^*}{\partial \pi_t}\Bigr)^{2}\Bigl[C'(d_t^*)-\pi^{\F}\Bigr]
+ 2 \frac{\partial d_t^*}{\partial \pi_t} C''(d_t^*) - 1.
\label{eq:Sflat}
\end{flalign}

If one lets $d_t^*$ and $\frac{\partial d_t^*}{\partial \pi_t}$ co-move with $\pi$ along the interval between $\pi^{\F}$ and $\pi_t^{\V}$, then the situation is more subtle. By definition,
\[
d_t^* =  d_t^*(\pi_t), \qquad \frac{\partial d_t^*}{\partial \pi_t} = \frac{\partial d_t^*}{\partial \pi_t}(\pi_t),
\]
so both the level of demand and its slope vary as $\pi$ varies. Consequently, the higher-order primitives
\[
\hat{\ELoss}_t'''(d_t^*), \qquad C''(d_t^*), \qquad \text{and} \qquad \frac{\partial d_t^*}{\partial \pi_t}
\]
are not constants but functions of the equilibrium demand $ d_t^*(\pi_t)$.

Flexibility slope is still given point-wise as in \eqref{eq:S}, but now every component of \eqref{eq:S} is $\pi$-dependent. Explicitly,
\begin{align} \label{eqn:curve_deriv}
\frac{\partial}{\partial\!\left(\frac{\partial d_t^*}{\partial \pi_t}\right)}
\Bigl(\partial_{\pi_t}^2 \Welf_t\Bigr)(\pi_t)
&= 3\,\hat{\ELoss}_t'''( d_t^*(\pi_t))\Bigl(\tfrac{\partial d_t^*}{\partial \pi_t}(\pi_t)\Bigr)^{2}\Bigl[C'( d_t^*(\pi_t))-\pi_t\Bigr] 
+ 2\,\tfrac{\partial d_t^*}{\partial \pi_t}(\pi_t)\,C''( d_t^*(\pi_t)) - 1.
\end{align}
Thus, as $\pi_t$ moves from $\pi^{\F}$ to $\pi_t^{\V}$, the direct variation comes from the bracket $[C'( d_t^*(\pi_t))-\pi_t]$, but is also indirect variation because $\hat{\ELoss}_t'''$, $C''$, and $\frac{\partial d_t^*}{\partial \pi_t}$ are themselves changing as functions of $\pi_t$. Because we know $\hat{\ELoss}_t'''<0$, $\tfrac{\partial d_t^*}{\partial \pi_t}(\pi_t)\leq 0$, $C''( d_t^*(\pi_t))$ we can see that $C'( d_t^*(\pi_t))-\pi_t$ is the only component without a predefined sign. 

For the peak period (or any period where marginal cost is above the price), when $C'( d_t^*(\pi_t))-\pi\geq 0 $, \eqref{eqn:curve_deriv} is negative. On the other hand for the off-peak period $C'( d_t^*(\pi_t))-\pi_t\leq 0 $, so the sign of \eqref{eqn:curve_deriv} is ambiguous. Because an increase in $\partial d_{t}/\partial\pi_t \leq 0$, an increase in this term decreases the flexibility, $|\partial d_{t}/\partial\pi_t|$. So, curvature is increasing in flexibility for the periods where price exceeds marginal cost and is ambiguous in flexibility for periods where price is exceeded by marginal cost.

\hspace{1cm}

\section{Proofs for Section \ref{sec:Cons_Outcomes}}\label{appendix:Cons_Outcomes}



\textbf{Proof for Theorem \ref{thm:delta_util}}

 The total change in utility from flat to variable pricing $\Delta \Util_{i} = \sum_{t \in \mathcal{T}} \Delta \Util_{it}$ can be written as the integral of consumption over the interval between equilibrium prices,

 \begin{align}
    \Delta \Util_{it} = \int_{\pi^{\F}}^{\pi_{t}^{\V}}\frac{\partial \Util_i(\pi_{it})}{\partial \pi} d\pi &= -\int_{\pi^{\F}}^{\pi_{t}^{\V}} A_id^*_{it}(\pi) d\pi. \label{eqn:utility_change_appendix}
\end{align}

Thus, total change in utility can be written as 

\begin{align}
    \Delta \Util_{i} &= -\int_{\pi^{\F}}^{\pi_{\PE}^{\V}} A_{i}d_{i\PE}^{*}(\pi)  d\pi - \int_{\pi^{\F}}^{\pi_{\OP}^{\V}} A_{i}d_{i\OP}^{*}(\pi)  d\pi \nonumber\\
    &= A_{i}\left(-\int_{\pi^{\F}}^{\pi_{\PE}^{\V}} d_{i\PE}^{*}(\pi)  d\pi + \int_{\pi_{\OP}^{\V}}^{\pi^{\F}} d_{i\OP}^{*}(\pi)  d\pi\right) \label{eqn:total_util_breakdown_app}.
\end{align}

\begin{proof}

We reference \ref{eqn:total_util_breakdown} to derive this sufficient and necessary condition for linear and isoelastic demand. We work with a scaled quantity $\frac{\Delta \Util_{i}}{A_{i}}$, shown in \eqref{eqn:scaled_delta_util}, for notational simplicity.

\begin{align}
    \frac{\Delta \Util_{i}}{A_{i}} = \frac{\Delta \Util_{i\PE}}{A_{i}} + \frac{\Delta \Util_{i\OP}}{A_{i}} = -\int_{\pi^{\F}}^{\pi_{\PE}^{\V}} d_{i\PE}^{*}(\pi)  d\pi + \int_{\pi_{\OP}^{\V}}^{\pi^{\F}} d_{i\OP}^{*}(\pi)  d\pi. \label{eqn:scaled_delta_util}
\end{align}

As shown in Figure \ref{fig:delta_util_first_order_approx}, we use first order approximations to obtain upper bounds on $\frac{\Delta \Util_{i\PE}}{A_{i}}$ and $\frac{\Delta \Util_{i\OP}}{A_{i}}$.

\begin{figure}[H]
    \centering
    \begin{minipage}{.35\linewidth}
        \centering
        \begin{tikzpicture}[scale=4]
            \filldraw[gray, opacity=0.5] (0.5,0) -- (0.5,{2*((0.5+1)^-1 - 0.4)}) -- (0.7,{2*((0.5+1)^-1 - 0.4 - 0.2*(0.5+1)^-2)}) -- (0.7,0) -- cycle;
            \draw[-Latex] (0,0) -- (1.2,0) node[anchor=north] {$\pi$};
            \draw[-Latex] (0,0) -- (0,1.4) node[anchor=east] {$d$};
            \draw[ProcessBlue, domain=0:1, smooth] plot (\x, {2*((\x+1)^-1 - 0.4)}) node[anchor=south] {$d_{i\PE}$};
            \draw[dashed] (0.5,{2*((0.5+1)^-1 - 0.4)}) -- (0.5,0) node[anchor=north] {$\pi^{\F}$};
            \draw[dashed] (0.7,{2*((0.7+1)^-1 - 0.4)}) -- (0.7,0) node [anchor=north] {$\pi_{\PE}^{\V}$};
            \draw[dashdotted] (0.5,{2*((0.5+1)^-1 - 0.4)}) -- (0.7,{2*(((0.5+1)^-1 - 0.4) - 0.2*(0.5+1)^-2)});
            \draw (0.55,{2*((0.5+1)^-1 - 0.4 - 0.05*(0.5+1)^-2)}) -- (0.65,{2*((0.5+1)^-1 - 0.4 - 0.05*(0.5+1)^-2)}) node[anchor=south] {$\frac{\partial d_{i\PE}(\pi_{\PE}^{\V})}{\partial \pi}$} -- (0.65,{2*((0.5+1)^-1 - 0.4 - 0.15*(0.5+1)^-2)});
        \end{tikzpicture}
    \end{minipage}
    \begin{minipage}{.35\linewidth}
        \centering
        \begin{tikzpicture}[scale=4]
            \filldraw[gray, opacity=0.5] (0.5,0) -- (0.5,{2*((0.5+1)^-1 - 0.4)}) -- (0.3,{2*(((0.5+1)^-1 - 0.4) + 0.2*(0.3+1)^-2)}) -- (0.3,0) -- cycle;
            \draw[-Latex] (0,0) -- (1.2,0) node[anchor=north] {$\pi$};
            \draw[-Latex] (0,0) -- (0,1.4) node[anchor=east] {$d$};
            \draw[Red, domain=0:1, smooth] plot (\x, {2*((\x+1)^-1 - 0.4)}) node[anchor=south] {$d_{i\OP}$};
            \draw[dashed] (0.5,{2*((0.5+1)^-1 - 0.4)}) -- (0.5,0) node[anchor=north] {$\pi^{\F}$};
            \draw[dashed] (0.3,{2*((0.3+1)^-1 - 0.4)}) -- (0.3,0) node [anchor=north] {$\pi_{\OP}^{\V}$};
            \draw[dashdotted] (0.3,{2*(((0.5+1)^-1 - 0.4) + 0.2*(0.3+1)^-2)}) -- (0.5,{2*((0.5+1)^-1 - 0.4)});
            \draw (0.35,{2*(((0.5+1)^-1 - 0.4) + 0.15*(0.3+1)^-2)}) -- (0.45,{2*(((0.5+1)^-1 - 0.4) + 0.15*(0.3+1)^-2)}) node[anchor=south] {$\frac{\partial d_{i\OP}(\pi^{\F})}{\partial \pi}$} -- (0.45,{2*(((0.5+1)^-1 - 0.4) + 0.05*(0.3+1)^-2)});
        \end{tikzpicture}
    \end{minipage}
    \caption{The shaded area on the blue curve illustrates a lower bound for $\frac{|\Delta \Util_{i\PE}|}{A_{i}}$ (upper bound for $\frac{\Delta \Util_{i\PE}}{A_{i}}$), while the shaded area on the red curve illustrates an upper bound for $\frac{\Delta \Util_{i\OP}}{A_{i}}$.}
    \label{fig:delta_util_first_order_approx}
\end{figure}
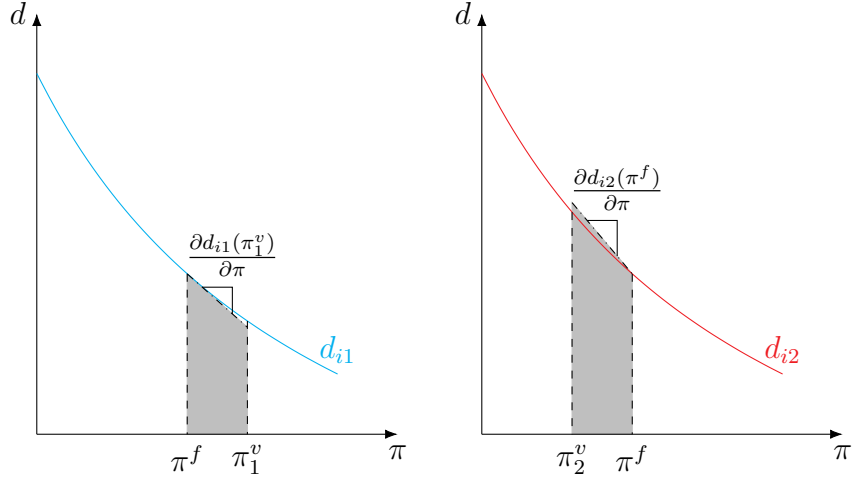

Because of the assumption of convexity of $d^*_{i\PE}$ with respect to $\pi$, we have the lower bound

\begin{align}
    d^*_{i\PE}(\pi) &\geq d_{i\PE}^{\F} + (\pi-\pi^{\F}) \frac{\partial d^*_{i\PE}(\pi^{\F})}{\partial \pi} \quad \text{for all } \pi \in [\pi^{\F}, \pi_{\PE}^{\V}].
\intertext{resulting in the upper bound}
    \frac{\Delta \Util_{i\PE}}{A_{i}} = -\int_{\pi^{\F}}^{\pi_{\PE}^{\V}} d^*_{i\PE}(\pi)  d\pi &\leq -\int_{\pi^{\F}}^{\pi_{\PE}^{\V}} \left(d_{i\PE}^{\F} + (\pi - \pi^{\F}) \frac{\partial d^*_{i\PE}(\pi^{\F})}{\partial \pi}\right)  d\pi \nonumber\\
    &= -|\Delta \pi_{\PE}|d_{i\PE}^{\F} - \frac{|\Delta \pi_{\PE}|^{2}}{2}\frac{\partial d^*_{i\PE}(\pi^{\F})}{\partial \pi} \label{eqn:delta_util_peak_upper_bd}.
\end{align}

Similarly, we have the upper bound 

\begin{align}
    d^*_{i2}(\pi) &\leq d_{i\OP}^{\F} - (\pi^{\F}-\pi)\frac{\partial d^*_{i\OP}(\pi_{\OP}^{\V})}{\partial \pi} \quad \text{for all } \pi \in [\pi_{\OP}^{\V}, \pi^{\F}], \nonumber
\intertext{resulting in the upper bound,}
    \frac{\Delta \Util_{i\OP}}{A_{i}} = \int_{\pi_{\OP}^{\V}}^{\pi^{\F}} d^*_{i\OP}(\pi)  d\pi &\leq \int_{\pi_{\OP}^{\V}}^{\pi^{\F}} \left(d_{i\OP}^{\F} - (\pi^{\F} - \pi) \frac{\partial d^*_{i\OP}(\pi_{\OP}^{\V})}{\partial \pi}\right)  d\pi\\
    &= |\Delta \pi_{\OP}|d_{i\OP}^{\F} - \frac{|\Delta \pi_{\OP}|^{2}}{2}\frac{\partial d^*_{i\OP}(\pi_{\OP}^{\V})}{\partial \pi}. \label{eqn:delta_util_off_peak_upper_bd}
\end{align}

Combining equations \eqref{eqn:delta_util_peak_upper_bd} and \eqref{eqn:delta_util_off_peak_upper_bd}, we compute the sufficient condition

\begin{align*}
    \frac{\Delta \Util_{i}}{A_{i}} = \frac{\Delta \Util_{i\PE}}{A_{i}} + \frac{\Delta \Util_{i\OP}}{A_{i}} \leq -|\Delta \pi_{1}|d_{i\PE}^{\F} - \frac{|\Delta \pi_{\PE}|^{2}}{2} \frac{\partial d^*_{i\PE}(\pi^{\F})}{\partial \pi} + |\Delta \pi_{2}|d_{i\OP}^{\F} - \frac{|\Delta \pi_{2}|^{2}}{2}\frac{\partial d^*_{i\OP}(\pi_{\OP}^{\V})}{\partial \pi} &< 0\\
    \implies (-|\Delta \pi_{\PE}|d_{i\PE}^{\F} + |\Delta \pi_{\OP}|d_{i\OP}^{\F}) + \left(\frac{|\Delta \pi_{\PE}|^{2}}{2} \left|\frac{\partial d^*_{i\PE}(\pi^{\F})}{\partial \pi}\right| + \frac{|\Delta \pi_{\OP}|^{2}}{2} \left|\frac{\partial d^*_{i\OP}(\pi_{\OP}^{\V})}{\partial \pi}\right|\right) < 0.
\end{align*}

Because the approximations made to generate the bound \eqref{eqn:delta_util_suff} are first-order approximations, \eqref{eqn:delta_util_suff} is a sufficient and necessary condition for $\Delta \Util_{i} < 0$ when demand is affine. 
Thus, for a consumer $i$ with quadratic electricity loss $\ELoss_{it}$, $\Delta \Util_{i} < 0$ if and only if

    \begin{align*}
        -|\Delta \pi_{\PE}|d_{i\PE}^{\F} + |\Delta \pi_{\OP}|d_{i\OP}^{\F} + \frac{|\Delta \pi_{\PE}|^{2}}{2} \left|\frac{\partial d^*_{i\PE}(\pi^{\F})}{\partial \pi}\right| + \frac{|\Delta \pi_{\OP}|^{2}}{2} \left|\frac{\partial d^*_{i\OP}(\pi_{\OP}^{\V})}{\partial \pi}\right| < 0,
    \end{align*}

    which occurs exactly when their change in incidence if they were completely inelastic, is upper bounded by a weighted sum of consumer flexibility.

    Now we derive sufficient and necessary conditions for isoelastic utility. Starting with the general form of demand under a variable price $\pi$, we can substitute
$d(\pi) = d^\F \cdot ( \frac{\pi}{\pi^\F} )^{\epsilon(\pi)}$ into the expression of $\Delta \Util_i$,
which results in the following restatement of total change in utility,
\begin{align*}
\Delta \Util_i = -\int_{\pi^\F}^{\pi_{\PE}^\V} d_{i\PE}^\F \cdot \left( \frac{\pi_{\PE}}{\pi^\F} \right)^{-\epsilon_{i\PE}(\pi_{\PE})}   d\pi_{\PE} + \int_{\pi_{\OP}^\V}^{\pi^\F} d_{i\OP}^\F \cdot \left( \frac{\pi_{\OP}}{\pi^\F} \right)^{-\epsilon_{i\OP}(\pi_{\OP})}   d\pi_{\OP}.
\end{align*}

This expression incorporates demand responses over two pricing intervals: from $\pi^\F$ to $\pi_{\PE}^\V$ for peak and from $\pi_{\OP}^\V$ to $\pi^\F$ for off-peak periods. Each term reflects the demand elasticity as a function of price within these intervals.

Elasticity at a given price $\pi$ can be decomposed as follows:
\begin{align*}
\epsilon(\pi) = \epsilon^\F + \Delta \epsilon(\pi),
\end{align*}
where $\epsilon^\F$ is the elasticity at the flat price $\pi^\F$, and $\Delta \epsilon(\pi)$ is movement away from the flat price as the price changes. Substituting the decomposition $\epsilon(\pi) = \epsilon^\F + \Delta \epsilon(\pi)$, we rewrite:
\begin{align}
d(\pi) = d^\F \cdot \left( \frac{\pi}{\pi^\F} \right)^{\epsilon^\F + \Delta \epsilon(\pi)} = d^\F \cdot \left( \frac{\pi}{\pi^\F} \right)^{\epsilon^\F} \cdot \left( \frac{\pi}{\pi^\F} \right)^{\Delta \epsilon(\pi)} \label{eqn:elastic_demand}.
\end{align}
where, for isoelastic utility, $\left( \frac{\pi}{\pi^\F} \right)^{\Delta \epsilon(\pi)}=1$

Substituting this back into \eqref{eqn:elastic_demand}, we obtain:
\begin{align*}
d(\pi) = d^\F \cdot \left( \frac{\pi}{\pi^\F} \right)^{\epsilon^\F}.
\end{align*}
 This expression of demand reflects the demand function's response to both the baseline elasticity at the flat price $\epsilon^\F$ and the variability term $\Delta \epsilon(\pi)$.
Substituting the decomposed elasticity into the utility change expression, we obtain:
\begin{align*}
\Delta \Util_i = -\int_{\pi^\F}^{\pi_{\PE}^\V} d_{i\PE}^\F \cdot \left( \frac{\pi_{\PE}}{\pi^\F} \right)^{\epsilon_{i\PE}^\F} \cdot    d\pi_{\PE}
+ \int_{\pi_{\OP}^\V}^{\pi^\F} d_{i\OP}^\F \cdot \left( \frac{\pi_{\OP}}{\pi^\F} \right)^{\epsilon_{i\OP}^\F}    d\pi_{\OP}.
\end{align*} 
We then take this integral and find when it is less than zero:
\begin{align}
\frac{d_{i\PE}^\F \cdot \pi^\F}{1 + \epsilon_{i\PE}^\F} \left[ 1 - \left( \frac{\pi_{\PE}^\V}{\pi^\F} \right)^{1 + \epsilon_{i\PE}^\F} \right] 
+ \frac{d_{i\OP}^\F \cdot \pi^\F}{1 + \epsilon_{i\OP}^\F} \left[ 1 - \left( \frac{\pi_{\OP}^\V}{\pi^\F} \right)^{1 + \epsilon_{i\OP}^\F} \right] < 0. \nonumber\end{align}
We then can then divide by the flat price term,
\begin{align*}
\frac{d_{i\PE}^\F }{1 + \epsilon_{i\PE}^\F} \left[ 1 - \left( \frac{\pi_{\PE}^\V}{\pi^\F} \right)^{1 + \epsilon_{i\PE}^\F} \right] 
+ \frac{d_{i\OP}^\F }{1 + \epsilon_{i\OP}^\F} \left[ 1 - \left( \frac{\pi_{\OP}^\V}{\pi^\F} \right)^{1 + \epsilon_{i\OP}^\F} \right] < 0.
\end{align*}
This condition is necessary and sufficient when elasticity remains constant: a consumer $i$ loses utility if and only if  this inequality is satisfied.
\end{proof}

\textbf{Proof for Corollary \ref{cor:sufficient_change}}\\ The change in utility during peak periods is:

   \[
   \Delta \Util_{i} = -A_i\int_{\pi^\F}^{\pi_{\PE}^\V} d^*_{it}(\pi)   d\pi+A_i\int^{\pi^\F}_{\pi_{\OP}^\V} d^*_{it}(\pi)   d\pi 
   \]
   In Theorem \ref{thm:delta_util}, we presented the following relationships. 
   \begin{align*}
    \frac{\Delta \Util_{i\PE}}{A_{i}} = -\int_{\pi^{\F}}^{\pi_{\PE}^{\V}} d^*_{i\PE}(\pi)  d\pi &\leq -\int_{\pi^{\F}}^{\pi_{\PE}^{\V}} \left(d_{i\PE}^{\F} + (\pi - \pi^{\F}) \frac{\partial d^*_{i\PE}(\pi^{\F})}{\partial \pi}\right)  d\pi \nonumber
    = -|\Delta \pi_{\PE}|d_{i\PE}^{\F} - \frac{|\Delta \pi_{\PE}|^{2}}{2}\frac{\partial d^*_{i\PE}(\pi^{\F})}{\partial \pi}\\
    \frac{\Delta \Util_{i\OP}}{A_{i}} = \int_{\pi_{\OP}^{\V}}^{\pi^{\F}} d^*_{i\OP}(\pi)  d\pi &\leq \int_{\pi_{\OP}^{\V}}^{\pi^{\F}} \left(d_{i\OP}^{\F} - (\pi^{\F} - \pi) \frac{\partial d^*_{i\OP}(\pi_{\OP}^{\V})}{\partial \pi}\right)  d\pi
    = |\Delta \pi_{\OP}|d_{i\OP}^{\F} - \frac{|\Delta \pi_{\OP}|^{2}}{2}\frac{\partial d^*_{i\OP}(\pi_{\OP}^{\V})}{\partial \pi}.
\end{align*}
We can therefore say that \eqref{eqn:delta_util_suff} is sufficient for all utility functions of the form of \eqref{eqn:ConsumUtil} that fall under Assumptions \ref{Assumption:Consumer}.\\

\textbf{Proof for Proposition \ref{prop:low_high_util_compare} }

\begin{proof}
    We rearrange \eqref{eqn:util_diff_general} to find that $\Delta \Util_{\L} < \Delta \Util_{\H}$ is equivalent to the condition
    \begin{align}
        \int_{\pi^{\F}}^{\pi_{\PE}^{\V}} A_{\H}d_{\hp}(\pi)  d\pi + \int_{\pi_{\OP}^{\V}}^{\pi^{\F}} A_{\L}d_{\lo}(\pi)  d\pi < \int_{\pi_{\OP}^{\V}}^{\pi^{\F}} A_{\H}d_{\ho}(\pi)  d\pi + \int_{\pi^{\F}}^{\pi_{\PE}^{\V}} A_{\L}d_{\lp}(\pi)  d\pi. \label{eqn:delta_util_l_h}
    \end{align}

    Here, we can use the convexity of $d_{it}(\pi)$ to upper bound the left hand side and lower bound the right hand side of \eqref{eqn:delta_util_l_h}. Similar to Figure \ref{fig:delta_util_first_order_approx}, we compute the upper bound on the left hand side of \eqref{eqn:delta_util_l_h} as
    \begin{align}
        &A_{\H} \int_{\pi^{\F}}^{\pi_{\PE}^{\V}} \left(d_{\hp}^{\F} + (\pi - \pi^{\F}) \frac{\partial d_{\hp}^{\V}}{\partial \pi}\right)  d\pi + A_{\L} \int_{\pi_{\OP}^{\V}}^{\pi^{\F}} \left(d_{\lo}^{\F} - (\pi^{\F} - \pi)\frac{\partial d_{\lo}^{\V}}{\partial \pi}\right) d\pi \nonumber\\
        = &A_{\H}\left(|\Delta \pi_{\PE}|d_{\hp}^{\F} - \frac{|\Delta \pi_{\PE}|^{2}}{2}\left|\frac{\partial d_{\hp}^{\V}}{\partial \pi}\right|\right) + A_{\L}\left(|\Delta \pi_{\OP}|d_{\lo}^{\F} + \frac{|\Delta \pi_{\OP}|^{2}}{2}\left|\frac{\partial d_{\lo}^{\V}}{\partial \pi}\right|\right),
\label{eqn:delta_util_l_h_first_order_approx_upper_bound_lhs}
    \end{align}
    while we compute the lower bound on the right hand side of \eqref{eqn:delta_util_l_h} as
    \begin{align}
         &A_{\H}\int_{\pi_{\OP}^{\V}}^{\pi^{\F}} \left(d_{\ho}^{\F} - (\pi^{\F} - \pi)\frac{\partial d_{\ho}^{\F}}{\partial \pi}\right)  d\pi + A_{\L}\int_{\pi^{\F}}^{\pi_{\PE}^{\V}} \left(d_{\lp}^{\F} + (\pi - \pi^{\F})\frac{\partial d_{\lp}^{\F}}{\partial \pi}\right)  d\pi\nonumber\\
         = & A_{\H}\left(|\Delta \pi_{\OP}|d_{\ho}^{\F} + \frac{|\Delta \pi_{\OP}|^{2}}{2}\left|\frac{\partial d_{\ho}^{\F}}{\partial \pi}\right|\right) + A_{\L}\left(|\Delta \pi_{\PE}|d_{\lp}^{\F} - \frac{|\Delta \pi_{\PE}|^{2}}{2}\left|\frac{\partial d_{\lp}^{\F}}{\partial \pi}\right|\right).
         \label{eqn:delta_util_l_h_first_order_approx_lower_bound_rhs}
    \end{align}

    Combining the bounds in Equations \eqref{eqn:delta_util_l_h_first_order_approx_upper_bound_lhs} and \eqref{eqn:delta_util_l_h_first_order_approx_lower_bound_rhs},

    \begin{align}
        A_{\L}\left(-|\Delta \pi_{\PE}|d_{\lp}^{\F} + |\Delta \pi_{\OP}|d_{\lo}^{\F} + \frac{|\Delta \pi_{\PE}|^{2}}{2}\left|\frac{\partial d_{\lp}^{\F}}{\partial \pi}\right| + \frac{|\Delta \pi_{\OP}|^{2}}{2}\left|\frac{\partial d_{\lo}^{\V}}{\partial \pi}\right|\right) \nonumber\\< A_{\H}\left(-|\Delta \pi_{\PE}|d_{\hp}^{\F} + |\Delta \pi_{\OP}|d_{\ho}^{\F} + \frac{|\Delta \pi_{\PE}|^{2}}{2}\left|\frac{\partial d_{\hp}^{\V}}{\partial \pi}\right| + \frac{|\Delta \pi_{\OP}|^{2}}{2}\left|\frac{\partial d_{\ho}^{\F}}{\partial \pi}\right|\right). \label{eqn:exact_compare_quad}
    \end{align}

 In the case of quadratic electricity loss, corresponding to affine electricity demand, the approximations of $\Delta \Util_{\L}$ and $\Delta \Util_{\H}$ in \ref{eqn:exact_compare_quad} are exact. 

 We now examine the isoelastic case. To establish the condition under which \( \Delta \Util_{\L} - \Delta \Util_{\H} < 0 \), reflecting that the utility loss for a low-income consumer (\( \L \)) is smaller than that for a high-income consumer (\( \H \)), we begin with the utility change expression for each consumer. The change in utility is given by \eqref{eqn:exact_compare_quad} (adjusted for common factors).

For a low-income consumer (\( \L \)) and a high-income consumer (\( \H \)), we subtract their respective utility changes to obtain 
\begin{align*}
\Delta \Util_{\L} - \Delta \Util_{\H}\propto&A_{\L}\left( \frac{d_{\L\PE}^\F }{1 + \epsilon_{\L\PE}^\F} \left[ 1 - \left( \frac{\pi_{\PE}^\V}{\pi^\F} \right)^{1 + \epsilon_{\L\PE}^\F} \right] 
+ \frac{d_{\L\OP}^\F }{1 + \epsilon_{\L\OP}^\F} \left[ 1 - \left( \frac{\pi_{\OP}^\V}{\pi^\F} \right)^{1 + \epsilon_{\L\OP}^\F} \right]\right) \\&-A_{\H}\left(\frac{d_{\H\PE}^\F }{1 + \epsilon_{\H\PE}^\F} \left[ 1 - \left( \frac{\pi_{\PE}^\V}{\pi^\F} \right)^{1 + \epsilon_{\H\PE}^\F} \right] 
- \frac{d_{\H\OP}^\F }{1 + \epsilon_{i\OP}^\F} \left[ 1 - \left( \frac{\pi_{\OP}^\V}{\pi^\F} \right)^{1 + \epsilon_{\H\OP}^\F} \right]\right).
\end{align*}

This expression separates the differences in utility loss into peak and off-peak contributions. For \( \Delta \Util_{\L} - \Delta \Util_{\H} < 0 \) to hold, the peak-period utility difference must be dominated by the off-peak-period utility difference, such that the net utility loss for the high-income consumer exceeds that for the low-income consumer. Simplifying further, this inequality depends on the relative magnitudes of \( \frac{1}{1 + \epsilon} \) for low- and high-income consumers, as larger \( \epsilon \) values amplify the response to price changes.
\end{proof}

\textbf{Proof for Corollary \ref{cor:sufficient_change_compare}}
\begin{proof}
    
From the proof for Proposition \ref{prop:low_high_util_compare} , combining the bounds in \eqref{eqn:delta_util_l_h_first_order_approx_upper_bound_lhs} and \eqref{eqn:delta_util_l_h_first_order_approx_lower_bound_rhs} results in the sufficient condition,
    \begin{align}
        A_{\L}\left(-|\Delta \pi_{\PE}|d_{\lp}^{\F} + |\Delta \pi_{\OP}|d_{\lo}^{\F} + \frac{|\Delta \pi_{\PE}|^{2}}{2}\left|\frac{\partial d_{\lp}^{\F}}{\partial \pi}\right| + \frac{|\Delta \pi_{\OP}|^{2}}{2}\left|\frac{\partial d_{\lo}^{\V}}{\partial \pi}\right|\right) \nonumber\\< A_{\H}\left(-|\Delta \pi_{\PE}|d_{\hp}^{\F} + |\Delta \pi_{\OP}|d_{\ho}^{\F} + \frac{|\Delta \pi_{\PE}|^{2}}{2}\left|\frac{\partial d_{\hp}^{\V}}{\partial \pi}\right| + \frac{|\Delta \pi_{\OP}|^{2}}{2}\left|\frac{\partial d_{\ho}^{\F}}{\partial \pi}\right|\right). \label{eqn:suff_compare_quad}
    \end{align}
As these provide an upper bound on low-income utility change and a lower bound on high-income utility change \eqref{eqn:suff_compare_quad} is sufficient but not necessary for the low income consumer to lose utility relative to the high income consumer.

\end{proof}

\section{Proofs for Section \ref{sec:extensions}}\label{appendix:extensions}
\subsection{Proof of the Profit Constraint Condition}

In the unconstrained operator problem shown in \eqref{eqn:operator_opt}-\eqref{eqn:price_constraint}, the operator sets the period-specific price equal to marginal cost, $\pi_t = C'(d_t)$. Introducing a profit requirement alters this outcome if the welfare maximizing profit is not equal to the profit requirement. We examine how prices change given a profit constraint but do not analyze the results of Section \ref{sec:Cons_Outcomes} as the conditions there take price changes as input. To begin, suppose the operator must satisfy an aggregate profit condition,
\begin{align}
    \sum_{t\in\mathcal{T}} \Prof_t(\pi) = \bar{\Prof}, \label{app:profit_constraint}
\end{align}

where $\Prof_t(\pi) = \pi_t d_t^*(\pi_t) - C(d_t^*(\pi_t))$. Note that instead of the equality we could instead include an inequality to represent either a floor (to ensure cost recovery) or a cap (to limit rents).
To incorporate \eqref{app:profit_constraint}, the planner’s problem is written with a Lagrange multiplier $ \nu$:
$$
\mathcal{L} = \sum_t \Welf_t(\pi_t) +  \nu \left(\sum_t \Prof_t(\pi_t) - \bar{\Prof}\right),
$$
where $\Welf_t(\pi_t) = U_t(d_t(\pi_t)) - C(d_t(\pi_t))$ is by-period welfare. 
We then differentiate with respect to $\pi_t$. By the envelope theorem we have,
$$
\big(C'-\pi_t^*\big)\frac{\partial d^*_t}{\partial \pi_t}
 + \nu \left[d_t^*+(\pi^*_t-C')\frac{\partial d^*_t}{\partial \pi_t}\right]=0.
$$
Recall the price elasticities relationship to demand, flexibility and price: $\epsilon_t = \dfrac{\partial d^*_t}{\partial \pi_t}\cdot\dfrac{\pi_t}{d^*_t}$ (note $\epsilon_t<0$). Then for the optimal $\pi^*$,
$$
d^*_t=\frac{\pi_t^*}{\epsilon_t}\frac{\partial d^*_t}{\partial \pi_t}.
$$

Substitute and factor $\partial d^*_t/\partial \pi_t\neq 0$:
$$(C'-\pi_t^*) + \nu \left[\frac{\pi_t^*}{\epsilon_t}+(\pi_t^*-C')\right]=0.$$

Rearrange:
\begin{align}
(C'-\pi_t^*)(1-\nu)=\nu \frac{\pi_t^*}{\epsilon_t}.
\label{distort}
\end{align}
When $\nu>0$ binds, prices deviate from marginal cost. Solve \eqref{distort} for the distortion:
$$
\pi_t^* - C'(d_t^*)
= - \frac{\nu}{1-\nu} \frac{\pi_t^*}{\epsilon_t}
$$
because $\epsilon_t<0$. So the wedge from marginal cost is proportional to the markup factor $\dfrac{\nu}{1-\nu}$ and inversely proportional to the (absolute) demand elasticity.

If the profit floor $\bar{\Prof}$ is set below the unconstrained profit level, then the constraint binds with $ \nu>0$. Because $\tfrac{\partial \Welf_t}{\partial \pi_t}$ is increasing in price and is zero at $\pi_t^{\V}$, the binding constraint forces prices upward: $\pi_t^*>\pi_t^{\V}$. In this case, both peak and off-peak prices rise further above the flat price than in the unconstrained solution.

Conversely, if the profit cap $\bar{\Prof}$ is set above the unconstrained profit level, the constraint binds with $ \nu<0$. In this case, equilibrium requires lowering prices relative to $\pi_t^{\V}$, so that $\pi_t^*<\pi_t^{\V}$. Both peak and off-peak prices move closer to the flat price compared to the unconstrained outcome. If the cap exceeds the maximum feasible profit (i.e., monopoly profit), no feasible solution exists.

This markup (if $ \nu>0$) or markdown (if $ \nu<0$) is inversely proportional to demand elasticity; the operator adjusts prices more in periods where consumers are less responsive to price. 
Because the multiplier $ \nu$ applies to the profit constraint globally, it ties together all periods. A tighter profit requirement (larger $| \nu|$) increases the deviation in every period, but the allocation of that deviation across periods is determined by relative elasticities. Taking the ratio of distortions for two periods $t,t'$ yields
$$
\frac{\pi_t - C'(d_t)}{\pi_{t'} - C'(d_{t'})} = \frac{\pi_t/\epsilon_t}{\pi_{t'}/\epsilon_{t'}}.
$$

This condition shows that the markups are inversely proportional across periods, so that periods with smaller $|\epsilon_t|$ bear a larger share of the adjustment. The cross-period link is essential: regardless of how profits are constrained in aggregate, the optimal adjustment balances out across periods in inverse proportion to their elasticities.

\textbf{Multi Period Proofs}
    Recall from \eqref{eqn:operator_opt} the objective function of the grid operator under flat pricing:
    \begin{align*}
        \Welf = \sum_{t \in \mathcal{T}} \Welf_{t} = \sum_{t \in \mathcal{T}} \hat{\ELoss}_{t}(d_{t}(\pi^{\F})) + C(d_{t}(\pi^{\F})).
    \end{align*}
    To minimize this expression, when all consumer demand levels are in the interior of their limits $d_{it} \in (0, \bar{d}_{it})$, we can take the derivative of $\Welf$ with respect to $\pi^{\F}$, and equate the derivative to $0$. The derivative is as follows:
    \begin{align*}
        &\frac{\partial \Welf}{\partial \pi^{\F}} = 0\\
        \implies &\sum_{t \in \mathcal{T}} \left(\hat{\ELoss}_{t}'(d_{t}(\pi^{\F})) + C'(d_{t}(\pi^{\F}))\right) \left|\frac{\partial d_{t}(\pi^{\F})}{\partial \pi^{\F}} \right|= 0.
        \end{align*}
        Recalling that $\hat{\ELoss}_{t}'(d_{t}(\pi^{\F})) = -\pi^{\F}$ at equilibrium, and that $\partial d_{t}/\partial \pi^{\F}\neq 0$ is negative for interior solutions, we simplify to
        \begin{align*}
        & \sum_{t \in \mathcal{T}} \left(-\pi^{\F} + C'(d_{t}(\pi^{\F}))\right) \left|\frac{\partial d_{t}}{\partial \pi^{\F}}\right| = 0\\
        \implies & \sum_{t \in \mathcal{T}} C'(d_{t}(\pi^{\F})) \left|\frac{\partial d_{t}}{\partial \pi^{\F}}\right| = \pi^{\F} \sum_{t \in \mathcal{T}} \left|\frac{\partial d_{t}}{\partial \pi^{\F}}\right|\\
        \implies & \pi^{\F} = \frac{\sum_{t \in \mathcal{T}} C'(d_{t}(\pi^{\F})) |\partial d_{t}/\partial \pi^{\F}|}{\sum_{t \in \mathcal{T}} |\partial d_{t}/\partial \pi^{\F}|}, 
    \end{align*}
    as desired. We do not include the properties of the variable prices or demands as they need no further extension from what was presented in Section \ref{sec:equilibrium} and \ref{Appendix:equilibrium}. This is because prices and demand chosen period-by period for variable pricing rather across periods as is the case for flat pricing. 
    
    Examining consumer utility change, we first rewrite the change in utility $\Delta \Util_{i}$ and $\Delta \Util_{j}$ of consumers $i$ and $j$ due to a switch from flat to variable pricing as follows:
    \begin{align*}
        \Delta \Util_{i} = A_{i}\left[\sum_{t \in \mathcal{T}} \int_{\pi^{\F}}^{\pi_{t}^{\V}} (-d_{it}(\pi))  d\pi\right],\\
        \Delta \Util_{j} = A_{j}\left[\sum_{t \in \mathcal{T}} \int_{\pi^{\F}}^{\pi_{t}^{\V}} (-d_{jt}(\pi))  d\pi\right].
    \end{align*}
    To obtain a sufficient condition for $\Delta \Util_{i} < \Delta \Util_{j}$, we seek to upper bound $\Delta \Util_{i}$ and lower bound $\Delta \Util_{j}$. We obtain these bounds by employing the following bounds on demand, which can be verified due to convexity: for $\pi$ between $\pi^{\F}$ and $\pi_{t}^{\V}$, we have
    \begin{align*}
        d_{it}(\pi) \geq d_{it}(\pi^{\F}) + (\pi - \pi^{\F})d_{it}'(\pi^{\F}),\\
        d_{it}(\pi) \leq d_{it}(\pi^{\F}) + (\pi - \pi^{\F})d_{it}'(\pi_{t}^{\V}).
    \end{align*}

    We now use these bounds in two distinct cases:\\
    \textbf{Case 1: $\pi^{\F} < \pi_{t}^{\V}$.} Here, we can compute that
    \begin{align*}
        &-\int_{\pi^{\F}}^{\pi_{t}^{\V}} (d_{it}(\pi^{\F}) + (\pi - \pi^{\F})d_{it}'(\pi^{\F}))  d\pi \geq -\int_{\pi^{\F}}^{\pi_{t}^{\V}} d_{it}(\pi)  d\pi \geq -\int_{\pi^{\F}}^{\pi_{t}^{\V}} (d_{it}(\pi^{\F}) + (\pi - \pi^{\F})d_{it}'(\pi_{t}^{\V})).\\
        \intertext{Carrying out the integral computations, we find}
        &(\pi^{\F} - \pi_{t}^{\V})d_{it}^{\F} + \frac{|\Delta \pi_{t            }|^{2}}{2}|d_{it}'(\pi^{\F})| > -\int_{\pi^{\F}}^{\pi_{t}^{\V}} d_{it}(\pi)  d\pi > (\pi^{\F} - \pi_{t}^{\V})d_{it}^{\F} + \frac{|\Delta \pi_{t}|^{2}}{2}|d_{it}'(\pi_{t}^{\V})|.
    \end{align*}

\noindent\textbf{Case 2: $\pi^{\F} > \pi_{t}^{\V}$.} Here, we can compute that
    \begin{align*}
        &-\int_{\pi^{\F}}^{\pi_{t}^{\V}} (d_{it}(\pi^{\F}) + (\pi - \pi^{\F})d_{it}'(\pi^{\F}))  d\pi < -\int_{\pi^{\F}}^{\pi_{t}^{\V}} d_{it}(\pi)  d\pi < -\int_{\pi^{\F}}^{\pi_{t}^{\V}} (d_{it}(\pi^{\F}) + (\pi - \pi^{\F})d_{it}'(\pi_{t}^{\V})).\\
        \intertext{Carrying out the integral computations, we find}
        &(\pi^{\F} - \pi_{t}^{\V})d_{it}^{\F} + \frac{|\Delta \pi_{t}|^{2}}{2}|d_{it}'(\pi^{\F})| < -\int_{\pi^{\F}}^{\pi_{t}^{\V}} d_{it}(\pi)  d\pi < (\pi^{\F} - \pi_{t}^{\V})d_{it}^{\F} + \frac{|\Delta \pi_{t}|^{2}}{2}|d_{it}'(\pi_{t}^{\V})|.
    \end{align*}
    Therefore, for time periods $t_{1} \in \mathcal{T_{1}}$ where $\pi^{\F} < \pi_{t}^{\V}$, we can upper bound $\Delta \Util_{it}$ by
    \begin{align*}
        &A_{i}\left[(\pi^{\F} - \pi_{t}^{\V})d_{it}^{\F} + \frac{|\Delta \pi_{t}|^{2}}{2}|d_{it}'(\pi^{\F})|\right],\\
        \intertext{while we can lower bound $\Delta \Util_{jt}$ by}
        &A_{j}\left[(\pi^{\F} - \pi_{t}^{\V})d_{jt}^{\F} + \frac{|\Delta \pi_{t}|^{2}}{2}|d_{jt}'(\pi_{t}^{\V})|\right].
    \end{align*}

    Similarly, for time periods $t_{2} \in \mathcal{T_{2}}$ where $\pi^{\F} > \pi_{t}^{\V}$, we can upper bound $\Delta \Util_{it}$ by
    \begin{align*}
        &A_{i}\left[(\pi^{\F} - \pi_{t}^{\V})d_{it}^{\F} + \frac{|\Delta \pi_{t}|^{2}}{2}|d_{it}'(\pi_{t}^{\V})|\right],\\
        \intertext{while we can lower bound $\Delta \Util_{jt}$ by}
        &A_{j}\left[(\pi^{\F} - \pi_{t}^{\V})d_{jt}^{\F} + \frac{|\Delta \pi_{t}|^{2}}{2}|d_{jt}'(\pi^{\F})|\right].
    \end{align*}

    Adding up the respective upper bounds for $\Delta \Util_{it}$ and lower bounds for $\Delta \Util_{jt}$ across respective time periods $t_{1}$ in $\mathcal{T}_{1}$ and $t_{2}$ in $\mathcal{T}_{2}$, we obtain the sufficient condition for $\Delta \Util_{i}<0$ analogous to \eqref{eqn:delta_util_suff},
    \begin{align*}
        &A_{i}\left[\sum_{t \in \mathcal{T}}(\pi^{\F} - \pi_{t}^{\V})d_{it}^{\F} + \sum_{t_{1} \in \mathcal{T}_{1}} \frac{|\Delta \pi_{t_{1}}|^{2}}{2}|d_{it_{1}}'(\pi^{\F})| + \sum_{t_{2} \in \mathcal{T}_{2}} \frac{|\Delta \pi_{t_{2}}|^{2}}{2}|d_{it_{2}}'(\pi_{t_{2}}^{\V})|\right] 
        < 0 .
    \end{align*}
    We also reproduce the conditions $\Delta \Util_{i} < \Delta \Util_{j}$ shown in \eqref{eqn:linear_comparison},
    \begin{align*}
        &A_{i}\left[\sum_{t \in \mathcal{T}}(\pi^{\F} - \pi_{t}^{\V})d_{it}^{\F} + \sum_{t_{1} \in \mathcal{T}_{1}} \frac{|\Delta \pi_{t_{1}}|^{2}}{2}|d_{it_{1}}'(\pi^{\F})| + \sum_{t_{2} \in \mathcal{T}_{2}} \frac{|\Delta \pi_{t_{2}}|^{2}}{2}|d_{it_{2}}'(\pi_{t_{2}}^{\V})|\right] \nonumber\\
        < & A_{j}\left[\sum_{t \in \mathcal{T}}(\pi^{\F} - \pi_{t}^{\V})d_{jt}^{\F} + \sum_{t_{1} \in \mathcal{T}_{1}} \frac{|\Delta \pi_{t_{1}}|^{2}}{2}|d_{jt}'(\pi_{t_{1}}^{\V})| + \sum_{t_{2} \in \mathcal{T}_{2}} \frac{|\Delta \pi_{t_{2}}|^{2}}{2}|d_{jt}'(\pi^{\F})|\right].
    \end{align*}

We can also find multi-period utility change under isoelastic demand. In \eqref{clm:dec_U} have established,
$$
\frac{\partial \Util_{it}^{*}(\pi)}{\partial \pi} = -A_i  d_{it}(\pi),
$$
for interior solutions. Integrating from $\pi^{\F}$ to $\pi_t^{\V}$ gives the by-period utility change:
$$
\Delta \Util_{it}
= \Util_{it}^{\V}-\Util_{it}^{\F}
= -A_i\int_{\pi^{\F}}^{\pi_t^{\V}} d_{it}(\pi_t)  d\pi_t.
$$

Recall isoelastic demand 
$$
d^*_{it}(\pi_t)=d_{it}^{\F}\left(\frac{\pi_t}{\pi^{\F}}\right)^{\epsilon_{it}}
= d_{it}^{\F} (\pi^{\F})^{-\epsilon_{it}}\pi_t^{\epsilon_{it}}.
$$

Thus
$$
\Delta \Util_{it}
= -A_i  d_{it}^{\F}(\pi^{\F})^{-\epsilon_{it}}
\int_{\pi^{\F}}^{\pi_t^{\V}} \pi_t^{\epsilon_{it}}  d\pi
= -A_i  d_{it}^{\F}(\pi^{\F})^{-\epsilon_{it}}
\left[\frac{\pi_t^{1+\epsilon_{it}}}{1+\epsilon_{it}}\right]_{\pi^{\F}}^{\pi_t^{\V}}.
$$

Rearranging,
$$
\Delta \Util_{it}
= A_i \frac{d_{it}^{\F}}{1+\epsilon_{it}}
\left[1-\left(\frac{\pi_t^{\V}}{\pi^{\F}}\right)^{1+\epsilon_{it}}\right].
$$

Summing over $t$ yields the stated expression for $\Delta \Util_i$. Since $A_i>0$, the sign of $\Delta \Util_i$ is determined by the sign of the bracketed sum.

For consumers $\H$ and $\L$, we can compare their changes in utility using the following expression:
$$
\Delta \Util_{\H} - \Delta \Util_{\L}
= \sum_{t\in\mathcal T}
\left[
A_{\H}\frac{d_{\H t}^{\F}}{1+\epsilon_{\H t}}
\Big(1-(\tfrac{\pi_t^{\V}}{\pi^{\F}})^{1+\epsilon_{\H t}}\Big)
-
A_{\L}\frac{d_{lt}^{\F}}{1+\epsilon_{lt}}
\Big(1-(\tfrac{\pi_t^{\V}}{\pi^{\F}})^{1+\epsilon_{lt}}\Big)
\right].
$$

Hence, the low-income consumer loses more utility ($\Delta \Util_{\L}<\Delta \Util_{\H}$) iff the right-hand side is positive.

\paragraph{Declaration of generative AI and AI-assisted technologies in the writing process:}
During the preparation of this work the author(s) used generative AI in order to edit this paper. After using this tool/service, the author(s) reviewed and edited the content as needed and take(s) full responsibility for the content of the publication.

\bibliographystyle{elsarticle-num-names} 
\bibliography{References}

\begin{thebibliography}{33}
\expandafter\ifx\csname natexlab\endcsname\relax\def\natexlab#1{#1}\fi
\providecommand{\url}[1]{\texttt{#1}}
\providecommand{\href}[2]{#2}
\providecommand{\path}[1]{#1}
\providecommand{\DOIprefix}{doi:}
\providecommand{\ArXivprefix}{arXiv:}
\providecommand{\URLprefix}{URL: }
\providecommand{\Pubmedprefix}{pmid:}
\providecommand{\doi}[1]{\href{http://dx.doi.org/#1}{\path{#1}}}
\providecommand{\Pubmed}[1]{\href{pmid:#1}{\path{#1}}}
\providecommand{\bibinfo}[2]{#2}
\ifx\xfnm\relax \def\xfnm[#1]{\unskip,\space#1}\fi
\bibitem[{Borenstein and Bushnell(2022)}]{borenstein2022two}
\bibinfo{author}{S.~Borenstein}, \bibinfo{author}{J.~B. Bushnell},
\newblock \bibinfo{title}{Do two electricity pricing wrongs make a right? cost recovery, externalities, and efficiency},
\newblock \bibinfo{journal}{American Economic Journal: Economic Policy} \bibinfo{volume}{14} (\bibinfo{year}{2022}) \bibinfo{pages}{80--110}.
\bibitem[{{U.S. Energy Information Administration}(2020)}]{eia2020electricity}
\bibinfo{author}{{U.S. Energy Information Administration}},
\newblock \bibinfo{title}{Hourly electricity consumption varies throughout the day and across seasons}  (\bibinfo{year}{2020}). \URLprefix \url{https://www.eia.gov/todayinenergy/detail.php?id=42915}.
\bibitem[{Auffhammer et~al.(2017)Auffhammer, Baylis, and Hausman}]{Auffhammer_Baylis_Hausman_2017}
\bibinfo{author}{M.~Auffhammer}, \bibinfo{author}{P.~Baylis}, \bibinfo{author}{C.~H. Hausman},
\newblock \bibinfo{title}{Climate change is projected to have severe impacts on the frequency and intensity of peak electricity demand across the united states},
\newblock \bibinfo{journal}{Proceedings of the National Academy of Sciences} \bibinfo{volume}{114} (\bibinfo{year}{2017}) \bibinfo{pages}{1886–1891}. \DOIprefix\doi{10.1073/pnas.1613193114}.
\bibitem[{Schittekatte et~al.(2023)Schittekatte, Mallapragada, Joskow, and Schmalensee}]{schittekatte2023reforming}
\bibinfo{author}{T.~Schittekatte}, \bibinfo{author}{D.~Mallapragada}, \bibinfo{author}{P.~L. Joskow}, \bibinfo{author}{R.~Schmalensee},
\newblock \bibinfo{title}{Reforming retail electricity rates to facilitate economy-wide decarbonization},
\newblock \bibinfo{journal}{Joule} \bibinfo{volume}{7} (\bibinfo{year}{2023}) \bibinfo{pages}{831--836}.
\bibitem[{Anderson et~al.(2012)Anderson, White, and Finney}]{Anderson_White_Finney_2012}
\bibinfo{author}{W.~Anderson}, \bibinfo{author}{V.~White}, \bibinfo{author}{A.~Finney},
\newblock \bibinfo{title}{Coping with low incomes and cold homes},
\newblock \bibinfo{journal}{Energy Policy} \bibinfo{volume}{49} (\bibinfo{year}{2012}) \bibinfo{pages}{40–52}. \DOIprefix\doi{10.1016/j.enpol.2012.01.002}.
\bibitem[{Cong et~al.(2022)Cong, Nock, Qiu, and Xing}]{Cong_Nock_Qiu_Xing_2022}
\bibinfo{author}{S.~Cong}, \bibinfo{author}{D.~Nock}, \bibinfo{author}{Y.~L. Qiu}, \bibinfo{author}{B.~Xing},
\newblock \bibinfo{title}{Unveiling hidden energy poverty using the energy equity gap},
\newblock \bibinfo{journal}{Nature Communications} \bibinfo{volume}{13} (\bibinfo{year}{2022}) \bibinfo{pages}{2456}. \DOIprefix\doi{10.1038/s41467-022-30146-5}.
\bibitem[{Kwon et~al.(2023)Kwon, Cong, Nock, Huang, Qiu, and Xing}]{Kwon_2023}
\bibinfo{author}{M.~Kwon}, \bibinfo{author}{S.~Cong}, \bibinfo{author}{D.~Nock}, \bibinfo{author}{L.~Huang}, \bibinfo{author}{Y.~L. Qiu}, \bibinfo{author}{B.~Xing},
\newblock \bibinfo{title}{Forgone summertime comfort as a function of avoided electricity use},
\newblock \bibinfo{journal}{Energy Policy} \bibinfo{volume}{183} (\bibinfo{year}{2023}) \bibinfo{pages}{113813}. \DOIprefix\doi{10.1016/j.enpol.2023.113813}.
\bibitem[{Best and Sinha(2021)}]{Best_Sinha_2021}
\bibinfo{author}{R.~Best}, \bibinfo{author}{K.~Sinha},
\newblock \bibinfo{title}{Fuel poverty policy: Go big or go home insulation},
\newblock \bibinfo{journal}{Energy Economics} \bibinfo{volume}{97} (\bibinfo{year}{2021}) \bibinfo{pages}{105195}. \DOIprefix\doi{10.1016/j.eneco.2021.105195}.
\bibitem[{Brännlund and Vesterberg(2021)}]{Brannlund_Vesterberg_2021}
\bibinfo{author}{R.~Brännlund}, \bibinfo{author}{M.~Vesterberg},
\newblock \bibinfo{title}{Peak and off-peak demand for electricity: Is there a potential for load shifting?},
\newblock \bibinfo{journal}{Energy Economics} \bibinfo{volume}{102} (\bibinfo{year}{2021}) \bibinfo{pages}{105466}. \DOIprefix\doi{10.1016/j.eneco.2021.105466}.
\bibitem[{Mohring(1970)}]{Mohring_1970}
\bibinfo{author}{H.~Mohring},
\newblock \bibinfo{title}{The peak load problem with increasing returns and pricing constraints},
\newblock \bibinfo{journal}{The American Economic Review} \bibinfo{volume}{60} (\bibinfo{year}{1970}) \bibinfo{pages}{693–705}.
\bibitem[{Feldstein(1972)}]{Feldstein_1972}
\bibinfo{author}{M.~S. Feldstein},
\newblock \bibinfo{title}{Distributional equity and the optimal structure of public prices},
\newblock \bibinfo{journal}{The American Economic Review} \bibinfo{volume}{62} (\bibinfo{year}{1972}) \bibinfo{pages}{32–36}.
\bibitem[{Borenstein(2005)}]{Borenstein_longrun}
\bibinfo{author}{S.~Borenstein},
\newblock \bibinfo{title}{The long-run efficiency of real-time electricity pricing},
\newblock \bibinfo{journal}{The Energy Journal} \bibinfo{volume}{26} (\bibinfo{year}{2005}) \bibinfo{pages}{93–116}. \DOIprefix\doi{10.5547/ISSN0195-6574-EJ-Vol26-No3-5}.
\bibitem[{Aigner(1985)}]{aigner1985residential}
\bibinfo{author}{D.~Aigner},
\newblock \bibinfo{title}{The residential electricity time-of-use pricing experiments: what have we learned?},
\newblock in: \bibinfo{booktitle}{Social experimentation}, \bibinfo{publisher}{University of Chicago Press}, \bibinfo{year}{1985}, pp. \bibinfo{pages}{11--54}.
\bibitem[{Caves et~al.(1984)Caves, Christensen, and Herriges}]{caves1984consistency}
\bibinfo{author}{D.~W. Caves}, \bibinfo{author}{L.~R. Christensen}, \bibinfo{author}{J.~A. Herriges},
\newblock \bibinfo{title}{Consistency of residential customer response in time-of-use electricity pricing experiments},
\newblock \bibinfo{journal}{Journal of Econometrics} \bibinfo{volume}{26} (\bibinfo{year}{1984}) \bibinfo{pages}{179--203}.
\bibitem[{Faruqui and Sergici(2010)}]{Faruqui_Sergici_2010}
\bibinfo{author}{A.~Faruqui}, \bibinfo{author}{S.~Sergici},
\newblock \bibinfo{title}{Household response to dynamic pricing of electricity: a survey of 15 experiments},
\newblock \bibinfo{journal}{Journal of Regulatory Economics} \bibinfo{volume}{38} (\bibinfo{year}{2010}) \bibinfo{pages}{193–225}. \DOIprefix\doi{10.1007/s11149-010-9127-y}.
\bibitem[{Horowitz and Lave(2014)}]{Horowitz_Lave_2014}
\bibinfo{author}{S.~Horowitz}, \bibinfo{author}{L.~Lave},
\newblock \bibinfo{title}{Equity in residential electricity pricing},
\newblock \bibinfo{journal}{The Energy Journal} \bibinfo{volume}{35} (\bibinfo{year}{2014}). \URLprefix \url{http://www.iaee.org/en/publications/ejarticle.aspx?id=2556}. \DOIprefix\doi{10.5547/01956574.35.2.1}.
\bibitem[{Simshauser and Downer(2016)}]{simshauser2016inequity}
\bibinfo{author}{P.~Simshauser}, \bibinfo{author}{D.~Downer},
\newblock \bibinfo{title}{On the inequity of flat-rate electricity tariffs},
\newblock \bibinfo{journal}{The Energy Journal}  (\bibinfo{year}{2016}) \bibinfo{pages}{199--229}.
\bibitem[{Burger et~al.(2020)Burger, Knittel, P{\'e}rez-Arriaga, Schneider, and Vom~Scheidt}]{burger2020efficiency}
\bibinfo{author}{S.~P. Burger}, \bibinfo{author}{C.~R. Knittel}, \bibinfo{author}{I.~J. P{\'e}rez-Arriaga}, \bibinfo{author}{I.~Schneider}, \bibinfo{author}{F.~Vom~Scheidt},
\newblock \bibinfo{title}{The efficiency and distributional effects of alternative residential electricity rate designs},
\newblock \bibinfo{journal}{The Energy Journal} \bibinfo{volume}{41} (\bibinfo{year}{2020}).
\bibitem[{Jiang and Low(2011)}]{jiang2011multi}
\bibinfo{author}{L.~Jiang}, \bibinfo{author}{S.~Low},
\newblock \bibinfo{title}{Multi-period optimal energy procurement and demand response in smart grid with uncertain supply},
\newblock in: \bibinfo{booktitle}{2011 50th IEEE conference on decision and control and European control conference}, \bibinfo{organization}{IEEE}, \bibinfo{year}{2011}, pp. \bibinfo{pages}{4348--4353}.
\bibitem[{Jordehi(2019)}]{Jordehi_optDR_review}
\bibinfo{author}{A.~R. Jordehi},
\newblock \bibinfo{title}{Optimisation of demand response in electric power systems, a review},
\newblock \bibinfo{journal}{Renewable and sustainable energy reviews} \bibinfo{volume}{103} (\bibinfo{year}{2019}) \bibinfo{pages}{308--319}.
\bibitem[{Deng et~al.(2015)Deng, Yang, Chow, and Chen}]{deng2015survey}
\bibinfo{author}{R.~Deng}, \bibinfo{author}{Z.~Yang}, \bibinfo{author}{M.-Y. Chow}, \bibinfo{author}{J.~Chen},
\newblock \bibinfo{title}{A survey on demand response in smart grids: Mathematical models and approaches},
\newblock \bibinfo{journal}{IEEE Transactions on Industrial Informatics} \bibinfo{volume}{11} (\bibinfo{year}{2015}) \bibinfo{pages}{570--582}.
\bibitem[{Fell et~al.(2014)Fell, Li, and Paul}]{fell2014new}
\bibinfo{author}{H.~Fell}, \bibinfo{author}{S.~Li}, \bibinfo{author}{A.~Paul},
\newblock \bibinfo{title}{A new look at residential electricity demand using household expenditure data},
\newblock \bibinfo{journal}{International Journal of Industrial Organization} \bibinfo{volume}{33} (\bibinfo{year}{2014}) \bibinfo{pages}{37--47}.
\bibitem[{Borenstein(2012)}]{borenstein2012redistributional}
\bibinfo{author}{S.~Borenstein},
\newblock \bibinfo{title}{The redistributional impact of nonlinear electricity pricing},
\newblock \bibinfo{journal}{American Economic Journal: Economic Policy} \bibinfo{volume}{4} (\bibinfo{year}{2012}) \bibinfo{pages}{56--90}.
\bibitem[{Joskow and Wolfram(2012)}]{joskow2012dynamic}
\bibinfo{author}{P.~L. Joskow}, \bibinfo{author}{C.~D. Wolfram},
\newblock \bibinfo{title}{Dynamic pricing of electricity},
\newblock \bibinfo{journal}{American Economic Review} \bibinfo{volume}{102} (\bibinfo{year}{2012}) \bibinfo{pages}{381--385}.
\bibitem[{Joskow and Tirole(2006)}]{Joskow_Tirole_2006}
\bibinfo{author}{P.~Joskow}, \bibinfo{author}{J.~Tirole},
\newblock \bibinfo{title}{Retail electricity competition},
\newblock \bibinfo{journal}{The RAND Journal of Economics} \bibinfo{volume}{37} (\bibinfo{year}{2006}) \bibinfo{pages}{799–815}. \DOIprefix\doi{10.1111/j.1756-2171.2006.tb00058.x}.
\bibitem[{Borenstein(2007)}]{Borenstein_billvolatility}
\bibinfo{author}{S.~Borenstein},
\newblock \bibinfo{title}{Customer risk from real-time retail electricity pricing: Bill volatility and hedgability},
\newblock \bibinfo{journal}{The Energy Journal} \bibinfo{volume}{28} (\bibinfo{year}{2007}) \bibinfo{pages}{111--130}.
\bibitem[{Shi et~al.(2024)Shi, Wang, Qiu, Deng, Xie, Zhang, and Ma}]{Shi_Wang_Qiu_Deng_Xie_Zhang_Ma_2024}
\bibinfo{author}{H.~Shi}, \bibinfo{author}{B.~Wang}, \bibinfo{author}{Y.~L. Qiu}, \bibinfo{author}{N.~Deng}, \bibinfo{author}{B.~Xie}, \bibinfo{author}{B.~Zhang}, \bibinfo{author}{S.~Ma},
\newblock \bibinfo{title}{The unequal impacts of extremely high temperatures on households’ adaptive behaviors: Empirical evidence from fine-grained electricity consumption data},
\newblock \bibinfo{journal}{Energy Policy} \bibinfo{volume}{190} (\bibinfo{year}{2024}) \bibinfo{pages}{114170}. \DOIprefix\doi{10.1016/j.enpol.2024.114170}.
\bibitem[{Joskow and Tirole(2007)}]{Joskow_Tirole_2007}
\bibinfo{author}{P.~Joskow}, \bibinfo{author}{J.~Tirole},
\newblock \bibinfo{title}{Reliability and competitive electricity markets},
\newblock \bibinfo{journal}{The RAND Journal of Economics} \bibinfo{volume}{38} (\bibinfo{year}{2007}) \bibinfo{pages}{60–84}. \DOIprefix\doi{10.1111/j.1756-2171.2007.tb00044.x}.
\bibitem[{Chao(2010)}]{Chao_2010}
\bibinfo{author}{H.-p. Chao},
\newblock \bibinfo{title}{Price-responsive demand management for a smart grid world},
\newblock \bibinfo{journal}{The Electricity Journal} \bibinfo{volume}{23} (\bibinfo{year}{2010}) \bibinfo{pages}{7–20}. \DOIprefix\doi{10.1016/j.tej.2009.12.007}.
\bibitem[{Calver and Simcock(2021)}]{calver2021demand}
\bibinfo{author}{P.~Calver}, \bibinfo{author}{N.~Simcock},
\newblock \bibinfo{title}{Demand response and energy justice: A critical overview of ethical risks and opportunities within digital, decentralised, and decarbonised futures},
\newblock \bibinfo{journal}{Energy Policy} \bibinfo{volume}{151} (\bibinfo{year}{2021}) \bibinfo{pages}{112198}.
\bibitem[{Borenstein(2008)}]{Borenstein_equityblock}
\bibinfo{author}{S.~Borenstein},
\newblock \bibinfo{title}{Equity effects of increasing-block electricity pricing}  (\bibinfo{year}{2008}). \URLprefix \url{https://escholarship.org/uc/item/3sr1h8nc}.
\bibitem[{Gambardella and Pahle(2018)}]{gambardella2018time}
\bibinfo{author}{C.~Gambardella}, \bibinfo{author}{M.~Pahle},
\newblock \bibinfo{title}{Time-varying electricity pricing and consumer heterogeneity: Welfare and distributional effects with variable renewable supply},
\newblock \bibinfo{journal}{Energy Economics} \bibinfo{volume}{76} (\bibinfo{year}{2018}) \bibinfo{pages}{257--273}.
\bibitem[{Boyd and Vandenberghe(2004)}]{boyd_convex_2004}
\bibinfo{author}{S.~P. Boyd}, \bibinfo{author}{L.~Vandenberghe}, \bibinfo{title}{Convex optimization}, \bibinfo{publisher}{Cambridge univerisity press}, \bibinfo{year}{2004}.

\end{thebibliography}

\end{document}